\documentclass[11pt]{article}

\makeatletter
\renewcommand{\thefootnote}{}

\usepackage{epsfig,amsmath,graphicx,amssymb,overpic}

\usepackage{setspace}
\usepackage{graphicx}
\usepackage{dcolumn}
\usepackage{bm}
\usepackage{graphicx}
\usepackage{subfigure}
\usepackage{epsfig,amsmath,graphicx,amssymb,overpic,cite}

\usepackage{graphicx}
\usepackage{dcolumn}
\usepackage{bm}
\usepackage{graphicx}
\usepackage{subfigure}
\usepackage{amsthm}

\usepackage{pdfsync}
\usepackage{cite}
\usepackage{amscd}
\usepackage{amsmath}
\usepackage{latexsym}
\usepackage{amsfonts}
\usepackage{amssymb}
\usepackage{amsthm}
\usepackage{graphicx}
\usepackage{indentfirst}
\usepackage{enumerate}
\usepackage{amstext}
\usepackage[dvipdfm,
            pdfstartview=FitH,
            CJKbookmarks=true,
            bookmarksnumbered=true,
            bookmarksopen=true,
            colorlinks=true,
            pdfborder=001,
            citecolor=magenta,
            linkcolor=blue,
            ]{hyperref}

 \allowdisplaybreaks[4]

\newtheorem{remark}{Remark}

\usepackage{tikz}
\usetikzlibrary{shapes,arrows,positioning}
\usepackage{xcolor}

\newtheorem{theorem}{Theorem}

\newtheorem{proposition}{Proposition}

\newtheorem{RH}{Riemann-Hilbert Problem}

\def\be{\begin{equation}}
\def\ee{\end{equation}}
\def\bee{\begin{eqnarray}}
\def\ene{\end{eqnarray}}
\def\bes{\begin{subequations}}
\def\ees{\end{subequations}}

\def\d{\displaystyle}
\def\v{\vspace{0.05in}}

\begin{document}

\baselineskip=13pt
\renewcommand {\thefootnote}{\dag}
\renewcommand {\thefootnote}{\ddag}
\renewcommand {\thefootnote}{ }

\pagestyle{plain}

\begin{center}
\baselineskip=16pt \leftline{} \vspace{-.3in} {\Large \bf On Large-Space and Long-Time Asymptotic Behaviors of Kink-Soliton Gases in the Sine-Gordon Equation} \\[0.2in]
\end{center}

\begin{center} \small
{\bf Guoqiang Zhang$^{1}$, Weifang Weng$^2$, and Zhenya Yan}$^{1,3,*}$\footnote{$^{*}${\it Email address}: zyyan@mmrc.iss.ac.cn (Corresponding author)}  \\[0.1in]
{$^1${\footnotesize KLMM,  Academy of Mathematics and Systems Science,  Chinese Academy of Sciences, Beijing 100190, China}}\\
{$^2${\footnotesize School of Mathematical Sciences,University of Electronic Science and Technology of China, Chengdu 611731, China}}\\
{$^3${\footnotesize School of Mathematical Sciences, University of Chinese Academy of Sciences, Beijing 100049, China}} \\[0.18in]
\end{center}

\vspace{0.1in}
\baselineskip=14pt
\noindent {\bf Abstract.}\, {\small
In this paper, we conduct a comprehensive analysis of the large-space and long-time asymptotics of kink-soliton gases in the sine-Gordon equation, addressing an important open problem highlighted in the recent work [Phys. Rev. E 109 (2024) 061001]. 
We focus on kink-soliton gases modeled within a Riemann-Hilbert framework and characterized by two types of generalized reflection coefficients, each defined on the interval \(\left[\eta_1, \eta_2\right]\):
\( r_0(\lambda) = (\lambda - \eta_1)^{\beta_1} (\eta_2 - \lambda)^{\beta_2} |\lambda - \eta_0|^{\beta_0} \gamma(\lambda) \) and \( r_c(\lambda) = (\lambda - \eta_1)^{\beta_1} (\eta_2 - \lambda)^{\beta_2} \chi_c(\lambda) \gamma(\lambda) \),
where \(0 < \eta_1 < \eta_0 < \eta_2\) and \(\beta_j > -1\) for \(j = 0, 1, 2\). Here, \(\gamma(\lambda)\) is a continuous, strictly positive function defined on \([\eta_1, \eta_2]\), which extends analytically beyond this interval. The function \(\chi_c(\lambda)\) demonstrates a step-like behavior: it is given by \(\chi_c(\lambda) = 1\) for \(\lambda \in [\eta_1, \eta_0)\) and \(\chi_c(\lambda) = c^2\) for \(\lambda \in (\eta_0, \eta_2]\), with \(c\) as a positive constant distinct from one.
To rigorously derive the asymptotic results, we leverage the steepest descent method as developed by Deift and Zhou. A central component of this approach is constructing an appropriate \(g\)-function for the conjugation process. Unlike in the Korteweg-de Vries equation, the sine-Gordon equation presents unique challenges for \(g\)-function formulation, particularly concerning the singularity at the origin. The Riemann-Hilbert problem also requires carefully constructed local parametrices near endpoints \(\eta_j\) (\(j = 1, 2\)) and the singularity \(\eta_0\). At the endpoints \(\eta_j\), we employ a modified Bessel parametrix of the first kind. For the singularity \(\eta_0\), the parametrix selection depends on the reflection coefficient: the second kind of modified Bessel parametrix is used for \(r_0(\lambda)\), while a confluent hypergeometric parametrix is applied for \(r_c(\lambda)\).
}

\v\v \noindent {\bf Keyword.}\, Sine-Gordon equation,\,Kink-soliton gases,\, Inverse scattering transform,\, Riemann-Hilbert problem, \, Asymptotic behaviors

\v\v \noindent {\bf MSC code.}\, 35Q51, 37K40, 35Q15, 37K10, 37K15




\begin{spacing}{1.1}
\tableofcontents
\end{spacing}


\baselineskip=14pt


\section{Introduction and Main results}

\subsection{Backgrounds}

In this paper, we would like to study the rigorous large-space and long-time asymptotic behaviors for kink-soliton gases (i.e., the limit behavior of $N$-kink-soliton solution as $N\to \infty$) of the (1+1)-dimensional sine-Gordon (sG) equation in light-cone coordinates~\cite{sg,sg-book,bar}
\begin{gather} \label{sine-Gordon}
u_{xt}=\sin u, \qquad (x,t)\in\mathbb{R}\times \mathbb{R}^+,
\end{gather}
which arises from one of the most significant open problems in the review paper~\cite{22}, where $u=u(x,t)$ is a real-valued function, and the subscripts denote the partial derivatives with respect to variables. The sG equation \eqref{sine-Gordon} has another alternative form
\bee \label{sg2}
u_{\tau\tau}-u_{\xi\xi}+\sin u=0.
\ene
Note that Eq.~(\ref{sg2}) can be reduce to Eq.~(\ref{sine-Gordon}) via the simple transforms $u(\xi, \tau)\to u(x,t),\, \xi=x+t,\, \tau=x-t$. The sG equation was originally introduced by Bour~\cite{sg1862} in 1862, which  is a hyperbolic nonlinear wave equation involving the D'Alembert operator and nonlinear sine function, and used to describe two-dimensional  constant negative curvature surfaces in differential geometry. The sG equation appears in many physical settings, such as superconducting Josephson junctions~\cite{jose},
the self-induced transparency in nonlinear optics~\cite{sit},  crystal dislocations~\cite{cd}, DNA dynamics~\cite{dna1,dna2} and quantum field theory~\cite{qc}. The sG equation has been shown to possess many types of solutions, such as kink solitons~\cite{23,ks,ks1}, breathers~\cite{bs1,bs2}, multi-pole solitons~\cite{multi-s}, multi elliptic-localized solutions~\cite{ling23}, and etc. (see Refs.~\cite{sg-solu10,lamb,sg-book} and references therein).
By using the Riemann-Hilbert technique, man aspects have been studied, such as the Cauchy problem in the semiclassical limit~\cite{Buckingham2012, Buckingham2013},
construction of solutions and asymptotics in the quarter plane~\cite{Huang2018, Huang2018a}.

The sG equation \eqref{sine-Gordon} is a completely integrable system, solved by the inverse scattering transform \cite{23,kaup}, based on its Lax pair
\bee \label{lax}
\left\{\begin{array}{l}
\psi_x=U(x,t;\lambda)\psi, \quad U=\lambda \sigma_3-\dfrac{\mathrm{i}}{2}u_x\sigma_2, \v\\
\psi_t=V(x,t;\lambda)\psi,\quad V=\dfrac{1}{4\lambda}\left(\sigma_1 \sin u+\sigma_3\cos u\right),
 \end{array}\right.
 \ene
where $\psi=\psi(x,t;\lambda)$ is the eigenfunction, and $\lambda$ denotes the isopsectral parameter, and three Pauli matrices are
\begin{gather}\label{sigma}
\sigma_1 =
\begin{pmatrix}
0 & 1 \\
1 & 0
\end{pmatrix}, \quad
\sigma_2 =
\begin{pmatrix}
0 & -\mathrm{i} \\
\mathrm{i} & 0
\end{pmatrix}, \quad
\sigma_3 =
\begin{pmatrix}
1 & 0 \\
0 & -1
\end{pmatrix}.
\end{gather} The compatible condition of the Lax pair, $U_t-V_x+[U, V]=0$, just generates the sG equation \eqref{sine-Gordon}. The numerical IST was studied for the sG equation~\cite{sg-num}. Based on the IST, some initial-value problems of the sG equation were also explored~\cite{kaup,zak,tal,buck}. Moreover, some long-time asymptotics were analyzed via the Riemann-Hilbert problems, such as, the long-time asymptotic behavior was studied for the pure radiation (i.e., solitonless) solution of the sG equation with the Schwartz class of initial data~\cite{zhou19}; the asymptotic properties were studied for the Cauchy problem of the semiclassical sG equation~\cite{miller22}; and the long-time asymptotics and stability were also investigated for the sG equation with  weighted Sobolev initial data~\cite{chen22}.

Since the discovery of multi-soliton solutions in integrable nonlinear systems via the inverse scattering transform (IST)~\cite{soliton-book91, GGKM}, the investigation of multi-soliton interactions has become fundamental in soliton theory and its applications. In 1971, Zakharov~\cite{1} introduced the notion of a soliton gas, conceptualized as an infinite collection of weakly interacting solitons governed by the Korteweg-de Vries (KdV) equation. Subsequent work expanded Zakharov’s model from a sparse gas to a dense KdV soliton gas~\cite{2}, utilizing spectral theory through the thermodynamic limit of finite-gap solutions. This spectral approach was later applied to describe not only soliton gases but also breather gases for the focusing nonlinear Schrödinger (NLS) equation~\cite{3,4}, as well as bidirectional soliton gases in dispersive hydrodynamic systems within the defocusing NLS framework~\cite{5}. Beyond the realm of spectral theory, the mathematical properties of soliton gases have also been extensively studied, encompassing integrable reductions~\cite{6}, hydrodynamic reductions~\cite{7}, minimal energy configurations~\cite{8}, and aspects of classical integrability in hydrodynamics~\cite{9}. In the context of numerical simulations, it has been suggested that phenomena such as the nonlinear phase of spontaneous modulational instability~\cite{10} and the emergence of rogue waves~\cite{11,12} can be fundamentally linked to soliton gas dynamics.

In recent years, the asymptotic analysis of soliton gases has received substantial attention, particularly through the use of Riemann-Hilbert techniques. The soliton gas for the KdV equation, known as the primitive potential, is constructed using both the dressing method~\cite{13} and Riemann-Hilbert problem formulations~\cite{14}, where two reflection coefficients contribute to the jump conditions. For cases involving a single reflection coefficient, a comprehensive asymptotic analysis has been conducted~\cite{15} employing the steepest descent methodology developed by Deift and Zhou~\cite{16} and further elaborated in subsequent works~\cite{17,18,19}. Recently, this asymptotic framework has been adapted to the modified KdV equation~\cite{20}, facilitating the study of interactions between soliton gases and large individual solitons, and enabling a detailed description of various wave properties. These include the local phase shift within the soliton gas, the positioning of soliton peaks, and notably, the average soliton peak velocity as predicted by a kinetic equation. Beyond the discrete spectra confined to segments of the real axis \(\mathbb{R}\) or the imaginary axis \(\mathrm{i}\mathbb{R}\), soliton gases derived from the \(N\)-soliton solution in the limit \(N \to \infty\) have been extended to the focusing NLS equation with bounded domain discrete spectra, unveiling a remarkable soliton shielding effect~\cite{21}.

A recent survey~\cite{22} identified numerous unresolved questions surrounding soliton gases, with the rigorous asymptotic analysis of sine-Gordon kink-soliton gases highlighted as a particularly critical problem. Motivated by significant advancements in this field, as discussed in~\cite{15}, the current study aims to tackle this challenge. Starting from the Lax pair formulation~(\ref{lax}), we derive an \(N\)-kink solution sequence for the sine-Gordon equation~\eqref{sine-Gordon} via an associated Riemann-Hilbert problem using the standard IST framework~\cite{23, kaup}.

\begin{RH}\label{RH-1} Let $M^N$ represent  the solution to the Riemann-Hilbert problem defined as follows, adhering to a set of specific conditions:
\begin{itemize}

\item{} The function $M^N$ is analytic throughout the complex plane, with exceptions at  points $\pm\lambda_j$,  where $j=1, 2, \cdots, N$. These points  $\pm\lambda_j$ denote discrete spectral points constrained by $\lambda_j\in(\eta_1, \eta_2)$;

\item{} The function $M^N$ satisfies the normalization condition such that $M^N\to\mathbb{I}_2$ as $\lambda\to \infty$, where $\mathbb{I}_2$ denotes the $2\times 2$ identity matrix.

\item{} Each discrete spectral point $\pm\lambda_j$ is a simple pole of $M^N$, with residue conditions given by
\begin{gather}
\begin{aligned}
\mathop\mathrm{Res}\limits_{\lambda=\lambda_j}M^N&=\lim_{\lambda\to\lambda_j}M^N\mathcal{L}^{t\theta}\left[c_j\right],
\v\\
\mathop\mathrm{Res}\limits_{\lambda=-\lambda_j}M^N&=\lim_{\lambda\to -\lambda_j}M^N\mathcal{U}^{t\theta}\left[c_j\right],
\end{aligned}
\end{gather}
where $c_j < 0$ for all $j$, and the operators $\mathcal{L}^{t\theta}\left[c_j\right]$ and $\mathcal{U}^{t\theta}\left[c_j\right]$ are defined in Equation \eqref{LU}. The phase factor $\theta$ is given by
\begin{gather}
\theta=\frac{\xi\lambda+\lambda^{-1}}{4}, \qquad \xi=\frac{4x}{t}.
\end{gather}
\end{itemize}
\end{RH}

Utilizing the solution $M^N(x,t;\lambda)$ from the Riemann-Hilbert problem~\ref{RH-1}, the $N$-kink solutions for the sine-Gordon equation \eqref{sine-Gordon}, denoted by $u_N=u_N(x, t)$, can be explicitly derived through the following expressions:
\begin{gather}\label{uN-potential}
\frac{\partial}{\partial x}u_N =4 \lim_{\lambda\to\infty} \lambda M^N_{1, 2}, \quad
\cos u_N=1-2M^N_{1, 2}(0)^2, \quad
\sin u_N=-2M^N_{1, 1}(0)M^N_{1, 2}(0).
\end{gather}

In the context of the KdV equation~\cite{15}, the reflection coefficient is assumed to satisfies the following criteria: 1) it remains continuous and strictly positive for $\lambda \in [\eta_1, \eta_2]$; 2) it extends analytically to a neighborhood around this interval; and 3) it takes the same values over both $[\eta_1, \eta_2]$ and $[-\eta_2, -\eta_1]$. These conditions ensure that the solution $Y$ to the Riemann-Hilbert problem exhibits local logarithmic singularities at the interval endpoints.
Consequently, the parametrices near these endpoints can be described using modified Bessel functions of index zero. Additionally, when the reflection coefficient exhibits the behavior $\left|\lambda \mp \eta_j\right|^{\pm 1/2}$ close to $\lambda = \pm \eta_j$, as examined in \cite{20}, local parametrices become unnecessary, since the outer parametrix alone can adequately match the local properties and thereby serves as a global parametrix.
The outer parametrix concept initially arose in the study of long-time asymptotic behavior for the KdV equation with step-like initial conditions~\cite{24}. In \cite{25}, a matrix outer parametrix designed to handle identical jump matrices was introduced for analyzing the asymptotic properties of orthogonal polynomials associated with the Hermitian matrix model. This approach has since been applied to other cases, including the NLS shock problem~\cite{26}, the modified KdV equation under step-like initial conditions~\cite{27}, and the NLS equation with non-zero boundary conditions~\cite{28, 29}.

This paper introduces two novel forms of generalized reflection coefficients applicable to kink-soliton gases in the sine-Gordon (sG) framework, alongside their rigorous asymptotic characteristics:

\begin{itemize}

\item{} The first generalized reflection coefficient is defined as follows:
\begin{gather}\label{r0}
r_0(\lambda) = (\lambda - \eta_1)^{\beta_1} (\eta_2 - \lambda)^{\beta_2} |\lambda - \eta_0|^{\beta_0} \gamma(\lambda),
\end{gather}
where $\gamma(\lambda)$ is continuous, strictly positive for $\lambda \in [\eta_1, \eta_2]$, and is assumed to be analytically extendable within a neighborhood around $[\eta_1, \eta_2]$.
Unlike the standard reflection coefficient described in \cite{15}, which maintains positivity at the interval boundaries, this generalized coefficient $r_0(\lambda)$ introduces zeros and singularities at both $\eta_1$ and $\eta_2$. Moreover, it features additional zeros and singularities at an internal point $\eta_0$ within the interval $(\eta_1, \eta_2)$, thus diverging significantly from the original reflection coefficient. The absolute value ensures that $r_0(\lambda)$ remains positive over the subintervals $(\eta_1, \eta_0)$ and $(\eta_0, \eta_2)$.

\item{} The second generalized reflection coefficient is expressed as:
\begin{gather}\label{rc}
r_c(\lambda) = (\lambda - \eta_1)^{\beta_1} (\eta_2 - \lambda)^{\beta_2} \chi_c(\lambda) \gamma(\lambda).
\end{gather}
This coefficient retains the endpoint behavior of $r_0(\lambda)$ and includes a singularity at $\eta_0$. In contrast to $r_0(\lambda)$, however, $r_c(\lambda)$ is characterized by a jump discontinuity at $\eta_0$, introduced by the function $\chi_c$. Specifically, $\chi_c(\lambda)$ is defined as a piecewise function: $\chi_c(\lambda) = 1$ for $\lambda \in [\eta_1, \eta_0)$ and $\chi_c(\lambda) = c^2$ for $\lambda \in (\eta_0, \eta_2]$, where $c$ is a positive constant distinct from one ($c \neq 1$).
The exponents $\beta_j$ satisfy the conditions $\beta_1, \beta_2, \beta_0 > -1$, with the motivation for these restrictions to be elucidated in subsequent sections.

\end{itemize}

By modifying and employing an interpolation technique that converts poles into jump conditions, as utilized by Deift et al. \cite{30} in their analysis of the Toda rarefaction problem, it becomes possible to derive a specific class of $N$-soliton solutions $u_N$, where $N = N_1 + N_2$, as expressed in \eqref{uN-potential}. For convenience, we will continue to denote the solution to the modified Riemann-Hilbert problem by $M^N$.

\begin{RH}\label{RH-2} The function $M^N$, which is a $2 \times 2$ matrix, satisfies the following conditions:

\begin{itemize}
\item{} $M^N$ is analytic for all $\lambda \in \mathbb{C} \setminus (\Gamma_+ \cup \Gamma_-)$, where $\Gamma_\pm$ are two distinct, simple, closed curves encircling $[\eta_1, \eta_2]$ and $[-\eta_2, -\eta_1]$, respectively, with counter-clockwise orientation;

\item{} At $\lambda \to \infty$, the function $M^N$ satisfies the normalization condition:
$M^N\to \mathbb{I}_2$ as $\lambda\to\infty$;

\item{} Along the contours $\Gamma_\pm$, $M^N$ has well-defined boundary values, denoted $M^N_\pm$, which are related by the jump conditions:
\begin{gather}\label{MN}
M^N_+=M^N_-
\begin{cases}
\d\mathcal{L}^{t\theta}\left[-\sum_{j=1}^{N_1}\dfrac{\left(\eta_0-\eta_1\right)\,r\left(\lambda_{1, j}\right)}{2N_1\pi\,\left(\lambda-\lambda_{1, j}\right)}
-\sum_{j=1}^{N_2}\dfrac{\left(\eta_2-\eta_0\right)\,r\left(\lambda_{2, j}\right)}{2N_2\pi\,\left(\lambda-\lambda_{2, j}\right)}
\right], &\mathrm{for}\,\,\, \lambda\in\Gamma_+,\\[2em]
\d \mathcal{U}^{t\theta}\left[-\sum_{j=1}^{N_1}\dfrac{\left(\eta_0-\eta_1\right)\,r\left(-\lambda_{1, j}\right)}{2N_1\pi\,\left(\lambda+\lambda_{1, j}\right)}
-\sum_{j=1}^{N_2}\dfrac{\left(\eta_2-\eta_0\right)\,r\left(-\lambda_{2, j}\right)}{2N_2\pi\,\left(\lambda+\lambda_{2, j}\right)}\right],  &\mathrm{for}\,\,\, \lambda\in\Gamma_-.
\end{cases}
\end{gather}
\end{itemize}
\end{RH}
In this framework, discrete spectral points $\lambda_{1, j}$ are evenly spaced within the interval $(\eta_1, \eta_0)$, given by $\lambda_{1, j} = \eta_1 + j(\eta_0 - \eta_1) / (N_1 + 1)$, for $j = 1, 2, \dots, N_1$. Likewise, $\lambda_{2, j}$ are situated in $(\eta_0, \eta_2)$, defined by $\lambda_{2, j} = \eta_0 + j(\eta_2 - \eta_0) / (N_2 + 1)$, for $j = 1, 2, \dots, N_2$. These discrete spectral values reside within the broader interval $(\eta_1, \eta_2)$, allowing for singularities of index $\beta_j$ at the endpoints of the reflection coefficient.
Contrasting with \cite{15}, where the midpoint $\eta_0$ does not appear within the spectral distribution, this configuration permits additional local characteristics at $\lambda = \eta_0$, such as a singularity of index $\beta_0$ for $r_0$ and a discontinuity for $r_c$. Both types of reflection coefficients, $r_0$ and $r_c$, are incorporated in \eqref{MN}, with the values of $r$ on $(-\eta_2, -\eta_1)$ determined by the symmetry $r(\lambda) = r(-\lambda)$.

The limit technique serves as a powerful and efficient approach for identifying novel nonlinear wave solutions in integrable systems, particularly those not accessible through direct solution methods. For instance, by examining the limit of a sequence of Riemann-Hilbert problems corresponding to $N$th-order rogue waves, one can capture infinite-order rogue wave solutions, as demonstrated in \cite{31} within the framework of the inverse scattering transform (IST) \cite{32}.
This technique proves effective in cases where traditional methods such as the Darboux transformation \cite{33,34,35}, Hirota’s bilinear method \cite{36,37}, and the IST \cite{23,38} are unable to derive explicit solutions. Similarly, this limit approach has been applied to uncover an infinite-order soliton solution for the focusing NLS equation \cite{39}, as well as an infinite-order rational-soliton solution for the complex modified KdV equation~\cite{weng25}. A key element in this limiting process is the use of a suitable rescaling transform, which is essential for ensuring the convergence of the jump matrices sequence.
Furthermore, this rescaling facilitates the study of large-$N$ asymptotics of $N$-soliton solutions with initial data in the form of $N\, \mathrm{sech}(x)$, rendering it equivalent to a semiclassical limit problem \cite{40}. Unlike the rescaling approach, the method outlined in \cite{15} emphasizes the selection of appropriate norming constants, leading to convergence within the framework of a Riemann integral.

In this study, we adopt norming constants by discretizing the generalized reflection coefficients, specifically $r = r_0$ and $r = r_c$. The convergence process to establish a soliton gas not only approaches definite integrals but also extends to improper integrals. Alternatively, the limit process detailed in \cite{15} can be seen as a semiclassical approximation of the reflection coefficient, as discussed in \cite{41,42}. By taking the limits as $N_1 \to \infty$ and $N_2 \to \infty$, we derive the Riemann-Hilbert problem~\ref{RH-2} associated with a kink-soliton gas.

\begin{RH} \label{RH-3} The solution $M^\infty$, which is a $2 \times 2$ matrix-valued function, satisfies the following conditions:
\begin{itemize}

\item{} $M^\infty$ is analytic for $\lambda \in \mathbb{C} \setminus (\Gamma_+ \cup \Gamma_-)$;

\item{} The normalization condition at infinity is $M^\infty \to \mathbb{I}_2$ as $\lambda \to \infty$;

\item{} The boundary values on the contours relate as follows:
\begin{gather}\label{Minfty}
M^\infty_+=M^\infty_-
\begin{cases}
\mathcal{L}^{t\theta}\left[\mathrm{i}\left(\mathcal{P}_1+\mathcal{P}_2\right)
\right], &\mathrm{for}\,\,\, \lambda\in\Gamma_+,\\[0.5em]
\mathcal{U}^{t\theta}\left[\mathrm{i}\left(\mathcal{P}_{-1}+\mathcal{P}_{-2}\right)\right],  &\mathrm{for}\,\,\, \lambda\in\Gamma_-,
\end{cases}
\end{gather}
where
$\mathcal{P}_1=\int_{\eta_1}^{\eta_0}\frac{r(s)}{s-\lambda}\mathrm{d}s$, $\mathcal{P}_2=\int_{\eta_0}^{\eta_2}\frac{r(s)}{s-\lambda}\mathrm{d}s$, $\mathcal{P}_{-1}=\int_{-\eta_0}^{-\eta_1}\frac{r(s)}{s-\lambda}\mathrm{d}s$, and $\mathcal{P}_{-2}=\int_{-\eta_2}^{-\eta_0}\frac{r(s)}{s-\lambda}\mathrm{d}s$.
Following the approach of Zhou's vanishing lemma \cite{43}, used for the Riemann-Hilbert problem 7 in \cite{20}, a unique solution for $M^\infty$ can be established.
\end{itemize}
\end{RH}
The kink-soliton gas solution $u = u(x, t)$ is then obtained from $M^\infty$ through:
\begin{gather}\label{u-infinity}
\frac{\partial u}{\partial x}=4 \lim_{\lambda\to\infty} \lambda M^\infty_{1, 2}, \quad
\cos u=1-2M^\infty_{1, 2}(0)^2, \quad
\sin u=-2M^\infty_{1, 1}(0)M^\infty_{1, 2}(0).
\end{gather}
For this limiting process, it is essential that the jump matrices, as they transition from \eqref{MN} to \eqref{Minfty}, are consistent. This holds when the exponents $\beta_1$, $\beta_2$, and $\beta_0$ satisfy $\beta_1, \beta_2, \beta_0 > -1$. This validity is considered in two cases.
First, when $\beta_1, \beta_2, \beta_0 \ge 0$, the condition aligns directly with the definition of a definite Riemann integral. In the second case, where any of $\beta_1$, $\beta_2$, or $\beta_0$ lies within $(-1, 0)$, improper integrals emerge due to the singularities, making the convergence of the Riemann sum independent of the definite integral definition. However, convergence towards improper integrals can still be established by utilizing monotonicity and uniform continuity, relying on elementary calculus techniques.

\subsection{The main results}
To ensure that $E$ normalizes to the identity matrix $\mathbb{I}_2$ as $\lambda \to \infty$ and that its associated jump matrices uniformly and exponentially decay to the identity, a sequence of transformations is applied: $Y \mapsto T \mapsto S \mapsto E$. These transformations were initially employed by Deift, Kriecherbauer, McLaughlin, Venakides, and Zhou in their work on the asymptotic behavior of orthogonal polynomials with exponential weights \cite{25,44}, building on the Riemann-Hilbert framework proposed by Fokas, Its, and Kitaev \cite{45,46}.
The Deift-Zhou steepest descent method has since been adapted to a wide range of orthogonal polynomial classes, including those with logarithmic weights \cite{47,48}, Freud weights \cite{49}, and Laguerre polynomials \cite{50,51,52,53}. It has also been extended to measures supported on the complex plane \cite{54}, as well as Jacobi weights \cite{55}, modified Jacobi weights \cite{56,57}, and even discontinuous Gaussian weights \cite{58}, among others. Beyond polynomials, this technique has found applications in various other areas, such as analyzing the distribution of the longest increasing subsequence in random permutations \cite{17} and studying the asymptotic properties of the discrete holomorphic map $Z^a$ \cite{59}.

In this paper, we rigorously establish the asymptotic behavior of sine-Gordon kink-soliton gases, focusing on two generalized reflection coefficients, $r_0$ and $r_c$, across the regions $\left(-\infty, \xi_\mathrm{crit}\right)$, $\left(\xi_\mathrm{crit}, \xi_0\right)$, $\left(\xi_0, -\eta_2^{-2}\right)$, and $\left(-\eta_2^{-2}, +\infty\right)$. The critical points $\xi_\mathrm{crit}$ and $\xi_0$ are defined by the equations:
\begin{gather}
\xi_\mathrm{crit}=-\eta_2^{-2}\,W\left(\frac{\eta_1}{\eta_2}\right), \quad \xi_0=-\eta_2^{-2}\,W\left(\frac{\eta_0}{\eta_2}\right),
\end{gather}
 where $W: m \mapsto W(m)$ is given by
\begin{gather}
W(m)=\frac{eE(m)/eK(m)}{m\left(m^2-1+eE(m)/eK(m)\right)},
\end{gather}
with $eE$ and $eK$ representing the complete elliptic integrals as defined in \eqref{elliptic12}. The function $W$ aligns with a specific $g$-function used within the conjugation process to achieve exponential decay along the lenses, with the Airy parametrix characterizing the local behavior near $\lambda = \alpha$.
The parameter $\alpha$ is determined uniquely through the Whitham modulation equation \cite{60}:
\begin{gather}\label{Whitham}
\xi=-\eta_2^{-2}\, W\left(\frac{\alpha}{\eta_2}\right).
\end{gather}

The primary findings of this paper are summarized as follows.

\begin{theorem}[Large-$x$ asymptotics for the initial value $u(x, 0)$ of the kink-soliton gas]\label{large-x}
For these two types of generalized reflection coefficients $r=r_0$ and $r=r_c$, the large-$x$ asymptotics for the initial value
$u(x, 0)$ of  the kink-soliton gas are established as follows.

\begin{enumerate}[0.]
\item[\rm{\textbullet}]  For $\beta_2, \beta_1, \beta_0\ge 0$, in the limit of $x\to+\infty$, there exists a positive constant $\mu_0$ such that
\begin{gather}\label{large-x-right}
\begin{aligned}
&\frac{\mathrm{d}u(x, 0)}{\mathrm{d}x}=\mathcal{O}\left(\mathrm{e}^{-\mu_0 x}\right),\\[0.5em]
&\cos u(x, 0)=1+\mathcal{O}\left(\mathrm{e}^{-\mu_0 x}\right),\\[0.5em]
&\sin u(x, 0)=\mathcal{O}\left(\mathrm{e}^{-\mu_0 x}\right).
\end{aligned}
\end{gather}
\item[\rm{\textbullet}] For $\beta_2, \beta_1, \beta_0>-1$, in the limit of $x\to-\infty$,  we have
\begin{gather}\label{large-x-left}
\begin{aligned}
&\frac{\mathrm{d}u(x, 0)}{\mathrm{d}x}=
2\left(\eta_1-\eta_2\right)\frac{\vartheta_3\left(\frac{1}{2}+\frac{\Delta^0_1}{2\pi \mathrm{i}}; \tau_1\right)\vartheta_3\left(0; \tau_1\right)}{\vartheta_3\left(\frac{1}{2}; \tau_1\right)\vartheta_3\left(\frac{\Delta^0_1}{2\pi \mathrm{i}}; \tau_1\right)}
+\mathcal{O}\left(\frac{1}{\left|x\right|}\right),                                                             \\[0.5em]
&\cos u(x, 0)=1-8
\left(\mathcal{A}_1\frac{\vartheta_3\left(\frac{1+\tau_1}{2}+\frac{\Delta^0_1}{2\pi \mathrm{i}}; \tau_1\right)\vartheta_3\left(0; \tau_1\right)}{\vartheta_3'\left(\frac{1+\tau_1}{2}; \tau_1\right)\vartheta_3\left(\frac{\Delta^0_1}{2\pi \mathrm{i}}; \tau_1\right)}\right)^2
\mathrm{e}^{\Delta^0_1}
+\mathcal{O}\left(\frac{1}{\left|x\right|}\right),                                  \\[0.5em]
&\sin u(x, 0)=4\mathcal{A}_1
\frac{\vartheta_3\left(\frac{\tau_1}{2}-\frac{\Delta^0_1}{2\pi \mathrm{i}}; \tau_1\right)
\vartheta_3\left(\frac{1+\tau_1}{2}+\frac{\Delta^0_1}{2\pi \mathrm{i}}; \tau_1\right)
\vartheta_3^2\left(0; \tau_1\right)}{\vartheta_3\left(\frac{\tau_1}{2}; \tau_1\right)
\vartheta_3'\left(\frac{1+\tau_1}{2}; \tau_1\right)
\vartheta_3^2\left(\frac{\Delta^0_1}{2\pi \mathrm{i}}; \tau_1\right)}
+\mathcal{O}\left(\frac{1}{\left|x\right|}\right),
\end{aligned}
\end{gather}
where
$\mathcal{A}_1=\left(1-m_1\right)eK\left(m_1\right)$,
$m_1=\eta_1/\eta_2$,
$\tau_1=\mathrm{i}\,eK\left(\sqrt{1-m_1^2}\right)\left/\right.2eK\left(m_1\right)$,
$\Omega_1=-\pi \mathrm{i}\eta_2\left/\right.eK\left(m_1\right)$,
\begin{gather}
\Delta^0_1=\Omega_1(x+\phi_1), \quad
\phi_1=-\int_{\eta_1}^{\eta_2}\frac{\log r(s)}{\sqrt{\left(s^2-\eta_1^2\right)\left(\eta_2^2-s^2\right)}}\frac{\mathrm{d}s}{\pi},
\end{gather}
and $\vartheta_3$ being the Jacobi theta function, defined by
\begin{gather}
\vartheta_3\left(\lambda; \tau\right)=\sum_{n\in\mathbb{Z}}\exp\left\{2\pi\mathrm{i}n\lambda+\pm\mathrm{i}n^2\tau\right\}.
\end{gather}
\end{enumerate}
\end{theorem}

\begin{theorem}[Long-time asymptotics for the kink-soliton gas $u(x, t)$]
For these two types of generalized reflection coefficients $r=r_0$ and $r=r_c$, the long-time asymptotics for the sine-Gordon kink-soliton gas $u(x, t)$ are  established as follows.
\begin{enumerate}[0.]
\item[\rm{\textbullet}] For $\xi>-\eta_2^{-2}$ with $\beta_1, \beta_2, \beta_0\ge 0$, there exists a positive constant $\mu=\mu(\xi)$ such that
\begin{gather}\label{large-right}
\begin{aligned}
&\frac{\mathrm{d}u(x, t)}{\mathrm{d}x}=\mathcal{O}\left(\mathrm{e}^{-\mu t}\right),\\[0.5em]
&\cos u(x, t)=1+\mathcal{O}\left(\mathrm{e}^{-\mu t}\right),\\[0.5em]
&\sin u(x, t)=\mathcal{O}\left(\mathrm{e}^{-\mu t}\right).
\end{aligned}
\end{gather}

\item[\rm{\textbullet}] For $\xi\in\left(\xi_0, -\eta_2^{-2}\right)$ with  $\beta_0, \beta_1\ge 0, \beta_2>-1$,
and for  $\xi\in\left(\xi_{\mathrm{crit}}, \xi_0\right)$ with $\beta_0,\, \beta_2>-1, \beta_1\ge 0$,
\begin{gather}\label{large-middle}
\begin{aligned}
&\frac{\mathrm{d}u(x, t)}{\mathrm{d}x}=
2\left(\alpha-\eta_2\right)\frac{\vartheta_3\left(\frac{1}{2}+\frac{\Delta^\alpha}{2\pi \mathrm{i}}; \tau_\alpha\right)\vartheta_3\left(0; \tau_\alpha\right)}{\vartheta_3\left(\frac{1}{2}; \tau_\alpha\right)\vartheta_3\left(\frac{\Delta^\alpha}{2\pi \mathrm{i}}; \tau_\alpha\right)}
+\mathcal{O}\left(\frac{1}{t}\right),                \\[0.5em]
&\cos u(x, t)=1-8
\left(\mathcal{A}^\alpha\frac{\vartheta_3\left(\frac{1+\tau_\alpha}{2}+\frac{\Delta^\alpha}{2\pi \mathrm{i}}; \tau_\alpha\right)\vartheta_3\left(0; \tau_\alpha\right)}{\vartheta_3'\left(\frac{1+\tau_\alpha}{2}; \tau_\alpha\right)\vartheta_3\left(\frac{\Delta^\alpha}{2\pi \mathrm{i}}; \tau_\alpha\right)}\right)^2
\mathrm{e}^{\Delta^\alpha}
+\mathcal{O}\left(\frac{1}{t}\right),   \\[0.5em]
&\sin u(x, t)=4\mathcal{A}^\alpha
\frac{\vartheta_3\left(\frac{\tau_\alpha}{2}-\frac{\Delta^\alpha}{2\pi \mathrm{i}}; \tau_\alpha\right)
\vartheta_3\left(\frac{1+\tau_\alpha}{2}+\frac{\Delta^\alpha}{2\pi \mathrm{i}}; \tau_\alpha\right)
\vartheta_3^2\left(0; \tau_\alpha\right)}{\vartheta_3\left(\frac{\tau_\alpha}{2}; \tau_\alpha\right)
\vartheta_3'\left(\frac{1+\tau_\alpha}{2}; \tau_\alpha\right)
\vartheta_3^2\left(\frac{\Delta^\alpha}{2\pi \mathrm{i}}; \tau_\alpha\right)}
+\mathcal{O}\left(\frac{1}{t}\right),
\end{aligned}
\end{gather}
where
$\mathcal{A}^\alpha=\left(1-m_\alpha\right)eK\left(m_\alpha\right)$,
$m_\alpha=\alpha/\eta_2$,
$\tau_\alpha=\mathrm{i}\,eK\left(\sqrt{1-m_\alpha^2}\right)\left/\right.2eK\left(m_\alpha\right)$,
$\Omega^\alpha=-\pi \mathrm{i}\eta_2\left/\right.eK\left(m_\alpha\right)$,
\begin{gather}
\Delta^\alpha=\Omega^\alpha\left(x+\frac{t}{4\alpha \eta_2}+\phi^\alpha \right),  \quad
\phi^\alpha=-\int_{\alpha}^{\eta_2}\frac{\log r(s)}{\sqrt{\left(s^2-\alpha^2\right)\left(\eta_2^2-s^2\right)}}\frac{\mathrm{d}s}{\pi}.
\end{gather}
Specially, in the case of $r=r_0$ with $\beta_0=0$, \eqref{large-middle} holds for $\xi\in\left(\xi_\mathrm{crit}, -\eta_2^{-2}\right)$.

\item[\rm{\textbullet}] For $\xi<\xi_\mathrm{crit}$  with $\beta_1, \beta_2, \beta_0>-1$,
\begin{gather}\label{large-left}
\begin{aligned}
&\frac{\mathrm{d}u(x, t)}{\mathrm{d}x}=
2\left(\eta_1-\eta_2\right)\frac{\vartheta_3\left(\frac{1}{2}+\frac{\Delta_1}{2\pi \mathrm{i}}; \tau_1\right)\vartheta_3\left(0; \tau_1\right)}{\vartheta_3\left(\frac{1}{2}; \tau_1\right)\vartheta_3\left(\frac{\Delta_1}{2\pi \mathrm{i}}; \tau_1\right)}
+\mathcal{O}\left(\frac{1}{t}\right),                                                             \\[0.5em]
&\cos u(x, t)=1-8
\left(\mathcal{A}_1\frac{\vartheta_3\left(\frac{1+\tau_1}{2}+\frac{\Delta_1}{2\pi \mathrm{i}}; \tau_1\right)\vartheta_3\left(0; \tau_1\right)}{\vartheta_3'\left(\frac{1+\tau_1}{2}; \tau_1\right)\vartheta_3\left(\frac{\Delta_1}{2\pi \mathrm{i}}; \tau_1\right)}\right)^2
\mathrm{e}^{\Delta_1}
+\mathcal{O}\left(\frac{1}{t}\right),                                  \\[0.5em]
&\sin u(x, t)=4\mathcal{A}_1
\frac{\vartheta_3\left(\frac{\tau_1}{2}-\frac{\Delta_1}{2\pi \mathrm{i}}; \tau_1\right)
\vartheta_3\left(\frac{1+\tau_1}{2}+\frac{\Delta_1}{2\pi \mathrm{i}}; \tau_1\right)
\vartheta_3^2\left(0; \tau_1\right)}{\vartheta_3\left(\frac{\tau_1}{2}; \tau_1\right)
\vartheta_3'\left(\frac{1+\tau_1}{2}; \tau_1\right)
\vartheta_3^2\left(\frac{\Delta_1}{2\pi \mathrm{i}}; \tau_1\right)}
+\mathcal{O}\left(\frac{1}{t}\right),
\end{aligned}
\end{gather}
where $\Delta_1$ is denoted as
\begin{gather}
\Delta_1=\Omega_1\left(x+\frac{t}{4\eta_1\eta_2}+\phi_1\right)
\end{gather}

\end{enumerate}

\end{theorem}

\begin{remark}
In the results presented above, we restricted our focus to a single singularity, $\eta_0$, for the generalized reflection coefficients $r_0$ and $r_c$. However, it is feasible to generalize these coefficients to include an arbitrary number $n$ of singularities, denoted by the set $\{\eta_{0,j}\}_{j=1}^n$.

For the first generalized reflection coefficient, we can define:
\begin{gather}\label{general-0}
r_0 = (\lambda - \eta_1)^{\beta_1} (\eta_2 - \lambda)^{\beta_2} \left(\prod_{j=1}^n |\lambda - \eta_{0, j}|^{\beta_{0, j}}\right) \gamma(\lambda),
\end{gather}
where $\eta_1 < \eta_{0, 1} < \eta_{0, 2} < \cdots < \eta_{0, n} < \eta_2$ and each $\beta_{0, j} \in (-1, 0) \cup (0, +\infty)$.

Similarly, for the second generalized reflection coefficient, we define:
\begin{gather}\label{general-1}
r_c = (\lambda - \eta_1)^{\beta_1} (\eta_2 - \lambda)^{\beta_2} \left(\prod_{j=1}^n \chi_j(\lambda)\right) \gamma(\lambda),
\end{gather}
where $\chi_j(\lambda)$ is a step-like function such that $\chi_j(\lambda) = 1$ for $\lambda \in [\eta_{0, j-1}, \eta_{0, j})$ and $\chi_j(\lambda) = c_j^2$ for $\lambda \in (\eta_{0, j}, \eta_{0, j+1}]$, with constants $c_j \neq 0$. Here, we take $\eta_1 = \eta_{0, 0} < \eta_{0, 1} < \eta_{0, 2} < \cdots < \eta_{0, n} < \eta_{0, n+1} = \eta_2$.

The long-time asymptotic behavior of kink-soliton gases defined by these extended forms \eqref{general-0} and \eqref{general-1} can be determined using the methods outlined in this paper, with the added construction of local parametrices around each $\eta_{0, j}$. For the first case \eqref{general-0}, the local parametrix around $\eta_{0, j}$ can be derived using the second type of modified Bessel parametrix, while for the second case \eqref{general-1}, confluent hypergeometric parametrix are employed for the construction of these local solutions.

\end{remark}

\vspace{1em}
\noindent \textbf{Notations.}\,\,
This introduction concludes with an overview of the notational conventions utilized throughout the paper. The subscripts \( + \) and \( - \) indicate the non-tangential boundary values taken from the left and right sides, respectively, along a jump contour in the context of a Riemann-Hilbert problem.

For brevity in expressing the jump matrices, we introduce the notations \( \mathcal{L}^{\lambda_1}_{\lambda_2}\left[\lambda_0\right] \) and \( \mathcal{U}^{\lambda_1}_{\lambda_2}\left[\lambda_0\right] \), defined as:
\begin{gather}\label{LU}
\mathcal{L}^{\lambda_1}_{\lambda_2}\left[\lambda_0\right] = \mathrm{e}^{\lambda_1\sigma_3} \lambda_2^{\sigma_3} \mathcal{L}\left[\lambda_0\right] \lambda_2^{-\sigma_3} \mathrm{e}^{-\lambda_1\sigma_3}, \\[1em]
\mathcal{U}^{\lambda_1}_{\lambda_2}\left[\lambda_0\right] = \mathrm{e}^{\lambda_1\sigma_3} \lambda_2^{\sigma_3} \mathcal{U}\left[\lambda_0\right] \lambda_2^{-\sigma_3} \mathrm{e}^{-\lambda_1\sigma_3},
\end{gather}
where \( \mathcal{L}[\lambda_0] \) and \( \mathcal{U}[\lambda_0] \) denote the lower and upper triangular matrices, respectively, each with $1$'s along the diagonal:
\begin{gather}\label{LU0}
\mathcal{L}[\lambda_0] =
\begin{pmatrix}
1 & 0 \\
\lambda_0 & 1
\end{pmatrix}, \qquad
\mathcal{U}[\lambda_0] =
\begin{pmatrix}
1 & \lambda_0 \\
0 & 1
\end{pmatrix}.
\end{gather}
There are two notable special cases within \eqref{LU}. First, \( \mathcal{L}^{\lambda_1} \) and \( \mathcal{U}^{\lambda_1} \) correspond to \eqref{LU} when \( \lambda_2 = 1 \). Second, \( \mathcal{L}_{\lambda_2} \) and \( \mathcal{U}_{\lambda_2} \) are given by \eqref{LU} with \( \lambda_1 = 0 \).

The constants \( C \), \( C_0 \), and \( C_1 \) are defined as:
\begin{gather}
C = \frac{1}{\sqrt{2}}
\begin{pmatrix}
1 & \mathrm{i} \\
\mathrm{i} & 1
\end{pmatrix}, \quad
C_0 = -\sqrt{2\pi}
\begin{pmatrix}
1 & 0 \\
0 & \mathrm{i}
\end{pmatrix}, \quad
C_1 = \frac{1}{\sqrt{2}}
\begin{pmatrix}
1 & -1 \\
1 & 1
\end{pmatrix}.
\end{gather}
We denote by \( B(\lambda_0) \) a neighborhood around \( \lambda_0 \) on the \( \lambda \)-plane, while \( B^\zeta(\zeta_0) \) represents a neighborhood around \( \zeta_0 \) on the \( \zeta \)-plane. The symbols \( eK \) and \( eE \) refer to the complete elliptic integrals of the first and second kinds, respectively, defined as:
\begin{gather}\label{elliptic12}
eK(\lambda)=\int_0^{\pi/2}\frac{\mathrm{d}y}{\sqrt{1-\lambda^2\sin^2 y}}, \qquad  eE(\lambda)=\int_0^{\pi/2}\sqrt{1-\lambda^2\sin^2 y}\,\mathrm{d}y.
\end{gather}

{\it Organization of this paper}. The remainder of this paper is structured as follows: In Section 2, we present several solvable Riemann-Hilbert models that will be utilized in subsequent sections. Section 3 focuses on deriving a piecewise-defined $g$-function. In Section 4, we formulate the Riemann-Hilbert problem for the sine-Gordon kink-soliton gas. In Section 5, we propose the large-$x$ asymptotic behavior for the initial value of the kink-soliton gas. Section 6 is dedicated to deriving the long-time asymptotics of the sine-Gordon kink-soliton gas. Finally, in Section 7, we provide concluding remarks and discussions.

\section{Solvable Riemann-Hilbert models}

In what follows, we detail various solvable Riemann-Hilbert models that serve as foundational elements in our analysis.

\subsection{Airy parametrix $M^{\mathrm{Ai}}(\zeta)$}

The construction of the matrix $M^{\mathrm{Ai}}(\zeta)$ involves the Airy function of the first kind, $\mathrm{Ai}(\zeta)$, which satisfies the classical Airy differential equation:
\begin{gather}
\frac{\mathrm{d}^2y}{\mathrm{d}\zeta^2}-\zeta y=0.
\end{gather}
This Airy parametrix, introduced by Deift, Kriecherbauer, McLaughlin, Venakides, and Zhou, has been instrumental in the analysis of the asymptotic behavior of orthogonal polynomials associated with exponential weights. This method leverages the Deift-Zhou steepest descent technique in addressing a Riemann-Hilbert problem, as outlined in references \cite{44} and \cite{46}.
In this study, we adopt a slightly modified version of $M^{\mathrm{Ai}}(\zeta)$, illustrated in Figure \ref{BesselAiry}(Left). The matrix is defined as follows:
For $\zeta\in\mathrm{D}_1^\zeta$, the matrix $M^{\mathrm{Ai}}\left(\zeta\right)$ takes the form
\begin{gather}
M^{\mathrm{Ai}}\left(\zeta\right)=\zeta_0^{-\sigma_3/4}C_0
\begin{pmatrix}
\mathrm{Ai}'\left(\zeta_0\zeta\right)  & -\mathrm{e}^{2\pi \mathrm{i}/3} \mathrm{Ai}'\left(\mathrm{e}^{-2\pi\mathrm{i}/3}\zeta_0\zeta\right) \\[0.5em]
\mathrm{Ai}\left(\zeta_0\zeta\right)  & -\mathrm{e}^{-2\pi\mathrm{i}/3} \mathrm{Ai}\left(\mathrm{e}^{-2\pi\mathrm{i}/3}\zeta_0\zeta\right)
\end{pmatrix};
\end{gather}
For $\zeta\in\mathrm{D}_2^\zeta$, $M^{\mathrm{Ai}}\left(\zeta\right)$ is given by
\begin{gather}
M^{\mathrm{Ai}}\left(\zeta\right)=\zeta_0^{-\sigma_3/4}C_0
\begin{pmatrix}
-\mathrm{e}^{2\pi \mathrm{i}/3}\mathrm{Ai}'\left(\mathrm{e}^{-4\pi \mathrm{i}/3}\zeta_0\zeta\right)  & -\mathrm{e}^{-4\pi \mathrm{i}/3} \mathrm{Ai}'\left(\mathrm{e}^{2\pi \mathrm{i}/3}\zeta_0\zeta\right) \\[0.5em]
-\mathrm{e}^{-4\pi \mathrm{i}/3}\mathrm{Ai}\left(\mathrm{e}^{-4\pi \mathrm{i}/3}\zeta_0\zeta\right)  & -\mathrm{e}^{2\pi \mathrm{i}/3} \mathrm{Ai}\left(\mathrm{e}^{2\pi \mathrm{i}/3}\zeta_0\zeta\right)
\end{pmatrix};
\end{gather}
For $\zeta\in\mathrm{D}_3^\zeta$, we define
\begin{gather}
M^{\mathrm{Ai}}\left(\zeta\right)=\zeta_0^{-\sigma_3/4}C_0
\begin{pmatrix}
-\mathrm{e}^{2\pi \mathrm{i}/3}\mathrm{Ai}'\left(\mathrm{e}^{4\pi \mathrm{i}/3}\zeta_0\zeta\right)  & \mathrm{e}^{4\pi \mathrm{i}/3} \mathrm{Ai}'\left(\mathrm{e}^{2\pi \mathrm{i}/3}\zeta_0\zeta\right) \\[0.5em]
-\mathrm{e}^{4\pi \mathrm{i}/3}\mathrm{Ai}\left(\mathrm{e}^{4\pi \mathrm{i}/3}\zeta_0\zeta\right)  & \mathrm{e}^{2\pi \mathrm{i}/3} \mathrm{Ai}\left(\mathrm{e}^{2\pi \mathrm{i}/3}\zeta_0\zeta\right)
\end{pmatrix};
\end{gather}
$\zeta\in\mathrm{D}_4^\zeta$, we have
\begin{gather}
M^{\mathrm{Ai}}\left(\zeta\right)=\zeta_0^{-\sigma_3/4}C_0
\begin{pmatrix}
\mathrm{Ai}'\left(\zeta_0\zeta\right)  & \mathrm{e}^{4\pi \mathrm{i}/3} \mathrm{Ai}'\left(\mathrm{e}^{2\pi \mathrm{i}/3}\zeta_0\zeta\right) \\[0.5em]
\mathrm{Ai}\left(\zeta_0\zeta\right)  & \mathrm{e}^{2\pi \mathrm{i}/3} \mathrm{Ai}\left(\mathrm{e}^{2\pi \mathrm{i}/3}\zeta_0\zeta\right)
\end{pmatrix}.
\end{gather}
Here, $\zeta_0=\left(2/3\right)^{-2/3}$ is a scaling factor.

\begin{RH}
The matrix $M^{\mathrm{Ai}}\left(\zeta\right)$ is designed to satisfy a Riemann-Hilbert problem with several key properties.

\begin{itemize}
\item It is analytic for $\zeta \in \mathbb{C} \setminus \cup_{j=1}^4 \Sigma_j$.

\item It normalizes as $\zeta \to \infty$ according to the asymptotic condition:
\begin{gather}
M^{\mathrm{Ai}}\left(\zeta\right)=\zeta^{\sigma_3/4}C^{-1}\left(\mathbb{I}_2+\mathcal{O}\left(\frac{1}{\zeta^{3/2}}\right)\right)\mathrm{e}^{-\zeta^{3/2}\,\sigma_3}.
\end{gather}
\item On the contours $\zeta \in \bigcup_{j=1}^4 \Sigma_j^0$, the function $M^{\mathrm{Ai}}\left(\zeta\right)$ admits boundary values, which satisfy the following jump conditions:
\begin{gather}
M^{\mathrm{Ai}}_+\left(\zeta\right)=M^{\mathrm{Ai}}_-\left(\zeta\right)
\begin{cases}
\mathcal{L}\left[1\right], &\mathrm{for}\,\,\, \zeta\in\Sigma_1^0\cup\Sigma_3^0, \\[0.5em]
\mathrm{i}\sigma_2, &\mathrm{for}\,\,\, \zeta\in\Sigma_2^0, \\[0.5em]
\mathcal{U}\left[1\right], &\mathrm{for}\,\,\, \zeta\in\Sigma_3^0,
\end{cases}
\end{gather}
where $\Sigma_j^0 = \Sigma_j \setminus \{0\}$ for \(j = 1, 2, 3, 4\).

\item At the origin $\zeta = 0$, the matrix  $M^{\mathrm{Ai}}(\zeta)$ exhibits the following local behavior:
\begin{gather}
M^{\mathrm{Ai}}\left(\zeta\right)=
\mathcal{O}
\begin{pmatrix}
1 & 1\\
1 & 1
\end{pmatrix},
\quad \mathrm{as}\,\,\, \zeta\to 0.
\end{gather}
\end{itemize}

\end{RH}

\begin{figure}[!t]
\centering
\includegraphics[scale=0.3]{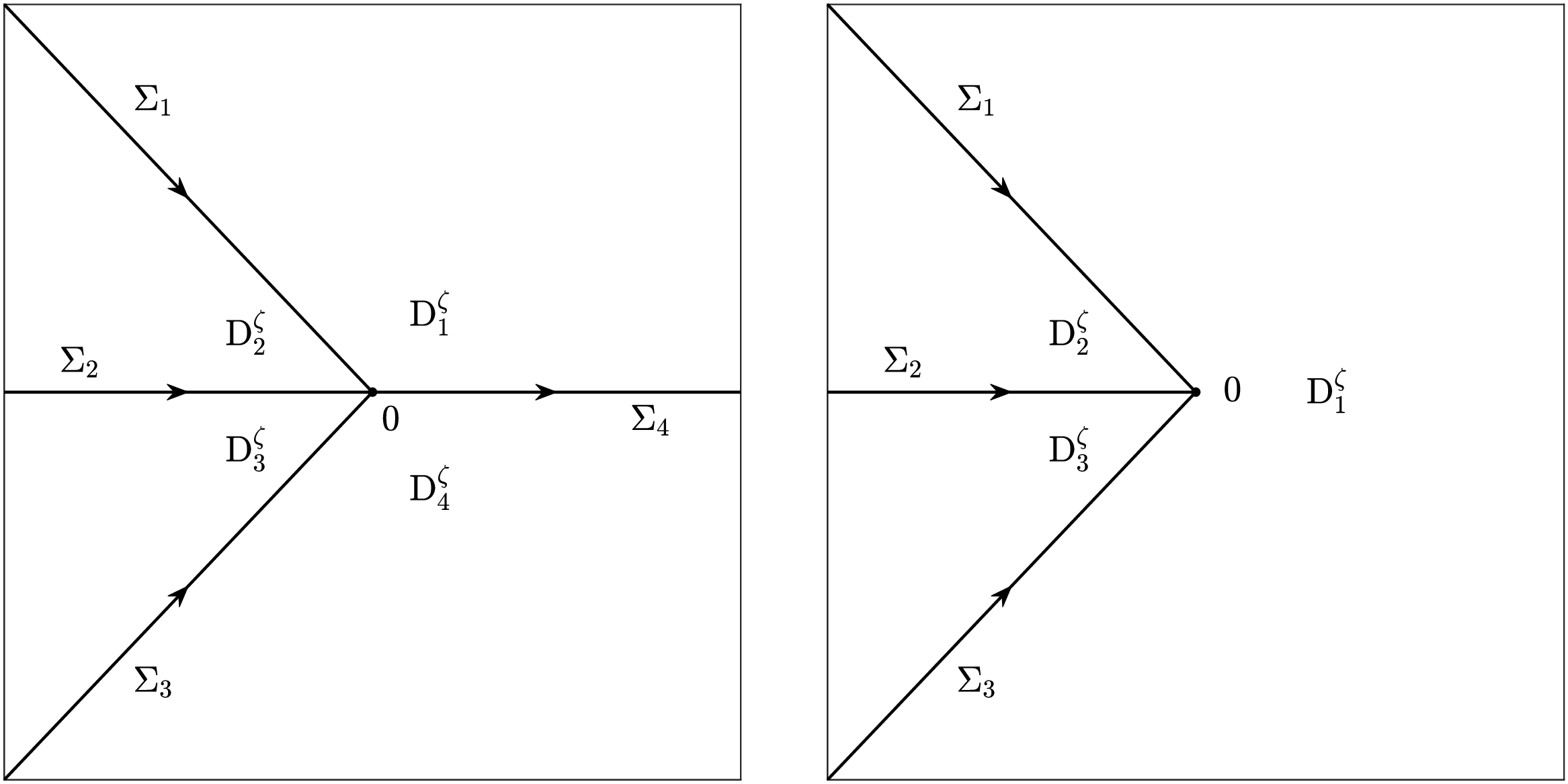}
\caption{Left: Jump contours for Airy parametrix $M^{\mathrm{mB}}$; Right: Jump contours for the first type of modified Bessel parametrix $M^{\mathrm{Ai}}$.}
\label{BesselAiry}
\end{figure}

\subsection{The first type of modified Bessel parametrix $M^{\mathrm{mB}}\left(\zeta; \beta\right)$}
The matrix \( M^{\mathrm{mB}}\left(\zeta; \beta\right) \) is constructed using the modified Bessel functions of the first and second kinds, denoted \( I_{\beta}(\zeta) \) and \( K_{\beta}(\zeta) \), respectively, where the index \(\beta\) lies within the range \((-1, +\infty)\). These functions provide solutions to the modified Bessel's differential equation:
\begin{gather}
\zeta^2 \frac{\mathrm{d}^2 y}{\mathrm{d}\zeta^2} +\zeta \frac{\mathrm{d}y}{\mathrm{d}\zeta}-\left(\zeta^2+\beta^2\right) y=0,
\end{gather}
where the first solution, \( I_{\beta}(\zeta) \), is given by the series expansion
\[
I_{\beta}(\zeta) = \left(\frac{\zeta}{2}\right)^{\beta} \sum_{n=0}^\infty \frac{\zeta^{2n}}{4^n \, \Gamma(\beta+n+1) \, n!}.
\]
Meanwhile, \( K_{\beta}(\zeta) \) is defined by its asymptotic behavior:
\[
K_{\beta}(\zeta) \sim \sqrt{\frac{\pi}{2\zeta}} \, \mathrm{e}^{-\zeta}, \quad \text{as} \quad \zeta \to \infty, \quad \text{for} \quad \arg(\zeta) \in (-3\pi/2, 3\pi/2).
\]
The modified Bessel parametrix, initially introduced by Kuijlaars, McLaughlin, Assche, and Vanlessen, is essential for capturing the local behavior near the endpoints in the analysis of orthogonal polynomials under modified Jacobi weights, as discussed in \cite{56}.
In the present context, with minor modifications as depicted in Figure \ref{BesselAiry}(Right), the parametrix \( M^{\mathrm{mB}}(\zeta; \beta) \) is constructed in the following manner:
For $\lambda\in\mathrm{D}_1^\zeta$, the matrix \( M^{\mathrm{mB}}\left(\zeta; \beta\right) \) is defined by
\begin{gather} \label{modified-Bessel-1}
M^{\mathrm{mB}}\left(\zeta; \beta\right)=-\mathrm{i}\sqrt{\pi}
\begin{pmatrix}
\mathrm{i}\sqrt{\zeta} & \\[0.5em]
&1
\end{pmatrix}
\begin{pmatrix}
I'_{\beta}\left(\sqrt{\zeta}\right)  & \mathrm{i}K'_{\beta}\left(\sqrt{\zeta}\right) \left/\right.\pi\\[0.5em]
I_{\beta}\left(\sqrt{\zeta}\right)  & \mathrm{i}K_{\beta}\left(\sqrt{\zeta}\right)\left/\right.\pi
\end{pmatrix};
\end{gather}
For $\lambda\in\mathrm{D}_2^\zeta$, the matrix is represented as
\begin{gather}\label{modified-Bessel-2}
M^{\mathrm{mB}}\left(\zeta; \beta\right)=\frac{1}{\sqrt{\pi}}
\begin{pmatrix}
\mathrm{i}\sqrt{\zeta} & \\[0.5em]
&1
\end{pmatrix}
\begin{pmatrix}
K'_{\beta}\left(\sqrt{\zeta}\,\mathrm{e}^{-\pi\mathrm{i}}\right)  & K'_{\beta}\left(\sqrt{\zeta}\right) \\[0.5em]
-K_{\beta}\left(\sqrt{\zeta}\,\mathrm{e}^{-\pi\mathrm{i}}\right)  & K_{\beta}\left(\sqrt{\zeta}\right)
\end{pmatrix};
\end{gather}
For $\lambda\in\mathrm{D}_3^\zeta$, the matrix is given by
\begin{gather}\label{modified-Bessel-3}
M^{\mathrm{mB}}\left(\zeta; \beta\right)=\frac{1}{\sqrt{\pi}}
\begin{pmatrix}
\mathrm{i}\sqrt{\zeta} & \\[0.5em]
&1
\end{pmatrix}
\begin{pmatrix}
-K'_{\beta}\left(\sqrt{\zeta}\,\mathrm{e}^{\pi\mathrm{i}}\right)  & K'_{\beta}\left(\sqrt{\zeta}\right) \\[0.5em]
K_{\beta}\left(\sqrt{\zeta}\,\mathrm{e}^{\pi\mathrm{i}}\right)  & K_{\beta}\left(\sqrt{\zeta}\right)
\end{pmatrix}.
\end{gather}

This formulation effectively encapsulates the structure of \( M^{\mathrm{mB}}(\zeta; \beta) \) across different regions, each defined by \(\mathrm{D}_j^\zeta\), which collectively characterize the modified Bessel parametrix in our setting.

\begin{RH}\label{RH-4} The matrix \( M^{\mathrm{mB}}\left(\zeta; \beta\right) \) satisfies a Riemann-Hilbert problem with the following properties:
\begin{itemize}

\item{} \( M^{\mathrm{mB}}\left(\zeta; \beta\right) \) is analytic in the complex plane for \( \zeta \in \mathbb{C} \setminus (\Sigma_1 \cup \Sigma_2 \cup \Sigma_3) \).

\item{} As \( \zeta \to \infty \), the matrix \( M^{\mathrm{mB}}\left(\zeta; \beta\right) \) exhibits the asymptotic behavior:
\begin{gather}
M^{\mathrm{mB}}\left(\zeta; \beta\right)=\zeta^{\sigma_3/4}C^{-1}\left(\mathbb{I}_2+\mathcal{O}\left(\frac{1}{\sqrt{\zeta}}\right)\right)\mathrm{e}^{\sqrt{\zeta}\,\sigma_3}.
\end{gather}

\item{} On the contours \( \zeta \in \Sigma_1^0 \cup \Sigma_2^0 \cup \Sigma_3^0 \), \( M^{\mathrm{mB}}\left(\zeta; \beta\right) \) has continuous boundary values, and these values satisfy the jump conditions:
\begin{gather}
M^{\mathrm{mB}}_+\left(\zeta; \beta\right)=M^{\mathrm{mB}}_-\left(\zeta; \beta\right)
\begin{cases}
\mathcal{L}\left[\mathrm{e}^{\beta \pi \mathrm{i}}\right], &\mathrm{for}\,\,\, \zeta\in\Sigma_1^0, \\[0.5em]
\mathrm{i}\sigma_2, &\mathrm{for}\,\,\, \zeta\in\Sigma_2^0, \\[0.5em]
\mathcal{L}\left[\mathrm{e}^{-\beta \pi \mathrm{i}}\right], &\mathrm{for}\,\,\, \zeta\in\Sigma_3^0,
\end{cases}
\end{gather}
where \( \Sigma_j^0 = \Sigma_j \setminus \{0\} \), for \( j = 1, 2, 3 \).
\item In the vicinity of the origin, \( \zeta = 0 \), the local behavior of \( M^{\mathrm{mB}}\left(\zeta; \beta\right) \) varies depending on the value of \(\beta\):
  For \(\beta > 0\), the matrix behaves as
  \begin{equation}
  M^{\mathrm{mB}}\left(\zeta; \beta\right) =
  \begin{cases}
  \mathcal{O}
  \begin{pmatrix}
  \left|\zeta\right|^{\beta/2} & \left|\zeta\right|^{-\beta/2}  \\[0.5em]
  \left|\zeta\right|^{\beta/2} & \left|\zeta\right|^{-\beta/2}
  \end{pmatrix},
  & \text{as} \,\,\, \zeta \in \mathrm{D}^\zeta_1 \to 0, \\[2em]
  \mathcal{O}
  \begin{pmatrix}
  \left|\zeta\right|^{-\beta/2} & \left|\zeta\right|^{-\beta/2}  \\[0.5em]
  \left|\zeta\right|^{-\beta/2} & \left|\zeta\right|^{-\beta/2}
  \end{pmatrix},
  & \text{as} \,\,\, \zeta \in \mathrm{D}^\zeta_2 \cup \mathrm{D}^\zeta_3 \to 0.
  \end{cases}
  \end{equation}

  For \(\beta = 0\), the behavior is given by:
  \begin{equation}
  M^{\mathrm{mB}}\left(\zeta; \beta\right) =
  \begin{cases}
  \mathcal{O}
  \begin{pmatrix}
  1 & \log\left|\zeta\right| \\[0.5em]
  1 & \log\left|\zeta\right|
  \end{pmatrix},
  & \text{as} \,\,\, \zeta \in \mathrm{D}^\zeta_1 \to 0, \\[2em]
  \mathcal{O}
  \begin{pmatrix}
  \log\left|\zeta\right| & \log\left|\zeta\right| \\[0.5em]
  \log\left|\zeta\right| & \log\left|\zeta\right|
  \end{pmatrix},
  & \text{as} \,\,\, \zeta \in \mathrm{D}^\zeta_2 \cup \mathrm{D}^\zeta_3 \to 0.
  \end{cases}
  \end{equation}

  For \( -1 < \beta < 0 \), the matrix displays the local behavior
  \begin{equation}
  M^{\mathrm{mB}}\left(\zeta; \beta\right) = \mathcal{O}
  \begin{pmatrix}
  \left|\zeta\right|^{\beta/2} & \left|\zeta\right|^{\beta/2} \\[0.5em]
  \left|\zeta\right|^{\beta/2} & \left|\zeta\right|^{\beta/2}
  \end{pmatrix}, \quad \text{as} \,\,\, \zeta \in \mathrm{D}^\zeta_1 \cup \mathrm{D}^\zeta_2 \cup \mathrm{D}^\zeta_3 \to 0.
  \end{equation}

\end{itemize}
\end{RH}

\subsection{The second type of modified Bessel parametrix $M^{\mathrm{mb}}(\zeta; \beta)$}

The matrix \( M^{\mathrm{mb}}(\zeta; \beta) \) is constructed based on the modified Bessel functions of the first and second kinds with indices \( \frac{\beta \pm 1}{2} \), where \( \beta \in (-1, 0) \cup (0, +\infty) \). It is essential to distinguish \( M^{\mathrm{mb}} \) from \( M^{\mathrm{mB}} \) as they represent distinct constructions; therefore, different notation is used for clarity.
A modified Bessel parametrix similar to \( M^{\mathrm{mb}} \) was initially proposed in \cite{61} to analyze strong asymptotic behavior of orthogonal polynomials under generalized Jacobi weights. In this study, we employ a slightly altered form, as illustrated in Figure \ref{BesselCH}(Left). The matrix \( M^{\mathrm{mb}}(\zeta; \beta) \) is constructed as follows:
For $\lambda\in\mathrm{D}_1^\zeta$, the matrix $M^{\mathrm{mb}}(\zeta; \beta)$ takes the form
\begin{gather}
M^{\mathrm{mb}}(\zeta; \beta)=C_1
\begin{pmatrix}
-G^+\left(\zeta\right) & G^+\left(\mathrm{e}^{-\pi\mathrm{i}}\zeta\right) \\[0.5em]
-G^-\left(\zeta\right) & G^-\left(\mathrm{e}^{-\pi\mathrm{i}}\zeta\right)
\end{pmatrix}
\mathrm{e}^{-\beta\pi\mathrm{i}\sigma_3/4};
\end{gather}
For $\lambda\in\mathrm{D}_2^\zeta$, the matrix $M^{\mathrm{mb}}(\zeta; \beta)$ is given by
\begin{gather}
M^{\mathrm{mb}}(\zeta; \beta)=C_1
\begin{pmatrix}
-H^+\left(\mathrm{e}^{-\pi\mathrm{i}}\zeta\right) & G^+\left(\mathrm{e}^{-\pi\mathrm{i}}\zeta\right) \\[0.5em]
-H^-\left(\mathrm{e}^{-\pi\mathrm{i}}\zeta\right) & -G^-\left(\mathrm{e}^{-\pi\mathrm{i}}\zeta\right)
\end{pmatrix}
\mathrm{e}^{-\beta\pi\mathrm{i}\sigma_3/4};
\end{gather}
For $\lambda\in\mathrm{D}_3^\zeta$, the matrix $M^{\mathrm{mb}}(\zeta; \beta)$ is defined by
\begin{gather}
M^{\mathrm{mb}}(\zeta; \beta)=C_1
\begin{pmatrix}
-H^+\left(\mathrm{e}^{-\pi\mathrm{i}}\zeta\right) & G^+\left(\mathrm{e}^{-\pi\mathrm{i}}\zeta\right) \\[0.5em]
-H^-\left(\mathrm{e}^{-\pi\mathrm{i}}\zeta\right) & -G^-\left(\mathrm{e}^{-\pi\mathrm{i}}\zeta\right)
\end{pmatrix}
\mathrm{e}^{\beta\pi\mathrm{i}\sigma_3/4};
\end{gather}
For $\lambda\in\mathrm{D}_4^\zeta$, the matrix $M^{\mathrm{mb}}(\zeta; \beta)$ takes the form
\begin{gather}
M^{\mathrm{mb}}(\zeta; \beta)=C_1
\begin{pmatrix}
-G^+\left(\mathrm{e}^{-2\pi\mathrm{i}}\zeta\right) & G^+\left(\mathrm{e}^{-\pi\mathrm{i}}\zeta\right) \\[0.5em]
-G^-\left(\mathrm{e}^{-2\pi\mathrm{i}}\zeta\right) & -G^-\left(\mathrm{e}^{-\pi\mathrm{i}}\zeta\right)
\end{pmatrix}
\mathrm{e}^{\beta\pi\mathrm{i}\sigma_3/4};
\end{gather}
For $\lambda\in\mathrm{D}_5^\zeta$, the matrix $M^{\mathrm{mb}}(\zeta; \beta)$ is in the form
\begin{gather}
M^{\mathrm{mb}}(\zeta; \beta)=C_1
\begin{pmatrix}
G^+\left(\mathrm{e}^{\pi\mathrm{i}}\zeta\right) & G^+\left(\zeta\right) \\[0.5em]
-G^-\left(\mathrm{e}^{\pi\mathrm{i}}\zeta\right) & G^-\left(\zeta\right)
\end{pmatrix}
\mathrm{e}^{-\beta\pi\mathrm{i}\sigma_3/4};
\end{gather}
For $\lambda\in\mathrm{D}_6^\zeta$, the matrix $M^{\mathrm{mb}}(\zeta; \beta)$ is defined by
\begin{gather}
M^{\mathrm{mb}}(\zeta; \beta)=C_1
\begin{pmatrix}
H^+\left(\zeta\right) & G^+\left(\zeta\right) \\[0.5em]
-H^-\left(\zeta\right) & G^-\left(\zeta\right)
\end{pmatrix}
\mathrm{e}^{-\beta\pi\mathrm{i}\sigma_3/4};
\end{gather}
For $\lambda\in\mathrm{D}_7^\zeta$, the matrix $M^{\mathrm{mb}}(\zeta; \beta)$ has the form
\begin{gather}
M^{\mathrm{mb}}(\zeta; \beta)=C_1
\begin{pmatrix}
H^+\left(\zeta\right) & G^+\left(\zeta\right) \\[0.5em]
-H^-\left(\zeta\right) & G^-\left(\zeta\right)
\end{pmatrix}
\mathrm{e}^{\beta\pi\mathrm{i}\sigma_3/4};
\end{gather}
For $\lambda\in\mathrm{D}_8^\zeta$, the matrix $M^{\mathrm{mb}}(\zeta; \beta)$ is given by
\begin{gather}
M^{\mathrm{mb}}(\zeta; \beta)=C_1
\begin{pmatrix}
G^+\left(\mathrm{e}^{-\pi\mathrm{i}}\zeta\right) & G^+\left(\zeta\right) \\[0.5em]
-G^-\left(\mathrm{e}^{-\pi\mathrm{i}}\zeta\right) & G^-\left(\zeta\right)
\end{pmatrix}
\mathrm{e}^{\beta\pi\mathrm{i}\sigma_3/4}.
\end{gather}
In these expressions, the functions \( G^\pm(\zeta) = \sqrt{\zeta/\pi} \, K_{(\beta \pm 1)/2}(\zeta) \) and \( H^\pm(\zeta) = \sqrt{\pi \zeta} \, I_{(\beta \pm 1)/2}(\zeta) \) are constructed using the modified Bessel functions \( K_{\nu}(\zeta) \) and \( I_{\nu}(\zeta) \) for appropriate indices, with \(\mathrm{arg}(\zeta)\) restricted to \( (-\pi/2, 3\pi/2) \).

\begin{RH}\label{RH-5}
The matrix \( M^{\mathrm{mb}}(\zeta; \beta) \) solves a \( 2 \times 2 \) Riemann-Hilbert problem with the following characteristics:

\begin{itemize}

\item{} $M^{\mathrm{mb}}\left(\zeta; \beta\right)$ is analytic in $\zeta$ for $\zeta\in\mathbb{C}\setminus(\cup_{j=1}^8\Sigma_j)$.

\item{} At infinity, \( M^{\mathrm{mb}}(\zeta; \beta) \) satisfies the normalization conditions:
\begin{gather}
M^{\mathrm{mb}}\left(\zeta; \beta\right)=
\begin{cases}
\left(\mathbb{I}_2+\mathcal{O}\left(\zeta^{-1}\right)\right)\mathrm{i}\sigma_2\,\mathrm{e}^{-\beta\pi\mathrm{i}\sigma_3/4}\mathrm{e}^{-\zeta\sigma_3}, &\mathrm{as}\,\,\, \zeta\in\mathrm{D}^\zeta_1\cup\mathrm{D}^\zeta_2\to\infty, \\[0.5em]
\left(\mathbb{I}_2+\mathcal{O}\left(\zeta^{-1}\right)\right)\mathrm{i}\sigma_2\,\mathrm{e}^{\beta\pi\mathrm{i}\sigma_3/4}\mathrm{e}^{-\zeta\sigma_3}, & \mathrm{as}\,\,\, \zeta\in\mathrm{D}^\zeta_3\cup\mathrm{D}^\zeta_4\to\infty, \\[0.5em]
\left(\mathbb{I}_2+\mathcal{O}\left(\zeta^{-1}\right)\right)\mathrm{e}^{-\beta\pi\mathrm{i}\sigma_3/4}\mathrm{e}^{\zeta\sigma_3}, &\mathrm{as}\,\,\, \zeta\in\mathrm{D}^\zeta_5\cup\mathrm{D}^\zeta_6\to\infty, \\[0.5em]
\left(\mathbb{I}_2+\mathcal{O}\left(\zeta^{-1}\right)\right)\mathrm{e}^{\beta\pi\mathrm{i}\sigma_3/4}\mathrm{e}^{\zeta\sigma_3}, & \mathrm{as}\,\,\, \zeta\in\mathrm{D}^\zeta_7\cup\mathrm{D}^\zeta_8\to\infty.
\end{cases}
\end{gather}

\item{} For \( \zeta \in \bigcup_{j=1}^8 \Sigma_j^0 \), \( M^{\mathrm{mb}}(\zeta; \beta) \) possesses continuous boundary values that are governed by the following jump conditions:
\begin{gather}
M^{\mathrm{mb}}_+\left(\zeta; \beta\right)=M^{\mathrm{mb}}_-\left(\zeta; \beta\right)
\begin{cases}
\mathcal{L}\left[\mathrm{e}^{-\beta \pi \mathrm{i}}\right], &\mathrm{for}\,\,\, \zeta\in\Sigma_1^0\cup\Sigma_5^0, \\[0.5em]
\mathrm{e}^{\beta\pi\mathrm{i}\sigma_3/2}, &\mathrm{for}\,\,\, \zeta\in\Sigma_2^0\cup\Sigma_6^0, \\[0.5em]
\mathcal{L}\left[\mathrm{e}^{\beta \pi \mathrm{i}}\right], &\mathrm{for}\,\,\, \zeta\in\Sigma_3^0\cup\Sigma_7^0, \\[0.5em]
\mathrm{i}\sigma_2, &\mathrm{for}\,\,\, \zeta\in\Sigma_4^0\cup\Sigma_8^0,
\end{cases}
\end{gather}
where $\Sigma_j^0=\Sigma_j\setminus \{0\}, \, j=1, 2, \cdots, 8$.

\item Near the origin, for \( \beta > -1 \) and \( \beta \neq 0 \), \( M^{\mathrm{mb}}(\zeta; \beta) \) displays the following local behaviors:
\begin{gather}
M^{\mathrm{mb}}\left(\zeta; \beta\right)=
\begin{cases}
\mathcal{O}
\begin{pmatrix}
\left|\zeta\right|^{\beta/2} & \left|\zeta\right|^{-\left|\beta\right|/2}  \\[0.5em]
\left|\zeta\right|^{\beta/2} & \left|\zeta\right|^{-\left|\beta\right|/2}
\end{pmatrix},
& \mathrm{as}\,\,\, \zeta\in \mathrm{D}^\zeta_2\cup\mathrm{D}^\zeta_3\cup\mathrm{D}^\zeta_6\cup\mathrm{D}^\zeta_7\to 0, \\[2em]
\mathcal{O}
\begin{pmatrix}
\left|\zeta\right|^{-\left|\beta\right|/2} & \left|\zeta\right|^{-\left|\beta\right|/2}  \\[0.5em]
\left|\zeta\right|^{-\left|\beta\right|/2} & \left|\zeta\right|^{-\left|\beta\right|/2}
\end{pmatrix},
& \mathrm{as}\,\,\, \zeta\in \mathrm{D}^\zeta_1\cup\mathrm{D}^\zeta_4\cup\mathrm{D}^\zeta_5\cup\mathrm{D}^\zeta_8\to 0.
\end{cases}
\end{gather}
\end{itemize}
\end{RH}

\begin{figure}[!t]
\centering
\includegraphics[scale=0.28]{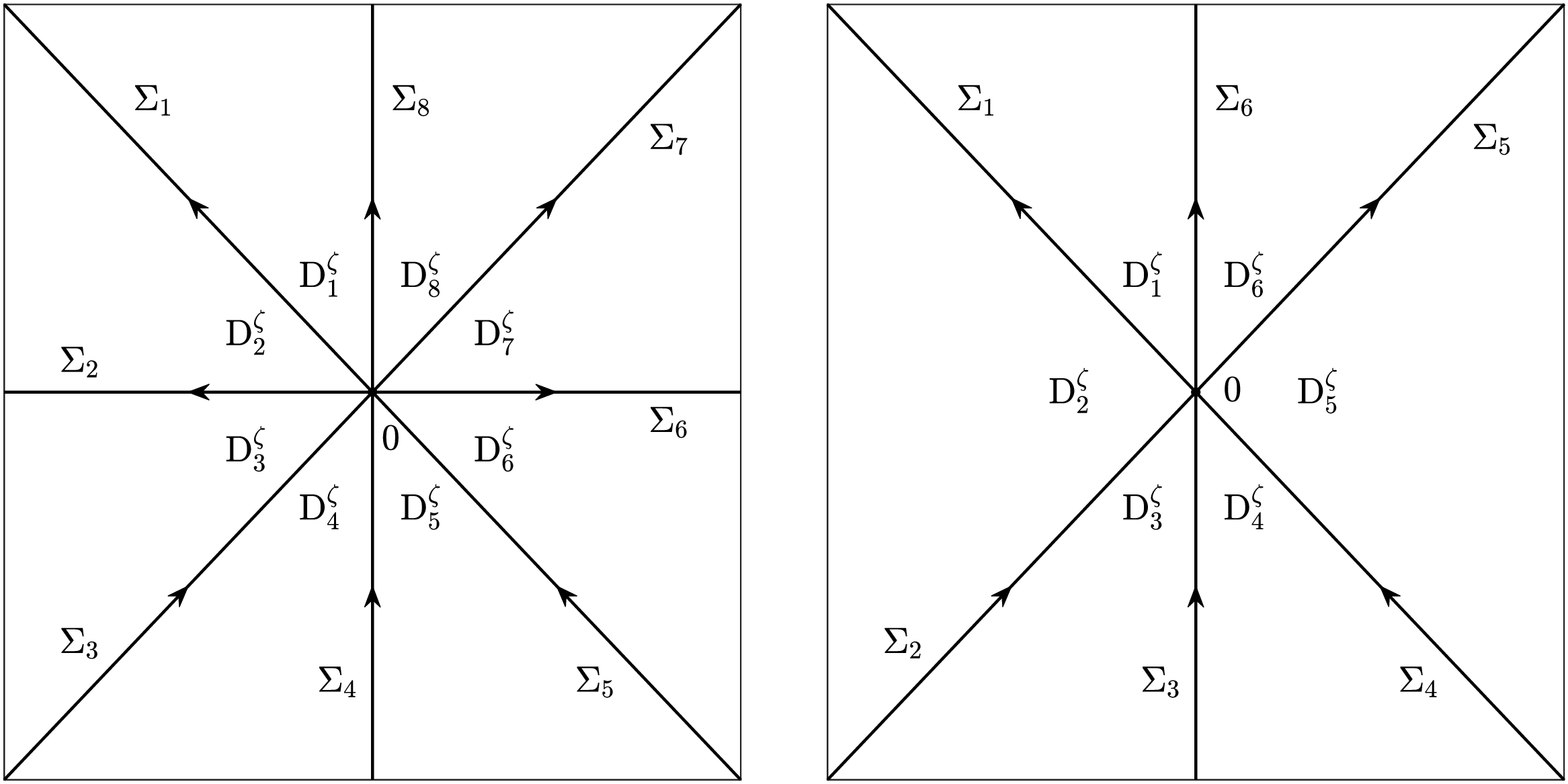}
\caption{Left: Jump contour for modified Bessel parametrix $M^{\mathrm{mb}}$; Right: Jump contour for Confluent Hypergeometric parametrix $M^{\mathrm{CH}}$.}
\label{BesselCH}
\end{figure}

\subsection{Confluent hypergeometric parametrix $M^{\mathrm{CH}}(\zeta; \kappa)$}

The confluent hypergeometric parametrix is built upon the confluent hypergeometric functions \( M(\zeta; \kappa) \) and \( U(\zeta; \kappa) \), which solve Kummer’s equation:
\begin{gather}
\zeta\frac{\mathrm{d}^2 y}{\mathrm{d}\zeta^2}+\left(1-\zeta\right)\frac{\mathrm{d}y}{\mathrm{d}\zeta}-\kappa y =0,
\end{gather}
The function \( M(\zeta; \kappa) \) is an entire solution of this equation and can be expressed as:
\[
M(\zeta; \kappa) = \sum_{n=1}^\infty \frac{\Gamma(\kappa + n)}{\Gamma(\kappa) \, n!} \zeta^n.
\]
In contrast, \( U(\zeta; \kappa) \) has a branch point at \(\zeta = 0\) and is defined by its asymptotic behavior:
\[
U(\zeta; \kappa) \sim \zeta^{-\kappa},
\]
as \( \zeta \to \infty \) for \( \mathrm{arg}(\zeta) \in (-3\pi / 2, 3\pi / 2) \). This parametrix was developed in \cite{62} to analyze the asymptotics of Hankel determinants and orthogonal polynomials under a Gaussian weight with a discontinuity.
In this work, we utilize a modified form of the confluent hypergeometric parametrix, denoted as \( M^{\mathrm{CH}}(\zeta; \kappa) \) with \( \kappa \in \mathrm{i} \mathbb{R} \), as illustrated in Figure \ref{BesselCH}(Right). The parametrix is formulated as follows:
For $\lambda\in\mathrm{D}_1^\zeta$, the matrix $M^{\mathrm{CH}}(\zeta; \kappa)$ is defined by
\begin{gather}
M^{\mathrm{CH}}(\zeta; \kappa)=
\begin{pmatrix}
-\dfrac{\Gamma\left(1+\kappa\right)}{\Gamma\left(-\kappa\right)}\,\mathrm{e}^{\kappa\pi\mathrm{i}}U\left(1+\kappa, \zeta\right) & U\left(-\kappa, \mathrm{e}^{-\pi\mathrm{i}}\zeta\right) \\[1em]
-\mathrm{e}^{\kappa\pi\mathrm{i}}U\left(\kappa, \zeta\right)& \dfrac{\Gamma\left(1-\kappa\right)}{\Gamma\left(\kappa\right)}U\left(1-\kappa,\mathrm{e}^{-\pi\mathrm{i}}\zeta\right)
\end{pmatrix}
\mathrm{e}^{-\zeta\sigma_3/2};
\end{gather}
For  $\lambda\in\mathrm{D}_2^\zeta$, it takes the form
\begin{gather}
M^{\mathrm{CH}}(\zeta; \kappa)=
\begin{pmatrix}
\Gamma\left(1+\kappa\right)M\left(-\kappa, \mathrm{e}^{-\pi\mathrm{i}}\zeta\right) & U\left(-\kappa, \mathrm{e}^{-\pi\mathrm{i}}\zeta\right) \\[1em]
-\Gamma\left(1-\kappa\right)M\left(1-\kappa, \mathrm{e}^{-\pi\mathrm{i}}\zeta\right)& \dfrac{\Gamma\left(1-\kappa\right)}{\Gamma\left(\kappa\right)}U\left(1-\kappa,\mathrm{e}^{-\pi\mathrm{i}}\zeta\right)
\end{pmatrix}
\mathrm{e}^{\zeta/2};
\end{gather}
For $\lambda\in\mathrm{D}_3^\zeta$, the expression is
\begin{gather}
M^{\mathrm{CH}}(\zeta; \kappa)=
\begin{pmatrix}
-\dfrac{\Gamma\left(1+\kappa\right)}{\Gamma\left(-\kappa\right)}\,\mathrm{e}^{-\kappa\pi\mathrm{i}}U\left(1+\kappa, \mathrm{e}^{-2\pi\mathrm{i}}\zeta\right) & U\left(-\kappa, \mathrm{e}^{-\pi\mathrm{i}}\zeta\right) \\[1em]
-\mathrm{e}^{-\kappa\pi\mathrm{i}}U\left(\kappa, \mathrm{e}^{-2\pi\mathrm{i}}\zeta\right)& \dfrac{\Gamma\left(1-\kappa\right)}{\Gamma\left(\kappa\right)}U\left(1-\kappa,\mathrm{e}^{-\pi\mathrm{i}}\zeta\right)
\end{pmatrix}
\mathrm{e}^{-\zeta\sigma_3/2};
\end{gather}
For  $\lambda\in\mathrm{D}_4^\zeta$, it is given by
\begin{gather}
M^{\mathrm{CH}}(\zeta; \kappa)=
\begin{pmatrix}
\mathrm{e}^{-\kappa\pi\mathrm{i}}U\left(-\kappa, \mathrm{e}^{\pi\mathrm{i}}\zeta\right) & \dfrac{\Gamma\left(1+\kappa\right)}{\Gamma\left(-\kappa\right)}U\left(1+\kappa, \zeta\right) \\[1em]
\dfrac{\Gamma\left(1-\kappa\right)}{\Gamma\left(\kappa\right)}\mathrm{e}^{-\kappa\pi\mathrm{i}}U\left(1-\kappa, \mathrm{e}^{\pi\mathrm{i}}\zeta\right) & U\left(\kappa, \zeta\right)
\end{pmatrix}
\mathrm{e}^{\zeta\sigma_3/2};
\end{gather}
For  $\lambda\in\mathrm{D}_5^\zeta$, the matrix is written as
\begin{gather}
M^{\mathrm{CH}}(\zeta; \kappa)=
\begin{pmatrix}
\Gamma\left(1+\kappa\right)M\left(1+\kappa, \zeta\right) & \dfrac{\Gamma\left(1+\kappa\right)}{\Gamma\left(-\kappa\right)}U\left(1+\kappa, \zeta\right) \\[1em]
-\Gamma\left(1-\kappa\right)M\left(\kappa, \zeta\right) & U\left(\kappa, \zeta\right)
\end{pmatrix}
\mathrm{e}^{-\zeta/2};
\end{gather}
For  $\lambda\in\mathrm{D}_6^\zeta$, the matrix form is
\begin{gather}
M^{\mathrm{CH}}(\zeta; \kappa)=
\begin{pmatrix}
\mathrm{e}^{\kappa\pi\mathrm{i}}U\left(-\kappa, \mathrm{e}^{-\pi\mathrm{i}}\zeta\right) & \dfrac{\Gamma\left(1+\kappa\right)}{\Gamma\left(-\kappa\right)}U\left(1+\kappa, \zeta\right) \\[1em]
\dfrac{\Gamma\left(1-\kappa\right)}{\Gamma\left(\kappa\right)}\mathrm{e}^{\kappa\pi\mathrm{i}}U\left(1-\kappa, \mathrm{e}^{-\pi\mathrm{i}}\zeta\right) & U\left(\kappa, \zeta\right)
\end{pmatrix}
\mathrm{e}^{\zeta\sigma_3/2}.
\end{gather}

\begin{RH}\label{RH-5}
The matrix \( M^{\mathrm{CH}}(\zeta; \kappa) \) satisfies a Riemann-Hilbert problem characterized by the following properties:

\begin{itemize}

\item{} $M^{\mathrm{CH}}\left(\zeta; \kappa\right)$ is analytic in $\zeta$ for $\zeta\in\mathbb{C}\setminus\left(\cup_{j=1}^6\Sigma_j\right)$.

    \item{} As \( \zeta \to \infty \), \( M^{\mathrm{CH}}(\zeta; \kappa) \) normalizes according to:
\begin{gather}
M^{\mathrm{CH}}\left(\zeta; \kappa\right)=
\begin{cases}
\left(\mathbb{I}_2+\mathcal{O}\left(\zeta^{-1}\right)\right)\zeta^{\kappa\sigma_3}\mathrm{e}^{\zeta\sigma_3/2}, &\mathrm{as}\,\,\, \zeta\in\mathrm{D}^\zeta_4\cup\mathrm{D}^\zeta_5\cup\mathrm{D}^\zeta_6\to\infty, \\[0.5em]
\left(\mathbb{I}_2+\mathcal{O}\left(\zeta^{-1}\right)\right)\mathrm{i}\sigma_2\mathrm{e}^{\kappa\pi\mathrm{i}\sigma_3}\zeta^{-\kappa\sigma_3}\mathrm{e}^{-\zeta\sigma_3}, & \mathrm{as}\,\,\, \zeta\in\mathrm{D}^\zeta_1\cup\mathrm{D}^\zeta_2\cup\mathrm{D}^\zeta_3\to\infty.
\end{cases}
\end{gather}

\item{}On \( \zeta \in \bigcup_{j=1}^6 \Sigma_j^0 \), \( M^{\mathrm{CH}}(\zeta; \kappa) \) has continuous boundary values denoted by \( M^{\mathrm{CH}}_+(\zeta; \kappa) \) and \( M^{\mathrm{CH}}_-(\zeta; \kappa) \), respectively. These values satisfy the following jump conditions:
\begin{gather}
M^{\mathrm{CH}}_+\left(\zeta; \kappa\right)=M^{\mathrm{CH}}_-\left(\zeta; \kappa\right)
\begin{cases}
\mathcal{L}\left[\mathrm{e}^{\kappa \pi \mathrm{i}}\right], &\mathrm{for}\,\,\, \zeta\in\Sigma_1^0\cup\Sigma_5^0, \\[0.5em]
\mathcal{L}\left[\mathrm{e}^{-\kappa \pi \mathrm{i}}\right], &\mathrm{for}\,\,\, \zeta\in\Sigma_2^0\cup\Sigma_4^0, \\[0.5em]
\mathrm{i}\sigma_2\,\mathrm{e}^{-\kappa\pi\mathrm{i}\sigma_3}, &\mathrm{for}\,\,\, \zeta\in\Sigma_3^0, \\[0.5em]
\mathrm{i}\sigma_2\,\mathrm{e}^{\kappa\pi\mathrm{i}\sigma_3}, &\mathrm{for}\,\,\, \zeta\in\Sigma_6^0,
\end{cases}
\end{gather}
 where $\Sigma_j^0=\Sigma_j\setminus \{0\}, \, j=1, 2, \cdots, 6$.
 \item Near the origin, \( M^{\mathrm{CH}}(\zeta; \kappa) \) exhibits the following local behaviors:
\begin{gather}
M^{\mathrm{CH}}\left(\zeta; \kappa\right)=
\begin{cases}
\mathcal{O}
\begin{pmatrix}
1 & \log\left|\zeta\right|  \\[0.5em]
1 & \log\left|\zeta\right|
\end{pmatrix},
& \mathrm{as}\,\,\, \zeta\in \mathrm{D}^\zeta_2\cup\mathrm{D}^\zeta_5\to 0, \\[2em]
\mathcal{O}
\begin{pmatrix}
\log\left|\zeta\right| & \log\left|\zeta\right|  \\[0.5em]
\log\left|\zeta\right| & \log\left|\zeta\right|
\end{pmatrix},
& \mathrm{as}\,\,\, \zeta\in \mathrm{D}^\zeta_1\cup\mathrm{D}^\zeta_3\cup\mathrm{D}^\zeta_4\cup\mathrm{D}^\zeta_6\to 0.
\end{cases}
\end{gather}

\end{itemize}
\end{RH}

\section{A generalized $g$-function construction}

The $g$-function method, originally developed by Deift, Venakides, and Zhou, emerged as a crucial tool in their study of the zero-dispersion limit of the KdV equation \cite{18}. This method has since proven fundamental in normalizing Riemann-Hilbert problems associated with sequences of orthogonal polynomials, particularly in controlling their behavior at infinity, using the equilibrium measure \cite{25, 44, 46}. To date, numerous asymptotic problems have been effectively tackled using Deift-Zhou’s steepest descent techniques, which often depend on the construction of a suitable $g$-function.

A notable example is the $g$-function constructed via a conformal mapping that transforms \(\mathbb{C} \setminus [-1, 1]\) onto the exterior of the unit circle, facilitating the analysis of strong asymptotics for orthogonal polynomials with modified Jacobi weights \cite{56}. Another variant involves deriving the $g$-function from the spectral curve, which is particularly relevant for studying the large-time asymptotics of infinite-order rogue waves \cite{31}.

In the case of the large-$x$ asymptotics for the initial value \( u(x, 0) \) of the kink-soliton gas solution of the sine-Gordon (sG) equation, constructing a \( g_0 \)-function for \( x < 0 \) becomes necessary. The \( g_0 \)-function form applied here is analogous to that in \cite{15}, and is given by:
\begin{gather}\label{g0}
g_0=\lambda-\int_{\eta_2}^\lambda \frac{\zeta^2-\rho}{R_0(\zeta)}\mathrm{d}\zeta,
\end{gather}
where \(\rho = \eta_2^2 \left(1 - \frac{eE(\eta_1 / \eta_2)}{eK(\eta_1 / \eta_2)}\right)\), and \( R_0(\zeta) \) represents a branch of \( \sqrt{(\zeta^2 - \eta_2^2)(\zeta^2 - \eta_1^2)} \) with branch cuts along \( (\eta_1, \eta_2) \cup (-\eta_2, -\eta_1) \). At infinity, \( R_0(\zeta) \) has the asymptotic form \( R_0(\zeta) = \zeta^2 + \mathcal{O}(1) \).

For a more detailed exposition on the derivation of such $g$-functions, one can consult references \cite{24, 27, 63, 64}. The function \( g_0 \) defined by \eqref{g0} is analytic in \(\lambda\) for \(\lambda \in \mathbb{C} \setminus [-\eta_2, \eta_2]\) and satisfies the normalization \( g_0(\lambda) = \mathcal{O}(\lambda^{-1}) \) as \(\lambda \to \infty\).

Additionally, \( g_0 \) adheres to the following jump conditions:
\begin{gather}
\left\{\begin{aligned}
g_{0+}+g_{0-}&=2\lambda,  && \mathrm{for}\,\,\, \lambda\in\left(\eta_1, \eta_2\right)\cup\left(-\eta_2, -\eta_1\right)\\
g_{0+}-g_{0-}&=\Omega_1,  &&\mathrm{for}\,\,\,\lambda\in\left(-\eta_1, \eta_1\right).
\end{aligned}\right.
\end{gather}

To analyze the long-time asymptotics of the kink-soliton gas \( u(x, t) \), we begin by examining the sign of the real part of \(\theta\), given by:
\begin{gather}
\Re(\theta)=\frac{\Re(\lambda)}{4}\left(\xi+\frac{1}{\left|\lambda\right|^2}\right).
\end{gather}
The sign charts corresponding to \(\Re(\theta)\) are illustrated in Figure \ref{Sign}. When \(\xi > -1/\eta_2^2\), a simple small-norm argument suffices. However, for \(\xi < -1/\eta_2^2\), constructing a $g$-function becomes necessary to achieve exponential decay.

\begin{figure}[!t]
\centering
\includegraphics[scale=0.3]{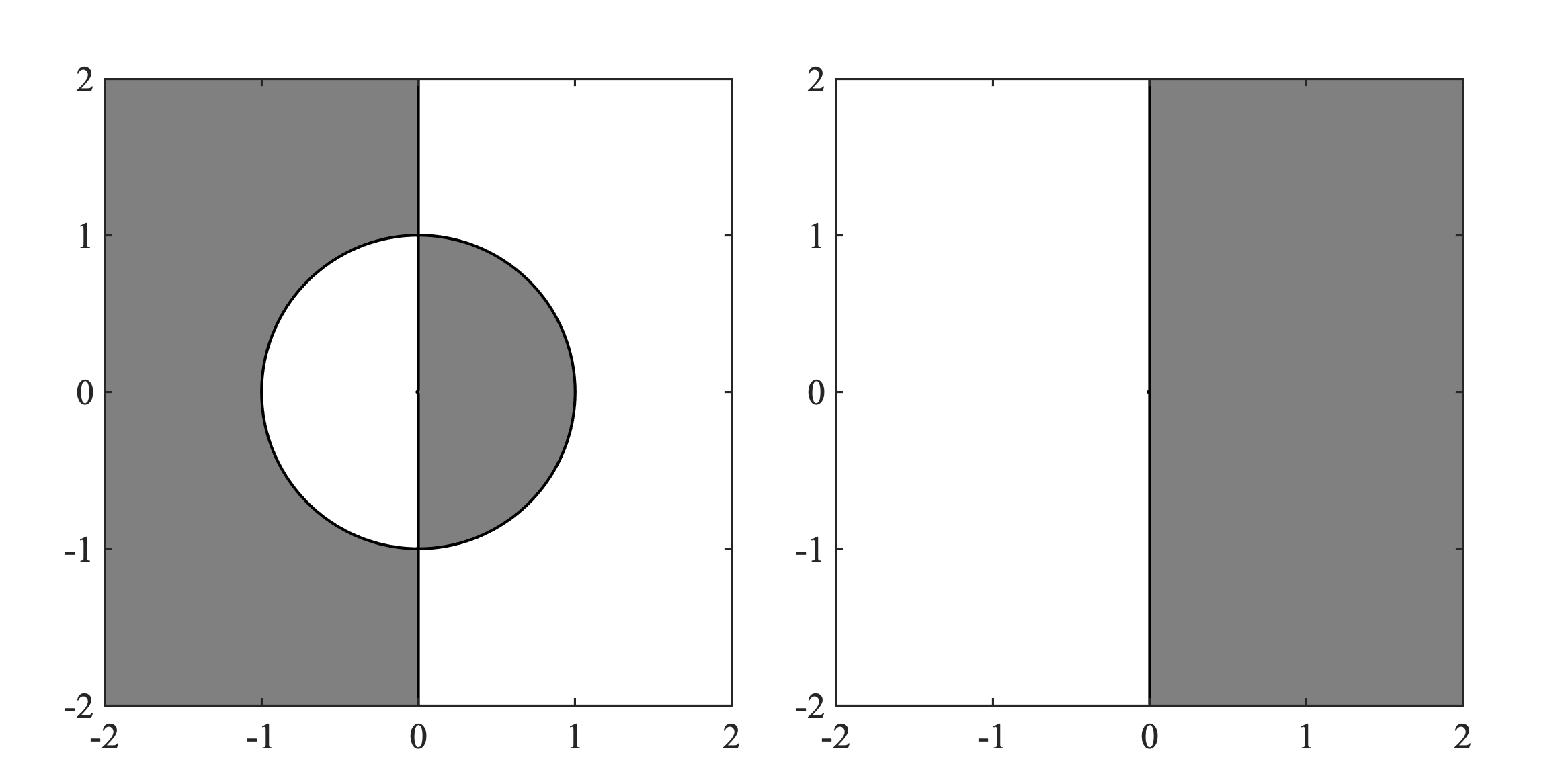}
\vspace{-0.05in}
\caption{Sign charts for $\Re\left(\theta\right)$ with $\xi=-1$(Left) and $\xi=0$(Right): $\Re\left(\theta\right)>0$ in greay regions and  $\Re\left(\theta\right)<0$ in white regions}
\label{Sign}
\end{figure}

Here, we construct a \(g\)-function \( g = g(\lambda; \xi) \) tailored to the sine-Gordon equation, defined by:
\begin{gather}\label{g-function}
g=\theta-p,
\end{gather}
where \( p = p(\lambda; \xi) \) is given as follows. For \( \lambda \in \{ \lambda \mid \Re(\lambda) \geq 0 \} \setminus [0, \eta_2) \), \( p(\lambda; \xi) \) is defined by:
\begin{gather}\label{p-right}
p(\lambda; \xi)=
\int_{\eta_2}^\lambda \frac{Q(y;\xi)}{4y^2R(y;\xi)}\mathrm{d}y.
\end{gather}
Similarly, for \( \lambda \in \{ \lambda \mid \Re(\lambda) \leq 0 \} \setminus (-\eta_2, 0] \), \( p(\lambda; \xi) \) is given by:
\begin{gather}\label{p-left}
p(\lambda; \xi)=
\int_{-\eta_2}^\lambda \frac{Q(y;\xi)}{4y^2R(y;\xi)}\mathrm{d}y,
\end{gather}
where the integration path extends from \(\pm \eta_2\) to \(\lambda\).
The functions \( R(y; \xi) \) and \( Q(y; \xi) \) are defined depending on the region \(\xi \in (\xi_{\mathrm{crit}}, -\eta_2^{-2}) \cup (-\infty, \xi_{\mathrm{crit}}) \). The function \( R(y; \xi) \) represents a branch of \( \sqrt{(y^2 - \eta_2^2)(y^2 - \alpha^2)} \), with branch cuts along \( (\alpha, \eta_2) \cup (-\eta_2, -\alpha) \), and behaves asymptotically as \( R(y; \xi) = y^2 + \mathcal{O}(1) \) at infinity. Here, \(\alpha\) is determined by the Whitham evolution equation \eqref{Whitham} when \( \xi \in (\xi_{\mathrm{crit}}, -\eta_2^{-2}) \), and by \(\alpha = \eta_1\) when \( \xi \in (-\infty, \xi_{\mathrm{crit}}) \).

For \( \xi \in (\xi_{\mathrm{crit}}, -\eta_2^{-2}) \), \( Q(y; \xi) \) is defined by:
\begin{gather}\label{Q1}
Q(y; \xi)=\left(y^2-\alpha^2\right)\left(\xi y^2-\frac{\eta_2}{\alpha}\right).
\end{gather}
For \( \xi \in (-\infty, \xi_{\mathrm{crit}}) \), \( Q(y; \xi) \) takes the form:
\begin{gather}\label{Q2}
Q(y; \xi)=\xi y^4+\eta_2^2\left(\xi+\frac{1}{\eta_1\eta_2}\right)\left(\frac{eE(\eta_1/\eta_2)}{eK(\eta_1/\eta_2)}-1\right)y^2+\eta_1\eta_2.
\end{gather}

This construction of the $g$-function provides the necessary framework for studying the exponential decay properties in the kink-soliton gas model for the sine-Gordon equation.

\begin{RH}\label{RH-7}
 The $g$-function defined by \eqref{g-function} satisfies the following Riemann-Hilbert problem:

 \begin{itemize}

 \item{} The $g$-function is analytic in \(\lambda\) for \(\lambda \in \mathbb{C} \setminus [-\eta_2, \eta_2]\).

  \item{}As \(\lambda \to \infty\), the $g$-function normalizes to \( g = \mathcal{O}(\lambda^{-1}) \).

  \item{} For \(\lambda \in (-\eta_2, -\alpha) \cup (-\alpha, \alpha) \cup (\alpha, \eta_2)\), the $g$-function has continuous boundary values, satisfying the following jump conditions:
\begin{gather}
\left\{\begin{aligned}
g_{+}+g_{-}&=2\theta,  && \mathrm{for}\,\,\, \lambda\in\left(\alpha, \eta_2\right)\cup\left(-\eta_2, -\alpha\right),\\[0.5em]
g_{+}-g_{-}&=\frac{\Omega^\alpha}{4}\left(\xi+\frac{1}{\alpha \eta_2}\right),  &&\mathrm{for}\,\,\,\lambda\in\left(-\alpha, \alpha\right),
\end{aligned}\right.
\end{gather}
where
$\Omega^\alpha=-\pi \mathrm{i}\eta_2\left/\right.eK\left(m_\alpha\right)$, if $\xi\in\left(\xi_{\mathrm{crit}}, -\eta_2^{-2}\right)$, and $\Omega^\alpha=\Omega_1$, if $\xi\in\left(-\infty, \xi_{\mathrm{crit}}\right)$.
\item Near the points \(\pm \eta_2\) and \(\pm \alpha\), the $g$-function exhibits the following local behavior:
\begin{gather}
g=\mathcal{O}\left(1\right), \quad \mathrm{as} \,\,\,\lambda\to\pm\eta_2, \pm\alpha.
\end{gather}
\end{itemize}
\end{RH}

\begin{proposition}\label{p-sign}
For $\xi\in\left(\xi_{\mathrm{crit}}, -\eta_2^{-2}\right)$, one has
\begin{gather}
\left\{\begin{array}{ll}
p_++p_->0, \quad &\mathrm{if}\,\,\, \lambda\in\left[\eta_1, \alpha \right),  \\[1em]
\Re\left(p_+\right)=0, \quad &\mathrm{if}\,\,\, \lambda\in\left[\alpha, \eta_2\right],  \\[1em]
\lim_{\Im(\lambda)\to 0^+}\dfrac{\partial \Re(p)}{\partial\Im(\lambda)}<0, \quad &\mathrm{if}\,\,\, \Re(\lambda)\in\left(\alpha, \eta_2\right).
\end{array}\right.
\end{gather}
\end{proposition}
\begin{proof}
As $\lambda\in\left[\eta_1, \alpha \right)$, one observes that
\begin{gather*}
p_++p_-=
\int_{\lambda}^\alpha \frac{Q(y;\xi)}{2y^2\left|R(y;\xi)\right|}\mathrm{d}y.
\end{gather*}
Together with the definition of $Q(y; \xi)$ in \eqref{Q1}, then $p_++p_->0$ follows.
As $\lambda\in\left[\alpha, \eta_2\right]$, a straightforward calculation shows that
\begin{gather*}
p_+=\mathrm{i}\int_{\lambda}^{\eta_2} \frac{Q(y;\xi)}{4y^2\left|R_+(y;\xi)\right|}\mathrm{d}y\in\mathrm{i}\mathbb{R}.
\end{gather*}
As $\Re(\lambda)\in\left(\alpha, \eta_2\right)$, it follows from \eqref{p-right} that
\begin{gather*}
\lim_{\Im(\lambda)\to 0^+}\frac{\partial \Re(p)}{\partial\Im(\lambda)}=\frac{Q(\Re(\lambda); \xi)}{4\Re(\lambda)^2\left|R_+\left(\Re(\lambda); \xi\right)\right|}<0.
\end{gather*}
We complete the proof.
\end{proof}

\begin{proposition}\label{p-sign-1}
for $\xi\in\left(-\infty, \xi_{\mathrm{crit}}\right)$,
\begin{gather}
\left\{\begin{array}{ll}
\Re\left(p_+\right)=0, \quad &\mathrm{if}\,\,\, \lambda\in\left[\eta_1, \eta_2\right];  \\[1em]
\lim_{\Im(\lambda)\to 0^+}\dfrac{\partial \Re(p)}{\partial\Im(\lambda)}<0, \quad &\mathrm{if}\,\,\, \lambda\in\left(\eta_1, \eta_2\right).
\end{array}\right.
\end{gather}
\end{proposition}
\begin{proof}
Performing the same way as that in Proposition \ref{p-sign}, we can prove that $p_+\in\mathrm{i}\mathbb{R}$, for $\lambda\in\left[\eta_1, \eta_2\right]$.
To prove the second property, we only need to claim the defining $Q$ in \eqref{Q2} satisfies the following inequality
\begin{gather*}
Q\left(\lambda;\xi \right)<0, \quad \text{for all}\,\,\, \lambda\in\left(\eta_1, \eta_2\right).
\end{gather*}
Note that $Q(0;\xi)=\eta_1\eta_2>0$ and $Q(\lambda; \xi)$ is a quartic and even polynomial, thus, $Q\left(\lambda;\xi \right)<0$ can be proven if $Q(\eta_1; \xi)<0$.
Taking the derivative for $Q(\eta_1; \xi)$ with respect to $\xi$, one obtains that
\begin{gather*}
\frac{\partial }{\partial \xi}Q(\eta_1, \xi)=\eta_1^2\eta_2^2\left(\frac{eE(\eta_1/\eta_2)}{eK(\eta_1/\eta_2)}+\frac{\eta_1^2}{\eta_2^2}-1\right).
\end{gather*}
Recalling the inequality
\begin{gather*}
1-\zeta^2<\frac{eE(\zeta)}{eK(\zeta)}<1, \quad \text{for all}\,\,\, \zeta\in\left(0, 1\right),
\end{gather*}
we derive $\partial Q(\eta_1, \xi)/\partial \xi>0$.
Then $Q(\eta_1; \xi)<0$ follows from the fact $Q(\eta_1; \xi_\mathrm{crit})=0$.
We complete the proof.
\end{proof}

\section{Riemann-Hilbert problem for the sine-Gordon kink-soliton gas}

The Riemann-Hilbert problem for $M^\infty$ reveals that the jump matrices \eqref{Minfty} are composed of Cauchy integrals. To simplify the problem, one can deform the jump contours and define a $2\times 2$ matrix-valued function $Y=Y(\lambda; x, t)$ as follows:
\begin{gather}
Y(\lambda; x, t)=\begin{cases}
M^\infty \mathcal{L}^{t\theta}\left[\mathrm{i}\left(\mathcal{P}_1+\mathcal{P}_2\right)
\right], &\mathrm{for}\,\,\, \lambda\,\,\, \text{interior to} \,\,\,\Gamma_+,  \\[0.5em]
M^\infty \mathcal{U}^{t\theta}\left[\mathrm{i}\left(\mathcal{P}_{-1}+\mathcal{P}_{-2}\right)
\right], &\mathrm{for}\,\,\, \lambda \,\,\, \text{interior to} \,\,\, \Gamma_-, \\[0.5em]
M^\infty, & \mathrm{for}\,\,\,\lambda\,\,\, \text{exterior to}\,\,\, \Gamma_+\cup\Gamma_-.
\end{cases}
\end{gather}

\begin{RH} The matrix function $Y(\lambda; x, t)$ satisfies the following properties:

\begin{itemize}

\item{}$Y$ is analytic in $\lambda$ for $\lambda\in\mathbb{C}\setminus\left(\left[-\eta_2, -\eta_1\right]\cup\left[\eta_1, \eta_2\right]\right)$;

\item{}It normalizes to $\mathbb{I}_2$ as $\lambda\to\infty$;

\item{} For $\lambda\in\left(\eta_1, \eta_0\right)\cup\left(\eta_0, \eta_2\right)\cup\left(-\eta_2, -\eta_0\right)\cup\left(-\eta_0, -\eta_1\right)$,  $Y$ admits continuous boundary values  denoted by $Y_+$ and $Y_-$,  respectively.
Utilizing the Sokhotski-Plemelj formula, these values  are related by the following jump conditions
\begin{gather}\label{Y-Jump}
Y_+(\lambda; x, t)=Y_-(\lambda; x, t)
\begin{cases}
\mathcal{L}^{t\theta}\left[\mathrm{i}r\right], &\mathrm{for}\,\,\, \lambda\in\left(\eta_1, \eta_0\right)\cup\left(\eta_0, \eta_2\right),\\[0.5em]
\mathcal{U}^{t\theta}\left[\mathrm{i}r\right],  &\mathrm{for}\,\,\, \lambda\in\left(-\eta_2, -\eta_0\right)\cup\left(-\eta_0, -\eta_1\right),
\end{cases}
\end{gather}
where the values of $r$ over the interval $\left(-\eta_2, -\eta_0\right)\cup\left(-\eta_0, -\eta_1\right)$ are determined by the symmetry $r(-\lambda)=r(\lambda)$.
\end{itemize}
\end{RH}
Specially, for the first type of generalized reflection coefficient $r=r_0$ in the case of $\beta_0=0$, $Y$ admits continuous boundary values for $\lambda\in\left(\eta_1, \eta_2\right)\cup\left(-\eta_2, -\eta_1\right)$, and they are related by
\begin{gather}\label{Y-Jump-Special}
Y_+(\lambda; x, t)=Y_-(\lambda; x, t)
\begin{cases}
\mathcal{L}^{t\theta}\left[\mathrm{i}r\right], &\mathrm{for}\,\,\, \lambda\in\left(\eta_1, \eta_2\right),\\[0.5em]
\mathcal{U}^{t\theta}\left[\mathrm{i}r\right],  &\mathrm{for}\,\,\, \lambda\in\left(-\eta_2, -\eta_1\right).
\end{cases}
\end{gather}
Note that $\pm\eta_1$ and $\pm\eta_2$ are endpoints, and thus the continuous boundary values $Y_\pm$ cannot be well-defined.
The difference in jump conditions from \cite{15} lies at $\lambda=\pm\eta_0$.
Since $r_0$ and $r_c$ are singular at these points, the limits of $Y$ as $\lambda\to\pm\eta_0$ do not exist.
To ensure a unique solution of $Y$, local behaviors at both the endpoints and the singularities $\pm\eta_0$ must be addressed.
Near each endpoint \( \eta_j \) for \( j=1, 2 \), as $\lambda\to\eta_j$, \( Y \) exhibits the following local behavior
\begin{gather}\label{local-Y-eta-1}
Y(\lambda; x, t)=
\begin{cases}
\mathcal{O}
\begin{pmatrix}
1 &1  \\
1 &1
\end{pmatrix}
, & \text{if} \,\,\, \beta_j\in\left(0, +\infty\right),  \\[1.5em]
\mathcal{O}
\begin{pmatrix}
\log \left|\lambda-\eta_j\right| &1  \\
\log \left|\lambda-\eta_j\right|  &1
\end{pmatrix}
, & \text{if} \,\,\, \beta_j=0,  \\[1.5em]
\mathcal{O}
\begin{pmatrix}
\left|\lambda-\eta_j\right|^{\beta_j} &1\\
\left|\lambda-\eta_j\right|^{\beta_j} &1
\end{pmatrix}
,  & \text{if} \,\,\, \beta_j\in\left(-1, 0\right),
\end{cases}
\end{gather}
with the \( \mathcal{O} \)-term interpreted element-wise.
As $\lambda\to -\eta_j$, the local behavior of $Y$ near each endpoint \( -\eta_j \) for \( j=1, 2 \) is
\begin{gather}\label{local-Y-eta-2}
Y(\lambda; x, t)=
\begin{cases}
\mathcal{O}
\begin{pmatrix}
1 &1  \\
1 &1
\end{pmatrix}
, & \text{if} \,\,\, \beta_j\in\left(0, +\infty\right),  \\[1.5em]
\mathcal{O}
\begin{pmatrix}
1 & \log\left|\lambda+\eta_j\right|  \\
1 & \log\left|\lambda+\eta_j\right|
\end{pmatrix}
, & \text{if} \,\,\, \beta_j=0,  \\[1.5em]
\mathcal{O}
\begin{pmatrix}
1& \left|\lambda+\eta_j\right|^{\beta_j}\\
1& \left|\lambda+\eta_j\right|^{\beta_j}
\end{pmatrix}
,  & \text{if} \,\,\, \beta_j\in\left(-1, 0\right).
\end{cases}
\end{gather}
It can be observed that near the endpoints $\pm\eta_j$ with $j=1, 2$, Y has the same local behavior for both $r_0$ and $r_c$, while near the $\pm\eta_0$, Y shows different local behaviors for the two cases.
For the first generalized reflection coefficient $r_0$, one only needs to consider the case of $\beta_0\ne 0$ due to the jump condition \eqref{Y-Jump-Special}, where
Y exhibits the following local behaviors:
if $\beta_0>0$,
\begin{gather}
Y=\mathcal{O}
\begin{pmatrix}
1 &1  \\
1 & 1
\end{pmatrix}, \quad
\mathrm{as}\,\,\, \lambda\to\pm\eta_0;
\end{gather}
if $\beta_0\in\left(-1, 0\right)$,
\begin{gather}\label{local-Y-eta-01}
Y(\lambda; x, t)=
\begin{cases}
\mathcal{O}
\begin{pmatrix}
\left|\lambda-\eta_0\right|^{\beta_0} &1 \\
\left|\lambda-\eta_0\right|^{\beta_0} &1
\end{pmatrix}
, & \text{as} \,\,\, \lambda\to\eta_0,  \\[1.5em]
\mathcal{O}
\begin{pmatrix}
1& \left|\lambda+\eta_0\right|^{\beta_0} \\
1& \left|\lambda+\eta_0\right|^{\beta_0}
\end{pmatrix}
,  & \text{as} \,\,\, \lambda\to-\eta_0.
\end{cases}
\end{gather}

For the second generalized coefficient $r=r_c$, Y has the following behaviors
\begin{gather}\label{local-Y-eta-02}
Y(\lambda; x, t)=
\begin{cases}
\mathcal{O}
\begin{pmatrix}
\log\left|\lambda-\eta_0\right| &1\\
\log\left|\lambda-\eta_0\right| &1
\end{pmatrix}
, & \mathrm{as} \,\,\,\lambda\to\eta_0,\\[1.5em]
\mathcal{O}
\begin{pmatrix}
1 & \log\left|\lambda+\eta_0\right| \\
1 & \log\left|\lambda+\eta_0\right|
\end{pmatrix}
, & \mathrm{as} \,\,\,\lambda\to-\eta_0.
\end{cases}
\end{gather}
Furthermore, the kink-soliton gas, denoted as $u=u(x, t)$ can be recovered from $Y$ by
\begin{gather}\label{potential-formula}
\begin{aligned}
&\frac{\partial \,u(x, t)}{\partial x}=4 \lim_{\lambda\to\infty} \lambda Y_{1, 2}(\lambda; x, t), \\[0.5em]
&\cos u(x, t)=1-2Y_{1, 2}(0; x, t)^2, \\[0.5em]
&\sin u(x, t)=-2Y_{1, 1}(0; x, t)Y_{1, 2}(0; x, t).
\end{aligned}
\end{gather}
Specially, $Y(\lambda; x, 0)$ is related with the initial value $u(x, 0)$ of the kink-soliton gas $u(x, t)$. Its jump condition is written as
\begin{gather}\label{Y0-Jump}
Y_+(\lambda; x, 0)=Y_-(\lambda; x, 0)
\begin{cases}
\mathcal{L}^{\lambda x}\left[\mathrm{i}r\right], &\mathrm{for}\,\,\, \lambda\in\left(\eta_1, \eta_0\right)\cup\left(\eta_0, \eta_2\right),\\[0.5em]
\mathcal{U}^{\lambda x}\left[\mathrm{i}r\right],  &\mathrm{for}\,\,\, \lambda\in\left(-\eta_2, -\eta_0\right)\cup\left(-\eta_0, -\eta_1\right),
\end{cases}
\end{gather}
if the reflection coefficient is taken as $r=r_0, r_c$ with $\beta_0\ne 0$.
For the first type of reflection coefficient $r=r_0$ with $\beta_0=0$, the jump condition is
\begin{gather}\label{Y0-Jump-Special}
Y_+(\lambda; x, 0)=Y_-(\lambda; x, 0)
\begin{cases}
\mathcal{L}^{\lambda x}\left[\mathrm{i}r\right], &\mathrm{for}\,\,\, \lambda\in\left(\eta_1, \eta_2\right),\\[0.5em]
\mathcal{U}^{\lambda x}\left[\mathrm{i}r\right],  &\mathrm{for}\,\,\, \lambda\in\left(-\eta_2, -\eta_1\right).
\end{cases}
\end{gather}

\section{Large-$x$ asymptotics for the initial value $u(x, 0)$ of the sine-Gordon kink-soliton gas}
As \( x \to +\infty \), the asymptotic behavior can be derived using a standard small-norm argument, leading to \eqref{large-x-right}. In this section, we derive the asymptotic behavior for the regime \( x \to -\infty \). To achieve this, we employ a series of transformations: \( Y(\lambda; x, 0) \to T(\lambda; x, 0) \to S(\lambda; x, 0) \to E(\lambda; x, 0) \). These transformations ensure that \( E(\lambda; x, 0) \) normalizes to the identity matrix \( \mathbb{I}_2 \) as \(\lambda \to \infty\) and that the associated jump matrices decay exponentially and uniformly to \( \mathbb{I}_2 \).

Following the Deift-Zhou steepest descent method, we proceed by performing triangular decompositions, which facilitate contour deformation:
\begin{equation}
\begin{aligned}
\mathcal{L}^{\lambda x}[\mathrm{i}r] &= \mathcal{U}^{\lambda x}[-\mathrm{i}r^{-1}] \left( \mathcal{L}^{\lambda x}[\mathrm{i}r] + \mathcal{U}^{\lambda x}[\mathrm{i}r^{-1}] - 2\mathbb{I}_2 \right) \mathcal{U}^{\lambda x}[-\mathrm{i}r^{-1}], \\[0.5em]
\mathcal{U}^{\lambda x}[\mathrm{i}r] &= \mathcal{L}^{\lambda x}[-\mathrm{i}r^{-1}] \left( \mathcal{L}^{\lambda x}[\mathrm{i}r^{-1}] + \mathcal{U}^{\lambda x}[\mathrm{i}r] - 2\mathbb{I}_2 \right) \mathcal{L}^{\lambda x}[-\mathrm{i}r^{-1}].
\end{aligned}
\end{equation}
However, since the sign of \(\Re(\lambda)\) indicates that exponential decay fails along certain parts of the corresponding lenses, it is necessary to introduce a conjugation operation prior to contour deformation. This operation is achieved by incorporating an appropriate \( g_0 \)-function to manage the exponential terms.

To further simplify the jump matrices, a scalar \( f_0 \)-function is also introduced. This function works in conjunction with the \( g_0 \)-function to yield constant jump matrices across the relevant contours

\subsection{Riemann-Hilbert problem for $T(\lambda; x, 0)$}
The $f_0$-function is analytic in $\lambda$ for $\lambda\in\mathbb{C}\setminus\left[-\eta_2, \eta_2\right]$ and normalizes to 1 as $\lambda\to\infty$.
Its continuous boundary values $f_{0\pm}$ are related by
\begin{gather}
\left\{
\begin{aligned}
&f_{0+}f_{0-} =r, &&\text{for} \,\,\,\lambda\in\left(\eta_1, \eta_0\right)\cup\left(\eta_0, \eta_2\right);    \\[0.5em]
&f_{0+}f_{0-} =r^{-1},   && \text{for} \,\,\,\lambda\in\left(-\eta_2, -\eta_0\right)\cup\left(-\eta_0, -\eta_1\right); \\[0.5em]
&f_{0+}^{-1}f_{0-}=\mathrm{e}^{\Omega_1\phi_1},  && \text{for} \,\,\,\lambda\in\left(-\eta_1, \eta_1\right).
\end{aligned}\right.
\end{gather}
In particular, for the first type of generalized reflection coefficient $r=r_0$ in the case of $\beta_0=0$, the jump condition is slightly changed to
\begin{gather}
\left\{\begin{aligned}
&f_{0+}f_{0-} =r, &&\text{for} \,\,\,\lambda\in\left(\eta_1, \eta_2\right);    \\[0.5em]
&f_{0+}f_{0-} =r^{-1},   && \text{for} \,\,\,\lambda\in\left(-\eta_2,  -\eta_1\right); \\[0.5em]
&f_{0+}^{-1}f_{0-}=\mathrm{e}^{\Omega_1\phi_1},  && \text{for} \,\,\,\lambda\in\left(-\eta_1, \eta_1\right).
\end{aligned}\right.
\end{gather}
Divided by $R_{0+}$, the $f_0$-function is obtained by taking the logarithm and using Plemelj's formula as follows
\begin{gather}
f_0=\exp\left\{\frac{R_0}{\pi \mathrm{i}}\left(\int_{\eta_1}^{\eta_2}\frac{\log r(s)}{R_{0+}\left(s\right)}\frac{\lambda}{s^2-\lambda^2}\mathrm{d}s
-\int_{0}^{\eta_1}\frac{\Omega_1\phi_1}{R_0(s)}\frac{\lambda}{s^2-\lambda^2}\mathrm{d}s\right)\right\}.
\end{gather}
With the definitions of the $g_0$- and $f_0$-functions, the following conjugation is performed
\begin{gather}\label{Conjugation-0}
T(\lambda; x, 0)=Y(\lambda; x, 0)\mathrm{e}^{xg_0\sigma_3}f_0^{-\sigma_3}.
\end{gather}

\begin{RH} The $2\times 2$ matrix-valued function $T(\lambda; x, 0)$ solves the following Riemann-Hilbert problem:

\begin{itemize}

\item{} $T$ is analytic in $\lambda$ for $\lambda\in\mathbb{C}\setminus\left(\left[-\eta_2, \eta_2\right]\right)$;

 \item{}It normalizes to the identity matrix $\mathbb{I}_2$ as $\lambda\to\infty$;

 \item{} For  $\lambda\in\left(-\eta_2, \eta_2\right)\setminus\left\{\pm\eta_1, \pm\eta_0\right\}$, $T$ admits continuous boundary values, which are related by the following jump conditions
\begin{gather}
T_+(\lambda; x, 0)=T_-(\lambda; x, 0)
\begin{cases}
\mathcal{U}^{xp_{0-}}_{f_{0-}}\left[-\mathrm{i}r^{-1}\right]\left(\mathrm{i}\sigma_1\right)\mathcal{U}^{xp_{0+}}_{f_{0+}}\left[-\mathrm{i}r^{-1}\right], &\mathrm{for}\,\,\, \lambda\in\left(\eta_1, \eta_0\right)\cup\left(\eta_0, \eta_2\right),\\[0.5em]
\mathcal{L}^{xp_{0-}}_{f_{0-}}\left[-\mathrm{i}r^{-1}\right]\left(\mathrm{i}\sigma_1\right)\mathcal{L}^{xp_{0+}}_{f_{0+}}\left[-\mathrm{i}r^{-1}\right], &\mathrm{for}\,\,\, \lambda\in\left(-\eta_2, -\eta_0\right)\cup\left(-\eta_0, -\eta_1\right),\\[0.5em]
\mathrm{e}^{\Delta_1^0\sigma_3}, &\mathrm{for}\,\,\, \lambda\in\left(-\eta_1, \eta_1\right),
\end{cases}
\end{gather}
where $p_0$ is defined by $p_0=\lambda-g_0$.
\end{itemize}
\end{RH}
In the above jump condition, the reflection coefficient $r$ is taken as $r_0$ with $\beta_0\ne 0$, and $r_c$. When $r$ is taken as the first type $r_0$ with $\beta_0=0$, the corresponding jump condition is slightly changed to
\begin{gather}
T_+(\lambda; x, 0)=T_-(\lambda; x, 0)
\begin{cases}
\mathcal{U}^{xp_{0-}}_{f_{0-}}\left[-\mathrm{i}r^{-1}\right]\left(\mathrm{i}\sigma_1\right)\mathcal{U}^{xp_{0+}}_{f_{0+}}\left[-\mathrm{i}r^{-1}\right], &\mathrm{for}\,\,\, \lambda\in\left(\eta_1, \eta_2\right),\\[0.5em]
\mathcal{L}^{xp_{0-}}_{f_{0-}}\left[-\mathrm{i}r^{-1}\right]\left(\mathrm{i}\sigma_1\right)\mathcal{L}^{xp_{0+}}_{f_{0+}}\left[-\mathrm{i}r^{-1}\right], &\mathrm{for}\,\,\, \lambda\in\left(-\eta_2, -\eta_1\right),\\[0.5em]
\mathrm{e}^{\Delta_1^0\sigma_3}, &\mathrm{for}\,\,\, \lambda\in\left(-\eta_1, \eta_1\right).
\end{cases}
\end{gather}
Near the endpoint $\eta_j's$, $Y(\lambda; x, 0)$ exhibits the following local behaviors as $\lambda\to\eta_j, j=1, 2$,
\begin{gather}
T(\lambda; x, 0)=
\begin{cases}
\mathcal{O}
\begin{pmatrix}
\left|\lambda-\eta_j\right|^{-\beta_j/2} &\left|\lambda-\eta_j\right|^{\beta_j/2}  \\[0.5em]
\left|\lambda-\eta_j\right|^{-\beta_j/2} &\left|\lambda-\eta_j\right|^{\beta_j/2}
\end{pmatrix}
, & \text{if} \,\,\, \beta_j\in\left(0, +\infty\right),  \\[1.5em]
\mathcal{O}
\begin{pmatrix}
\log \left|\lambda-\eta_j\right| &1  \\[0.5em]
\log \left|\lambda-\eta_j\right|  &1
\end{pmatrix}
, & \text{if} \,\,\, \beta_j=0,  \\[1.5em]
\mathcal{O}
\begin{pmatrix}
\left|\lambda-\eta_j\right|^{\beta_j/2} &\left|\lambda-\eta_j\right|^{\beta_j/2}\\[0.5em]
\left|\lambda-\eta_j\right|^{\beta_j/2} &\left|\lambda-\eta_j\right|^{\beta_j/2}
\end{pmatrix}
,  & \text{if} \,\,\, \beta_j\in\left(-1, 0\right).
\end{cases}
\end{gather}
Near the endpoint $-\eta_j's$, $Y(\lambda; x, 0)$ exhibits the following local behaviors as $\lambda\to -\eta_j, j=1, 2$,
\begin{gather}
T(\lambda; x, 0)=
\begin{cases}
\mathcal{O}
\begin{pmatrix}
\left|\lambda+\eta_j\right|^{\beta_j/2} &\left|\lambda+\eta_j\right|^{-\beta_j/2}  \\[0.5em]
\left|\lambda+\eta_j\right|^{\beta_j/2} &\left|\lambda+\eta_j\right|^{-\beta_j/2}
\end{pmatrix}
, & \text{if} \,\,\, \beta_j\in\left(0, +\infty\right),  \\[1.5em]
\mathcal{O}
\begin{pmatrix}
1& \log \left|\lambda+\eta_j\right|   \\[0.5em]
1& \log \left|\lambda+\eta_j\right|
\end{pmatrix}
, & \text{if} \,\,\, \beta_j=0,  \\[1.5em]
\mathcal{O}
\begin{pmatrix}
\left|\lambda+\eta_j\right|^{\beta_j/2} &\left|\lambda+\eta_j\right|^{\beta_j/2}\\[0.5em]
\left|\lambda+\eta_j\right|^{\beta_j/2} &\left|\lambda+\eta_j\right|^{\beta_j/2}
\end{pmatrix}
,  & \text{if} \,\,\, \beta_j\in\left(-1, 0\right).
\end{cases}
\end{gather}
Near the singularities $\pm \eta_0$, different types of reflection coefficients have different local behaviors.
For the first  type of  generalized reflection coefficient $r=r_0$ with $\beta_0\ne 0$, $T(\lambda; x, 0)$ exhibits the following local behaviors
\begin{gather}
T(\lambda; x, 0)=
\begin{cases}
\mathcal{O}
\begin{pmatrix}
 \left|\lambda\mp\eta_0\right|^{\mp\beta_0/2}&\left|\lambda\mp\eta_0\right|^{\pm\beta_0/2}\\[0.5em]
  \left|\lambda\mp\eta_0\right|^{\mp\beta_0/2}&\left|\lambda\mp\eta_0\right|^{\pm\beta_0/2}
\end{pmatrix}
, & \mathrm{if} \,\,\,\beta_0\in\left(0, +\infty\right), \\[1.5em]
\mathcal{O}
\begin{pmatrix}
 \left|\lambda\mp\eta_0\right|^{\beta_0/2}&\left|\lambda\mp\eta_0\right|^{\beta_0/2}\\[0.5em]
  \left|\lambda\mp\eta_0\right|^{\beta_0/2}&\left|\lambda\mp\eta_0\right|^{\beta_0/2}
\end{pmatrix}
, & \mathrm{if} \,\,\,\beta_0\in\left(-1, 0\right),
\end{cases}
\end{gather}
while for the second  type of  generalized reflection coefficient $r=r_c$, $T(\lambda; x, 0)$ exhibits the local behaviors
\begin{gather}
T(\lambda; x, 0)=
\begin{cases}
\mathcal{O}
\begin{pmatrix}
\log\left|\lambda-\eta_0\right| &1\\
\log\left|\lambda-\eta_0\right| &1
\end{pmatrix}
, & \mathrm{as} \,\,\,\lambda\to\eta_0,\\[1.5em]
\mathcal{O}
\begin{pmatrix}
1 & \log\left|\lambda+\eta_0\right| \\
1 & \log\left|\lambda+\eta_0\right|
\end{pmatrix}
, & \mathrm{as} \,\,\,\lambda\to-\eta_0.
\end{cases}
\end{gather}

\subsection{Riemann-Hilbert problem for $S(\lambda; x, 0)$}

The subsequent step involves deforming the original Riemann-Hilbert problem \( Y(\lambda; x, 0) \) to a modified problem, denoted by \( S(\lambda; x, 0) \). This deformation is achieved by opening lenses around specified intervals, as depicted in Figure \ref{Deformation-0}. We consider general reflection coefficients, specifically \( r = r_0 \) and \( r = r_c \), where \(\beta_0 \neq 0\).

The contour deformation opens lenses above and below specific intervals. The lens domains are defined as follows:
- The domains \( \mathrm{D}^+_{1, l} \) and \( \mathrm{D}^-_{1, l} \) are lenses situated above and below the interval \( (\eta_1, \eta_0) \), respectively.
The domains \( \mathrm{D}^+_{1, r} \) and \( \mathrm{D}^-_{1, r} \) are lenses above and below \( (\eta_0, \eta_2) \).
 The domains \( \mathrm{D}^+_{2, l} \) and \( \mathrm{D}^-_{2, l} \) are lenses above and below \( (-\eta_2, -\eta_0) \).
 The domains \( \mathrm{D}^+_{2, r} \) and \( \mathrm{D}^-_{2, r} \) are lenses above and below \( (-\eta_0, -\eta_1) \).
For convenience, the following notations are introduced for these lens domains:
\[
\mathrm{D}_{j, l} = \mathrm{D}_{j, l}^+ \cup \mathrm{D}_{j, l}^-, \quad \mathrm{D}_{j, r} = \mathrm{D}_{j, r}^+ \cup \mathrm{D}_{j, r}^-, \quad \mathrm{D}_j^+ = \mathrm{D}_{j, l}^+ \cup \mathrm{D}_{j, r}^+, \quad \mathrm{D}_j^- = \mathrm{D}_{j, l}^- \cup \mathrm{D}_{j, r}^-,
\]
\[
\mathrm{D}_{j} = \mathrm{D}_{j, l} \cup \mathrm{D}_{j, r}, \quad j = 1, 2.
\]

For the particular case where \( r = r_0 \) with \(\beta_0 = 0\), a different contour deformation is applied to the Riemann-Hilbert problem for \( T(\lambda; x, 0) \). This specific deformation is shown in Figure \ref{Deformation-00}, with lenses opened around:
\( \mathrm{D}^+_1 \) and \( \mathrm{D}^-_1 \), above and below \( (\eta_1, \eta_2) \), respectively.
 \( \mathrm{D}^+_2 \) and \( \mathrm{D}^-_2 \), above and below \( (-\eta_2, -\eta_1) \).

We define these domains as:
\[
\mathrm{D}_1 = \mathrm{D}^+_1 \cup \mathrm{D}^-_1, \quad \mathrm{D}_2 = \mathrm{D}^+_2 \cup \mathrm{D}^-_2.
\]
Additionally, we denote the domain outside these lenses as \( \mathrm{D}_\mathrm{o} \), given by:
\[
\mathrm{D}_\mathrm{o} = \mathbb{C} \setminus \overline{\mathrm{D}_1 \cup \mathrm{D}_2 \cup (-\eta_1, \eta_1)},
\]
with \( \mathrm{D}_\mathrm{o} = \mathrm{D}_\mathrm{o}^+ \cup \mathrm{D}_\mathrm{o}^- \cup (\eta_2, +\infty) \cup (-\infty, -\eta_2) \), where \( \mathrm{D}_\mathrm{o}^+ \) and \( \mathrm{D}_\mathrm{o}^- \) denote the parts in the upper and lower half-planes, respectively. This configuration aids in handling the Riemann-Hilbert problem with lenses opened for different cases, ensuring suitable decay properties across the contours.

\begin{figure}[!t]
\centering
\includegraphics[scale=0.38]{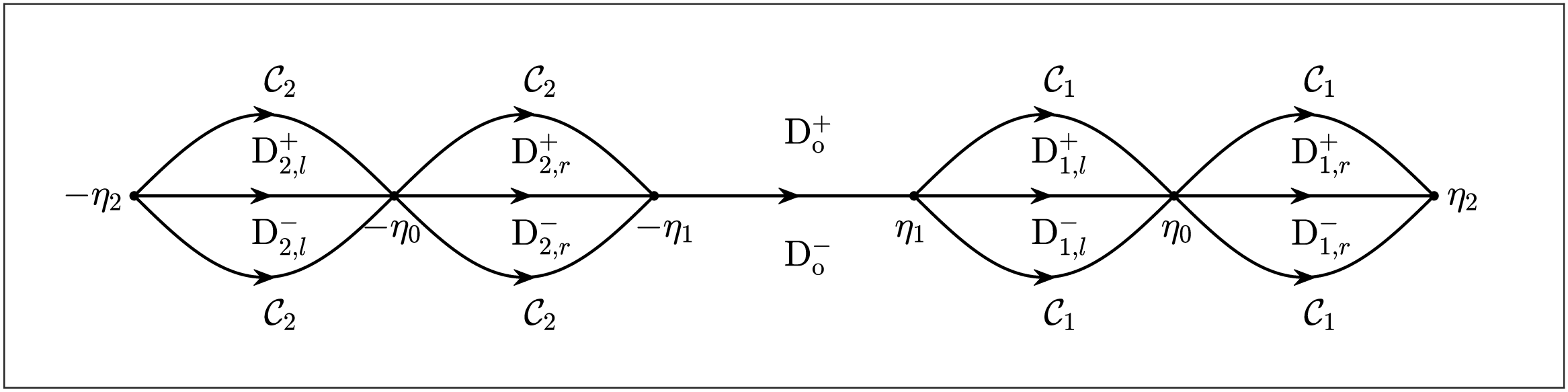}
\caption{Contour deformation by opening lenses for $r=r_0, r_c$ with $\beta_0\ne 0$.}
\label{Deformation-0}
\end{figure}

\begin{figure}[!t]
\centering
\includegraphics[scale=0.38]{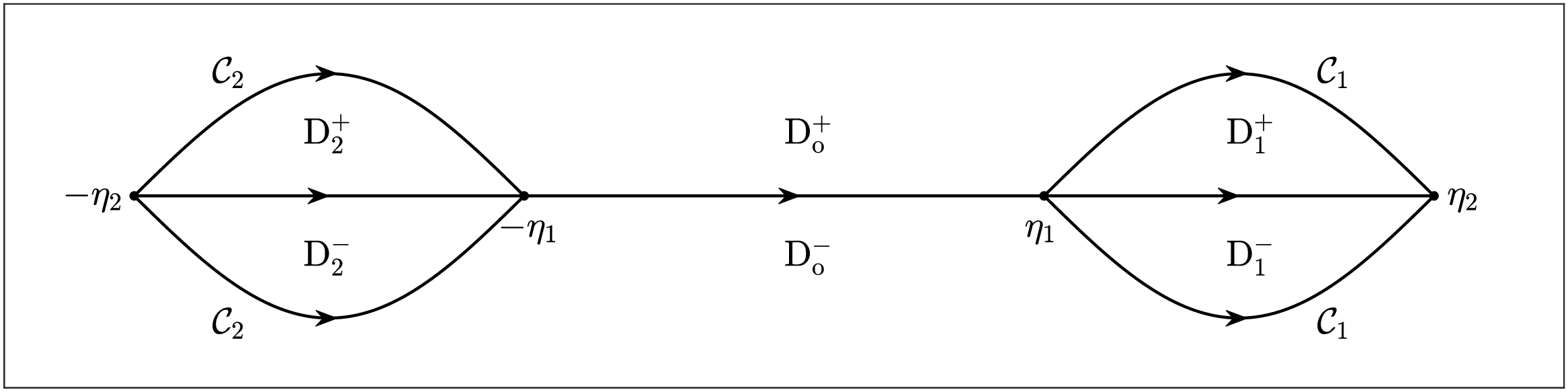}
\caption{Contour deformation by opening lenses for $r=r_0$ with $\beta_0=0$ }
\label{Deformation-00}
\end{figure}

The $2\times 2$ matrix-valued function $S(\lambda; x, 0)$ is defined as follows:
\begin{gather}\label{T0-To-S0}
S(\lambda; x, 0)=\begin{cases}
T(\lambda; x, 0)\mathcal{U}^{xp_0}_{f_0}\left[-\mathrm{i}r^{-1}\right]^{\mp 1}, & \mathrm{for}\,\,\, \lambda\in\mathrm{D}_1^\pm,  \\[0.5em]
T(\lambda; x, 0)\mathcal{L}^{xp_0}_{f_0}\left[-\mathrm{i}r^{-1}\right]^{\mp 1}, & \mathrm{for}\,\,\,  \lambda\in\mathrm{D}_2^\pm,  \\[0.5em]
T(\lambda; x, 0), & \mathrm{for}\,\,\,  \lambda\in\mathrm{D}_\mathrm{o}.
\end{cases}
\end{gather}

\begin{RH}
The matrix \( S(\lambda; x, 0) \) solves the following Riemann-Hilbert problem:

\begin{itemize}
\item{} $S(\lambda; x, 0)$ is analytic in $\lambda$ for $\lambda\in\mathbb{C}\setminus\left(\left[-\eta_2, \eta_2\right]\cup\mathcal{C}_1\cup\mathcal{C}_2\right)$;

\item{}It normalizes to the identity matrix $\mathbb{I}_2$
as $\lambda\to\infty$;

\item{}For \( \lambda \in \left(-\eta_2, \eta_2\right) \cup \mathcal{C}_1 \cup \mathcal{C}_2 \setminus \{\pm \eta_1, \pm \eta_0\} \), \( S(\lambda; x, 0) \) has continuous boundary values denoted by \( S_+(\lambda; x, 0) \) and \( S_-(\lambda; x, 0) \), which satisfy the following jump conditions:
\begin{gather}
S_+(\lambda; x, 0)=S_-(\lambda; x, 0)
\begin{cases}
\mathcal{U}^{xp_0}_{f_0}\left[-\mathrm{i}r^{-1}\right], &\mathrm{for}\,\,\, \lambda\in\mathcal{C}_1\setminus\left\{\eta_1, \eta_2, \eta_0\right\},\\[0.5em]
\mathcal{L}^{xp_0}_{f_0}\left[-\mathrm{i}r^{-1}\right], &\mathrm{for}\,\,\, \lambda\in\mathcal{C}_2\setminus\left\{-\eta_1, -\eta_2, -\eta_0\right\},\\[0.5em]
\mathrm{i}\sigma_1, &\mathrm{for}\,\,\, \lambda\in\left(\eta_1, \eta_2\right)\cup\left(-\eta_2, -\eta_1\right)\setminus\{\pm\eta_0\},\\[0.5em]
\mathrm{e}^{\Delta_1^0\sigma_3}, &\mathrm{for}\,\,\, \lambda\in\left(-\eta_1, \eta_1\right).
\end{cases}
\end{gather}
\end{itemize}
\end{RH}
For the specific case where \( r = r_0 \) with \( \beta_0 = 0 \), the jump conditions simplify as follows:
\begin{gather}
S_+(\lambda; x, 0)=S_-(\lambda; x, 0)
\begin{cases}
\mathcal{U}^{xp_0}_{f_0}\left[-\mathrm{i}r^{-1}\right], &\mathrm{for}\,\,\, \lambda\in\mathcal{C}_1\setminus\left\{\eta_1, \eta_2\right\},\\[0.5em]
\mathcal{L}^{xp_0}_{f_0}\left[-\mathrm{i}r^{-1}\right], &\mathrm{for}\,\,\, \lambda\in\mathcal{C}_2\setminus\left\{-\eta_1, -\eta_2\right\},\\[0.5em]
\mathrm{i}\sigma_1, &\mathrm{for}\,\,\, \lambda\in\left(\eta_1, \eta_2\right)\cup \left(-\eta_2, -\eta_1\right),\\[0.5em]
\mathrm{e}^{\Delta_1^0\sigma_3}, &\mathrm{for}\,\,\, \lambda\in\left(-\eta_1, \eta_1\right).
\end{cases}
\end{gather}

The local behavior of the matrix function \( S(\lambda; x, 0) \) near the endpoints \( \pm \eta_j \), \( j=1, 2 \), varies depending on whether \(\lambda\) is in the exterior region \( \mathrm{D}_{\mathrm{o}} \) or within the lens regions \( \mathrm{D}_1 \cup \mathrm{D}_2 \). These behaviors also depend on the parameter \(\beta_j\):
If $\beta_j\in \left(0, +\infty\right)$, $S(\lambda; x, 0)$ exhibits the following local behavior
\begin{gather}
S(\lambda; x, 0)=
\begin{cases}
\mathcal{O}
\begin{pmatrix}
 \left|\lambda\mp\eta_j\right|^{\mp\beta_j/2}&\left|\lambda\mp\eta_j\right|^{\pm\beta_j/2}\\[0.5em]
  \left|\lambda\mp\eta_j\right|^{\mp\beta_j/2}&\left|\lambda\mp\eta_j\right|^{\pm\beta_j/2}
\end{pmatrix}
, & \mathrm{as} \,\,\,\lambda\in\mathrm{D}_{\mathrm{o}}\to\pm\eta_j, \\[1.5em]
\mathcal{O}
\begin{pmatrix}
 \left|\lambda\mp\eta_j\right|^{-\beta_j/2}&\left|\lambda\mp\eta_j\right|^{-\beta_j/2}\\[0.5em]
 \left|\lambda\mp\eta_j\right|^{-\beta_j/2}&\left|\lambda\mp\eta_j\right|^{-\beta_j/2}
\end{pmatrix}
, & \mathrm{as} \,\,\,\lambda\in\mathrm{D}_1\cup\mathrm{D}_2\to\pm\eta_j,
\end{cases}
\end{gather}
 if $\beta_j=0$, the local behavior is
 \begin{gather}
S(\lambda; x, 0)=
\begin{cases}
\mathcal{O}
\begin{pmatrix}
\log\left|\lambda-\eta_j\right| &1\\[0.5em]
\log\left|\lambda-\eta_j\right| &1
\end{pmatrix}
, & \mathrm{as} \,\,\,\lambda\in\mathrm{D}_{\mathrm{o}}\to\eta_j,\\[1.5em]
\mathcal{O}
\begin{pmatrix}
1 & \log\left|\lambda+\eta_j\right| \\[0.5em]
1 & \log\left|\lambda+\eta_j\right|
\end{pmatrix}
, & \mathrm{as} \,\,\,\lambda\in\mathrm{D}_{\mathrm{o}}\to-\eta_j, \\[1.5em]
\mathcal{O}
\begin{pmatrix}
\log\left|\lambda\mp\eta_j\right| & \log\left|\lambda\mp\eta_j\right| \\[0.5em]
\log\left|\lambda\mp\eta_j\right| & \log\left|\lambda\mp\eta_j\right|
\end{pmatrix}
, & \mathrm{as} \,\,\,\lambda\in\mathrm{D}_1\cup\mathrm{D}_2\to \pm\eta_j,
\end{cases}
\end{gather}
 and  if $\beta_j\in\left(-1, 0\right)$, the local behavior is formulated as
\begin{gather}
S(\lambda; x, 0)=
\mathcal{O}
\begin{pmatrix}
 \left|\lambda\mp\eta_j\right|^{\beta_j/2}&\left|\lambda\mp\eta_j\right|^{\beta_j/2}  \\[0.5em]
 \left|\lambda\mp\eta_j\right|^{\beta_j/2}&\left|\lambda\mp\eta_j\right|^{\beta_j/2}
\end{pmatrix}, \quad
\mathrm{as} \,\,\,\lambda\in\mathrm{D}_1\cup\mathrm{D}_2\cup \mathrm{D}_{\mathrm{o}}\to\pm\eta_j.
\end{gather}

For the singularities at \(\pm \eta_0\), the local behavior of \( S(\lambda; x, 0) \) also depends on the specific type of reflection coefficient:
For the first type of generalized reflection coefficient $r=r_0$ with $\beta_0\ne 0$, $S(\lambda; x, 0)$ exhibits the following local behavior:
If $\beta_0\in \left(0, +\infty\right)$, the  local behavior is
\begin{gather}
S(\lambda; x, 0)=
\begin{cases}
\mathcal{O}
\begin{pmatrix}
 \left|\lambda\mp\eta_0\right|^{\mp\beta_0/2}&\left|\lambda\mp\eta_0\right|^{\pm\beta_0/2}\\[0.5em]
  \left|\lambda\mp\eta_0\right|^{\mp\beta_0/2}&\left|\lambda\mp\eta_0\right|^{\pm\beta_0/2}
\end{pmatrix}
, & \mathrm{as} \,\,\,\lambda\in\mathrm{D}_{\mathrm{o}}\to\pm\eta_0, \\[1.5em]
\mathcal{O}
\begin{pmatrix}
 \left|\lambda\mp\eta_0\right|^{-\beta_0/2}&\left|\lambda\mp\eta_0\right|^{-\beta_0/2}\\[0.5em]
 \left|\lambda\mp\eta_0\right|^{-\beta_0/2}&\left|\lambda\mp\eta_0\right|^{-\beta_0/2}
\end{pmatrix}
, & \mathrm{as} \,\,\,\lambda\in\mathrm{D}_1\cup\mathrm{D}_2\to\pm\eta_0,
\end{cases}
\end{gather}
and  if $\beta_0\in\left(-1, 0\right)$, the local behavior is
\begin{gather}
S(\lambda; x, 0)=
\mathcal{O}
\begin{pmatrix}
 \left|\lambda\mp\eta_0\right|^{\beta_0/2}&\left|\lambda\mp\eta_0\right|^{\beta_0/2}  \\[0.5em]
 \left|\lambda\mp\eta_0\right|^{\beta_0/2}&\left|\lambda\mp\eta_0\right|^{\beta_0/2}
\end{pmatrix}, \quad
\mathrm{as} \,\,\,\lambda\in\mathrm{D}_1\cup\mathrm{D}_2\cup \mathrm{D}_{\mathrm{o}}\to\pm\eta_0.
\end{gather}
For the second type of generalized reflection coefficient $r=r_c$, $S(\lambda; x, 0)$ exhibits the following local behavior
 \begin{gather}
S(\lambda; x, 0)=
\begin{cases}
\mathcal{O}
\begin{pmatrix}
\log\left|\lambda-\eta_0\right| &1\\[0.5em]
\log\left|\lambda-\eta_0\right| &1
\end{pmatrix}
, & \mathrm{as} \,\,\,\lambda\in\mathrm{D}_{\mathrm{o}}\to\eta_0,\\[1.5em]
\mathcal{O}
\begin{pmatrix}
1 & \log\left|\lambda+\eta_0\right| \\[0.5em]
1 & \log\left|\lambda+\eta_0\right|
\end{pmatrix}
, & \mathrm{as} \,\,\,\lambda\in\mathrm{D}_{\mathrm{o}}\to-\eta_0, \\[1.5em]
\mathcal{O}
\begin{pmatrix}
\log\left|\lambda\mp\eta_0\right| & \log\left|\lambda\mp\eta_0\right| \\[0.5em]
\log\left|\lambda\mp\eta_0\right| & \log\left|\lambda\mp\eta_0\right|
\end{pmatrix}
, & \mathrm{as} \,\,\,\lambda\in\mathrm{D}_1\cup\mathrm{D}_2\to \pm\eta_0.
\end{cases}
\end{gather}

\subsection{Riemann-Hilbert problem for $E(\lambda; x, 0)$}

The error matrix $E(\lambda; x, 0)$ is defined as:
\begin{gather}\label{E0}
E(\lambda; x, 0)=S(\lambda; x, 0)P(\lambda; x, 0)^{-1},
\end{gather}
where \( P(\lambda; x, 0) \) is the global parametrix, which differs based on the reflection coefficients.

For general reflection coefficients \( r = r_0 \) or \( r = r_c \) with \(\beta_0 \neq 0\), \( P(\lambda; x, 0) \) is defined as:
\begin{gather}
P(\lambda; x, 0)=
\begin{cases}
P^\infty\left(\lambda; x, 0\right), &\mathrm{for}\,\,\,\lambda\in\mathbb{C}\setminus\overline{B\left(\pm\eta_2, \pm\eta_0, \pm\eta_1\right)},  \\[0.5em]
P^{\eta_j}\left(\lambda; x, 0\right), & \mathrm{for}\,\,\,\lambda\in B\left(\eta_j\right),\,\, j=0,1,2,  \\[0.5em]
\sigma_2P^{\eta_j}\left(-\lambda; x, 0\right)\sigma_2, & \mathrm{for}\,\,\,\lambda\in B\left(-\eta_j\right),\,\, j=0,1,2,
\end{cases}
\end{gather}
with $B\left(\pm\eta_2, \pm\eta_0, \pm\alpha\right)=B(\eta_2)\cup B(-\eta_2)\cup B(\eta_0)\cup B(-\eta_0)\cup B(\alpha)\cup B(-\alpha)$;
For the special case of the first type of reflection coefficient \( r = r_0 \) with \(\beta_0 = 0\), \( P(\lambda; x, 0) \) is defined as:
\begin{gather}
P(\lambda; x, 0)=
\begin{cases}
P^\infty\left(\lambda; x, 0\right), &\mathrm{for}\,\,\,\lambda\in\mathbb{C}\setminus\overline{B\left(\pm\eta_2, \pm\eta_1\right)},  \\[0.5em]
P^{\eta_j}\left(\lambda; x, 0\right), & \mathrm{for}\,\,\,\lambda\in B\left(\eta_j\right),\,\, j=1,2,  \\[0.5em]
\sigma_2P^{\eta_j}\left(-\lambda; x, 0\right)\sigma_2, & \mathrm{for}\,\,\,\lambda\in B\left(-\eta_j\right),\,\, j=1,2,
\end{cases}
\end{gather}
with $B\left(\pm\eta_2, \pm\alpha\right)=B(\eta_2)\cup B(-\eta_2)\cup B(\alpha)\cup B(-\alpha)$.

The outer parametrix \( P^\infty(\lambda; x, 0) \) is constructed as follows:
\begin{gather}
\begin{aligned}
P^\infty_{1, 1}(\lambda; x, 0)&=\frac{\delta_1+\delta_1^{-1}}{2}
\frac{\vartheta_3\left(w_1+\frac{1}{4}+\frac{\Delta^0_1}{2\pi\mathrm{i}}; \tau_1\right)}{\vartheta_3\left(w_1+\frac{1}{4}; \tau_1\right)}
\frac{\vartheta_3\left(0; \tau_1\right)}{\vartheta_3\left(\frac{\Delta^0_1}{2\pi\mathrm{i}}; \tau_1\right)}, \\
P^\infty_{1, 2}(\lambda; x, 0)&=\frac{\delta_1-\delta_1^{-1}}{2}
\frac{\vartheta_3\left(w_1-\frac{1}{4}-\frac{\Delta^0_1}{2\pi\mathrm{i}}; \tau_1\right)}{\vartheta_3\left(w_1-\frac{1}{4}; \tau_1\right)}
\frac{\vartheta_3\left(0; \tau_1\right)}{\vartheta_3\left(\frac{\Delta^0_1}{2\pi\mathrm{i}}; \tau_1\right)},\\
P^\infty_{2, 1}(\lambda; x, 0)&=\frac{\delta_1-\delta_1^{-1}}{2}
\frac{\vartheta_3\left(w_1-\frac{1}{4}+\frac{\Delta^0_1}{2\pi\mathrm{i}}; \tau_1\right)}{\vartheta_3\left(w_1-\frac{1}{4}; \tau_1\right)}
\frac{\vartheta_3\left(0; \tau_1\right)}{\vartheta_3\left(\frac{\Delta^0_1}{2\pi\mathrm{i}}; \tau_1\right)}, \\
P^\infty_{2, 2}(\lambda; x, 0)&=\frac{\delta_1+\delta_1^{-1}}{2}
\frac{\vartheta_3\left(w_1+\frac{1}{4}-\frac{\Delta^0_1}{2\pi\mathrm{i}}; \tau_1\right)}{\vartheta_3\left(w_1+\frac{1}{4}; \tau_1\right)}
\frac{\vartheta_3\left(0; \tau_1\right)}{\vartheta_3\left(\frac{\Delta^0_1}{2\pi\mathrm{i}}; \tau_1\right)}.
\end{aligned}
\end{gather}
Here, \(\delta_1 = \delta_1(\lambda)\) is a branch of \( \left((\lambda + \eta_1)(\lambda - \eta_2) / (\lambda + \eta_2)(\lambda - \eta_1)\right)^{1/4} \), with branch cuts along \([ \eta_1, \eta_2 ] \cup [-\eta_2, -\eta_1]\), and it satisfies \(\delta_1 = 1 + \mathcal{O}(\lambda^{-1})\) as \(\lambda \to \infty\). The term \(w_1 = w_1(\lambda)\) is defined by:
\begin{gather}
w_1=-\frac{\eta_2}{4eK\left(m_1\right)}\int_{\eta_2}^\lambda \frac{\mathrm{d}\zeta}{R_0\left(\zeta\right)}.
\end{gather}
This construction ensures that \( P(\lambda; x, 0) \) closely approximates the solution \( S(\lambda; x, 0) \) around the points \( \pm\eta_j \) and satisfies the necessary jump conditions to facilitate the solution of the Riemann-Hilbert problem through the error matrix \( E(\lambda; x, 0) \).

The conformal mapping \(\mathcal{F}_0^{\eta_2} = p_0^2\) is defined in the vicinity of \(\lambda = \eta_2\), mapping specific regions as follows:
\(\mathrm{D}_\mathrm{o} \cap B(\eta_2)\) is mapped to \(\mathrm{D}_1^\zeta \cap B^\zeta(0)\),
 \(\mathrm{D}_{1}^+ \cap B(\eta_2)\) to \(\mathrm{D}_2^\zeta \cap B^\zeta(0)\),
 \(\mathrm{D}_{1}^- \cap B(\eta_2)\) to \(\mathrm{D}_3^\zeta \cap B^\zeta(0)\).
The local parametrix $P^{\eta_2}(\lambda; x, 0)$ for $\lambda\in\left(\mathrm{D}_1\cup\mathrm{D}_\mathrm{o}\right)\cap B(\eta_2)$,  is constructed as:
\begin{gather}
P^{\eta_2}(\lambda; x, 0)=P^\infty(\lambda; x, 0)A_0^{\eta_2}C\zeta_{\eta_2}^{-\sigma_3/4}M^{\mathrm{mB}}\left(\zeta_{\eta_2}; \beta_2\right)\mathrm{e}^{-\sqrt{\zeta_{\eta_2}}\sigma_3}\left(A_0^{\eta_2}\right)^{-1},
\end{gather}
where $\zeta_{\eta_2}=x^2\mathcal{F}_0^{\eta_2}$, $A_0^{\eta_2}=\left(\mathrm{e}^{\pi\mathrm{i}/4/}/f_0d^{\eta_2}\right)^{\sigma_3}\sigma_2$.
The function \( d^{\eta_2} \) varies with the reflection coefficient:
for the first generalized reflection coefficient, $r=r_0$,
\bee d^{\eta_2}=(\lambda-\eta_1)^{\beta_1/2}(\lambda-\eta_2)^{\beta_2/2}\left|\lambda-\eta_0\right|^{\beta_0/2}\gamma(\lambda)^{1/2},
\ene
and for the second generalized reflection coefficient, $r=r_c$, \bee d^{\eta_2}=(\lambda-\eta_1)^{\beta_1/2}(\lambda-\eta_2)^{\beta_2/2}\chi_c(\lambda)^{1/2}\gamma(\lambda)^{1/2}.
\ene

Similarly, the conformal mapping in the neighborhood of \(\lambda = \eta_1\), denoted \(\mathcal{F}_0^{\eta_1}\), is defined as \(\mathcal{F}_0^{\eta_1} = (p_0 \pm \Omega_1 / 2)^2\) for \(\lambda \in B(\eta_1) \cap \mathbb{C}^\pm\), mapping:
\(\mathrm{D}_\mathrm{o} \cap B(\eta_1)\) to \(\mathrm{D}_1^\zeta \cap B^\zeta(0)\),
 \(\mathrm{D}_{1}^- \cap B(\eta_1)\) to \(\mathrm{D}_2^\zeta \cap B^\zeta(0)\),
 \(\mathrm{D}_{1}^+ \cap B(\eta_1)\) to \(\mathrm{D}_3^\zeta \cap B^\zeta(0)\).
The local parametrix \( P^{\eta_1}(\lambda; x, 0) \) is formulated for \(\lambda \in (\mathrm{D}_{1}^\pm \cup \mathrm{D}_\mathrm{o}^\pm) \cap B(\eta_1)\) as:
\begin{gather}
P^{\eta_1}(\lambda; x, 0)=P^\infty(\lambda; x, 0)A_{0\pm}^{\eta_1}C\zeta_{\eta_1}^{-\sigma_3/4}M^{\mathrm{mB}}(\zeta_{\eta_1}; \beta_1)\mathrm{e}^{-\sqrt{\zeta_{\eta_1}}\sigma_3}\left(A_{0\pm}^{\eta_1}\right)^{-1},
\end{gather}
where  $\zeta_{\eta_1}=x^2\mathcal{F}_0^{\eta_1}$, and $A_{0\pm}^{\eta_1}=\left(\mathrm{e}^{\pi\mathrm{i}/4\mp x\Omega_1 /2}/f_0d^{\eta_1}\right)^{\sigma_3}\sigma_1$.
The function \( d^{\eta_1} \) is defined as:
for the first type of generalized reflection coefficient, $r=r_0$, \bee
d^{\eta_1}=(\eta_1-\lambda)^{\beta_1/2}(\eta_2-\lambda)^{\beta_2/2}\left|\lambda-\eta_0\right|^{\beta_0/2}\gamma(\lambda)^{1/2};
\ene
for the second type of generalized reflection coefficient, $r=r_c$, \bee
d^{\eta_1}=(\eta_1-\lambda)^{\beta_1/2}(\eta_2-\lambda)^{\beta_2/2}\chi_c(\lambda)^{1/2}\gamma(\lambda)^{1/2}.
\ene

For the first type of generalized reflection coefficient \( r = r_0 \) with \( \beta_0 \neq 0 \), the conformal mapping in the vicinity of \( \lambda = \eta_0 \) is defined by
\[
\mathcal{F}_0^{\eta_0} = \mp \left( p_0 - p_{0\pm}(\eta_0) \right),
\]
for \(\lambda \in B(\eta_0) \cap \mathbb{C}^\pm\). The mapping corresponds as follows:
 $\mathrm{D}_{1, r}^+\cap B(\eta_0)$ to $\mathrm{D}_1^\zeta\cap B^\zeta(0)$,
$\mathrm{D}_{1, l}^+\cap B(\eta_0)$ to $\mathrm{D}_4^\zeta\cap B^\zeta(0)$,
 $\mathrm{D}_{1, l}^-\cap B(\eta_0)$ to $\mathrm{D}_5^\zeta\cap B^\zeta(0)$,
 and
 $\mathrm{D}_{1, r}^-\cap B(\eta_0)$ to $\mathrm{D}_8^\zeta\cap B^\zeta(0)$.
 The domains are defined as:
 $\mathrm{D}_{\mathrm{o}, 1, r}^+=\left(\mathcal{F}_0^{\eta_0}\right)^{-1} \left(\mathrm{D}^\zeta_2\cap B^\zeta(0)\right)$,
 $\mathrm{D}_{\mathrm{o}, 1, l}^+=\left(\mathcal{F}_0^{\eta_0}\right)^{-1} \left(\mathrm{D}^\zeta_3\cap B^\zeta(0)\right)$,
 $\mathrm{D}_{\mathrm{o}, 1, l}^-=\left(\mathcal{F}_0^{\eta_0}\right)^{-1} \left(\mathrm{D}^\zeta_6\cap B^\zeta(0)\right)$,
and
 $\mathrm{D}_{\mathrm{o}, 1, r}^-=\left(\mathcal{F}_0^{\eta_0}\right)^{-1} \left(\mathrm{D}^\zeta_7\cap B^\zeta(0)\right)$.
The local parametrix \( P^{\eta_0}(\lambda; x, 0) \) is constructed as follows:
For $\lambda\in\left(\mathrm{D}_{1, r}^+\cup\mathrm{D}_{\mathrm{o}, 1, r}^+\right)\cap B(\eta_0)$,
\begin{gather}
P^{\eta_0}(\lambda; x, 0)=P^\infty(\lambda; x, 0) A_{0r+}^{\eta_0}\mathrm{e}^{\beta_0\pi\mathrm{i}\sigma_3/4}\left(-\mathrm{i}\sigma_2\right)M^{\mathrm{mb}}(\zeta_{\eta_0}; \beta_0)\mathrm{e}^{\zeta_{\eta_0}\sigma_3}\left(A_{0r+}^{\eta_0}\right)^{-1},
\end{gather}
with $A_{0r+}^{\eta_0}=\left(\mathrm{e}^{\pi\mathrm{i}/4+xp_{0+}(\eta_0)}/f_0d_r^{\eta_0}\right)^{\sigma_3}\sigma_2$;
For $\lambda\in \left(\mathrm{D}_{1, l}^+\cup\mathrm{D}_{\mathrm{o}, 1, l}^+\right)\cap B(\eta_0)$, $P^{\eta_0}$ is formulated as
\begin{gather}
P^{\eta_0}(\lambda; x, 0)=P^\infty (\lambda; x, 0)A_{0l+}^{\eta_0}\mathrm{e}^{-\beta_0\pi\mathrm{i}\sigma_3/4}\left(-\mathrm{i}\sigma_2\right)M^{\mathrm{mb}}(\zeta_{\eta_0}; \beta_0)\mathrm{e}^{\zeta_{\eta_0}\sigma_3}\left(A_{0l+}^{\eta_0}\right)^{-1},
\end{gather}
with $A_{0l+}^{\eta_0}=\left(\mathrm{e}^{\pi\mathrm{i}/4+xp_{0+}(\eta_0)}/f_0d_l^{\eta_0}\right)^{\sigma_3}\sigma_2$;
For $\lambda\in\left(\mathrm{D}_{1, l}^-\cup\mathrm{D}_{\mathrm{o}, 1, l}^-\right)\cap B(\eta_0)$,  $P^{\eta_0}$ is written as
\begin{gather}
P^{\eta_0}(\lambda; x, 0)=P^\infty(\lambda; x, 0) A_{0l-}^{\eta_0}\mathrm{e}^{\beta_0\pi\mathrm{i}\sigma_3/4}M^{\mathrm{mb}}(\zeta_{\eta_0}; \beta_0)\mathrm{e}^{-\zeta_{\eta_0}\sigma_3}\left(A_{0l-}^{\eta_0}\right)^{-1},
\end{gather}
with $A_{0l-}^{\eta_0}=\left(\mathrm{e}^{\pi\mathrm{i}/4+xp_{0-}(\eta_0)}/f_0d_l^{\eta_0}\right)^{\sigma_3}\sigma_2$;
For $\lambda\in \left(\mathrm{D}_{1, r}^-\cup\mathrm{D}_{\mathrm{o}, 1, r}^-\right)\cap B(\eta_0)$, $P^{\eta_0}$ is expressed as
\begin{gather}
P^{\eta_0}(\lambda; x, 0)=P^\infty(\lambda; x, 0) A_{0r-}^{\eta_0} \mathrm{e}^{-\beta_0\pi\mathrm{i}\sigma_3/4}M^{\mathrm{mb}}(\zeta_{\eta_0}; \beta_0)\mathrm{e}^{-\zeta_{\eta_0}\sigma_3}\left(A_{0r-}^{\eta_0}\right)^{-1},
\end{gather}
with $A_{0r-}^{\eta_0}=\left(\mathrm{e}^{\pi\mathrm{i}/4+xp_{0-}(\eta_0)}/f_0d_r^{\eta_0}\right)^{\sigma_3}\sigma_2$,
where $\zeta_{\eta_0}=-x\mathcal{F}^{\eta_0}_0$.
The functions \( d_l^{\eta_0} \) and \( d_r^{\eta_0} \) are defined as:
\[
d_l^{\eta_0} = (\lambda - \eta_1)^{\beta_1/2} (\eta_2 - \lambda)^{\beta_2/2} (\lambda - \eta_0)^{\beta_0/2} \gamma(\lambda)^{1/2},
\]
\[
d_r^{\eta_0} = (\lambda - \eta_1)^{\beta_1/2} (\eta_2 - \lambda)^{\beta_2/2} (\eta_0 - \lambda)^{\beta_0/2} \gamma(\lambda)^{1/2}.
\]

For the second type of generalized reflection coefficient \( r = r_c \), the conformal mapping around \( \lambda = \eta_0 \) is defined by:
\[
\mathcal{F}_0^{\eta_0} = \mp 2 \left( p_0 - p_{0\pm}(\eta_0) \right),
\]
for \(\lambda \in B(\eta_0) \cap \mathbb{C}^\pm\). This mapping has the following correspondences:
 $\mathrm{D}_{1, r}^+\cap B(\eta_0)$ to $\mathrm{D}_1^\zeta\cap B^\zeta(0)$,
$\mathrm{D}_\mathrm{o}^+\cap B(\eta_0)$ to $\mathrm{D}_2^\zeta\cap B^\zeta(0)$,
$\mathrm{D}_{1, l}^+\cap B(\eta_0)$ to $\mathrm{D}_3^\zeta\cap B^\zeta(0)$,
 $\mathrm{D}_{1, l}^-\cap B(\eta_0)$ to $\mathrm{D}_4^\zeta\cap B^\zeta(0)$,
 $\mathrm{D}_\mathrm{o}^-\cap B(\eta_0)$ to $\mathrm{D}_5^\zeta\cap B^\zeta(0)$,
 and
 $\mathrm{D}_{1, r}^-\cap B(\eta_0)$ to $\mathrm{D}_6^\zeta\cap B^\zeta(0)$.
The local parametrix  in the neighborhood of $\lambda=\eta_0$ is constructed as follows:
For $\lambda\in \left(\mathrm{D}_1^+\cup \mathrm{D}_\mathrm{o}^+\right)\cap B(\eta_0)$, the local parametrix is formulated as
\begin{gather}
P^{\eta_0}(\lambda; x, 0)=P^\infty (\lambda; x, 0) A^{\eta_0}_{0+}\left(\zeta_{\eta_0}^{\kappa_0\sigma_3}\mathrm{i}\sigma_2\mathrm{e}^{\kappa_0\pi\mathrm{i}\sigma_3}\right)^{-1}M^{\mathrm{CH}}(\zeta_{\eta_0}; \kappa_0)\mathrm{e}^{\zeta_{\eta_0}\sigma_3/2}\left(A^{\eta_0}_{0+}\right)^{-1};
\end{gather}
For $\lambda\in \left(\mathrm{D}_1^-\cup \mathrm{D}_\mathrm{o}^-\right)\cap B(\eta_0)$, the local parametrix is expressed as
\begin{gather}
P^{\eta_0}(\lambda; x, 0)=P^\infty (\lambda; x, 0)A^{\eta_0}_{0-}\mathrm{e}^{-\kappa_0\pi\mathrm{i}\sigma_3}M^{\mathrm{CH}}(\zeta_{\eta_0};  \kappa_0)\mathrm{e}^{-\zeta_{\eta_0}\sigma_3/2} \left(A^{\eta_0}_{0-}\right)^{-1}.
\end{gather}
Here:
\( A^{\eta_0}_{0\pm} = \left( \mathrm{e}^{\pi \mathrm{i} / 4 + x p_{0\pm}(\eta_0)} / f_0 d^{\eta_0} \right)^{\sigma_3} \sigma_2 \),
 \(\zeta_{\eta_0} = -x \mathcal{F}_0^{\eta_0}\),
 \(\kappa_0 = \frac{\mathrm{i}}{\pi} \log c \in \mathrm{i}\mathbb{R}\), where \( c \) is a constant parameter,
 \( d^{\eta_0} = \left( \lambda - \eta_1 \right)^{\beta_1 / 2} \left( \eta_2 - \lambda \right)^{\beta_2 / 2} c^{1/2} \gamma(\lambda)^{1/2} \).

\begin{figure}[!t]
\centering
\includegraphics[scale=0.35]{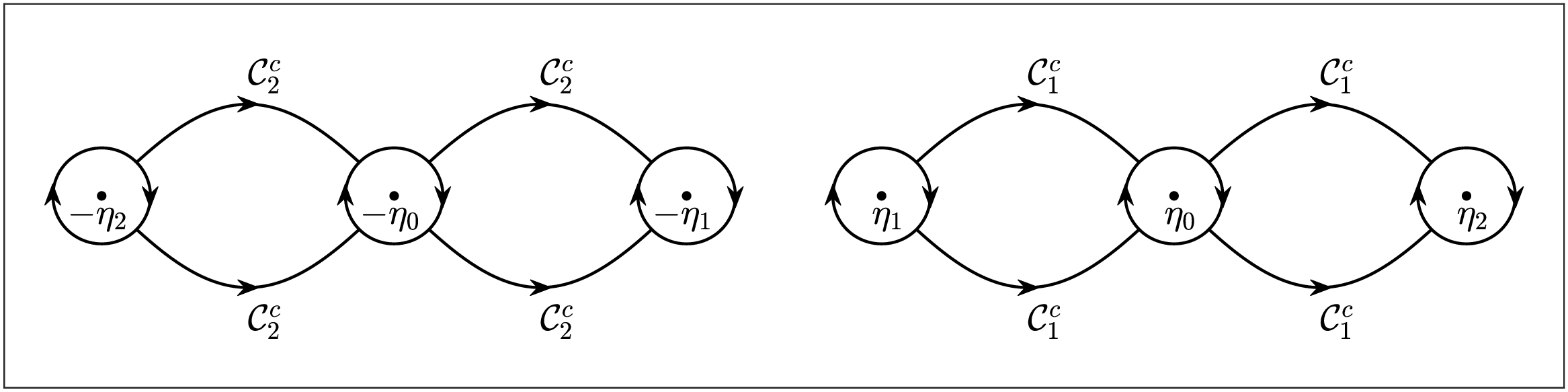}
\caption{Jump contours of the error matrix $E(\lambda; x, 0)$ for $r=r_0, r_c$ with $\beta_0\ne 0$}
\label{ErrorE0}
\end{figure}

\begin{figure}[!t]
\vspace{0.1in}
\centering
\includegraphics[scale=0.35]{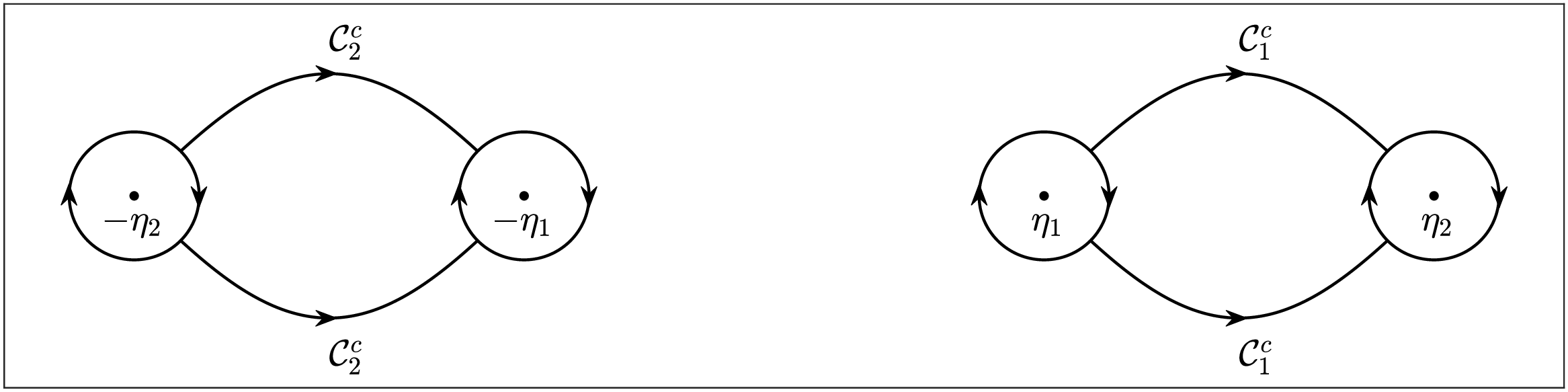}
\caption{Jump contours of the error vector $E(\lambda; x, 0)$ for $r=r_0$ with $\beta_0=0$}
\label{ErrorE01}
\end{figure}

\begin{RH} For the general reflection coefficients $r=r_0, r_c$ with $\beta_0\ne 0$, the error matrix $E(\lambda; x, 0)$ satisfies a Riemann-Hilbert problem, as depicted in Figure \ref{ErrorE0}:

 \begin{itemize}

 \item{} $E(\lambda; x, 0)$ is analytic in $\lambda$ for $\lambda\in\mathbb{C}\setminus \left(B(\pm\eta_2, \pm\eta_0, \pm\eta_1)\cup \mathcal{C}_1^c\cup\mathcal{C}_2^c\right)$;

 \item{}It normalizes to the identity matrix $\mathbb{I}_2$ at infinity;

  \item{} For $\lambda\in B(\pm\eta_2, \pm\eta_0, \pm\eta_1)\cup \mathcal{C}_1^c\cup\mathcal{C}_2^c$, $E(\lambda; x, 0)$ admits continuous boundary values, and they are related by the following jump conditions
\begin{gather}\label{E0-jump}
E_+(\lambda; x, 0)=E_-(\lambda; x, 0)V_0^E,
\end{gather}
where the jump matrices are
\begin{gather}
V_0^E=
\begin{cases}
P^\infty (\lambda; x, 0)\mathcal{U}^{xp_0}_{f_0}\left[-ir^{-1}\right] \left(P^\infty (\lambda; x, 0)\right)^{-1}, &\mathrm{for} \,\,\, \lambda\in\mathcal{C}^c_1,  \\[0.5em]
P^\infty (\lambda; x, 0)\mathcal{L}^{xp_0}_{f_0}\left[-ir^{-1}\right] \left(P^\infty (\lambda; x, 0)\right)^{-1}, &\mathrm{for} \,\,\, \lambda\in\mathcal{C}^c_2,  \\[0.5em]
P^{\eta_j}(\lambda; x, 0)\left(P^\infty(\lambda; x, 0)\right)^{-1}, &\mathrm{for} \,\,\, \lambda\in\partial B\left(\eta_j\right), \,\, j=0,1,2, \\[0.5em]
\sigma_2P^{\eta_j}(-\lambda; x, 0)\left(P^\infty\left(-\lambda; x, 0\right)\right)^{-1}\sigma_2, &\mathrm{for} \,\,\, \lambda\in\partial B\left(-\eta_j\right),\,\, j=0,1,2.
\end{cases}
\end{gather}
\end{itemize}
\end{RH}
For the special case of the first type of generalized reflection coefficient $r=r_0$ with $\beta_0=0$, a slightly different Riemann-Hilbert problem arises, as depicted in Figure \ref{ErrorE01}.
$E(\lambda; x, 0)$ is analytic in $\lambda$ for $\lambda\in\mathbb{C}\setminus \left(B(\pm\eta_2, \pm\eta_1)\cup \mathcal{C}_1^c\cup\mathcal{C}_2^c\right)$, and normalizes to the identity matrix $\mathbb{I}_2$ at infinity. For $\lambda\in B(\pm\eta_2, \pm\eta_1)\cup \mathcal{C}_1^c\cup\mathcal{C}_2^c$, $E(\lambda; x, 0)$ admits continuous boundary values, and they are related by \eqref{E0-jump} with jump matrices
\begin{gather}
V_0^E=
\begin{cases}
P^\infty (\lambda; x, 0)\mathcal{U}^{xp_0}_{f_0}\left[-ir^{-1}\right] \left(P^\infty (\lambda; x, 0)\right)^{-1}, &\mathrm{for} \,\,\, \lambda\in\mathcal{C}^c_1,  \\[0.5em]
P^\infty (\lambda; x, 0)\mathcal{L}^{xp_0}_{f_0}\left[-ir^{-1}\right] \left(P^\infty (\lambda; x, 0)\right)^{-1}, &\mathrm{for} \,\,\, \lambda\in\mathcal{C}^c_2,  \\[0.5em]
P^{\eta_j}(\lambda; x, 0)\left(P^\infty(\lambda; x, 0)\right)^{-1}, &\mathrm{for} \,\,\, \lambda\in\partial B\left(\eta_j\right),  \,\,j=1,2, \\[0.5em]
\sigma_2P^{\eta_j}(-\lambda; x, 0)\left(P^\infty\left(-\lambda; x, 0\right)\right)^{-1}\sigma_2, &\mathrm{for} \,\,\, \lambda\in\partial B\left(-\eta_j\right),\,\, j=0,2.
\end{cases}
\end{gather}
Near each self-intersection point, $E$ exhibits the following local behavior:
\begin{gather}
E(\lambda; x, 0)=\mathcal{O}
\begin{pmatrix}
1 & 1  \\
1 & 1
\end{pmatrix},
\end{gather}
as $\lambda$ tends to each self-intersection  point.

\begin{proposition}[Small norm estimate]
For the parameters \(\beta_0, \beta_1, \beta_2 > -1\), as \( x \to -\infty \), the jump matrices \( V_0^E \) satisfy the following small norm estimates:
\begin{gather}\begin{gathered}
\left\|V_0^E-\mathbb{I}_2\right\|_{L^1\cap L^2\cap L^\infty \left(\mathcal{C}_1^c\cup \mathcal{C}_2^c\right)}=\mathcal{O}\left(\mathrm{e}^{-\mu_0\left|x\right|}\right),   \\[0.5em]
\left\|V_0^E-\mathbb{I}_2\right\|_{L^1\cap L^2\cap L^\infty \left(B\left(\pm\eta_2, \pm\eta_1, \pm\eta_0\right)\right)}=\mathcal{O}\left(\left|x\right|^{-1}\right),
\end{gathered}
\end{gather}
which imply:
\begin{gather}\label{E0-Estimate}
E(0; x, 0)=\mathbb{I}_2+\mathcal{O}\left(\left|x\right|^{-1}\right), \quad E^{[1]}(x, 0)=\mathcal{O}\left(\left|x\right|^{-1}\right), \quad x\to-\infty,
\end{gather}
where $E^{[1]}(x, 0)=\lim_{\lambda\to \infty}\lambda \left(E(\lambda; x, 0)-\mathbb{I}_2\right)$.
\end{proposition}
\begin{proof}
The proof follows standard arguments and is omitted here for brevity.
\end{proof}

\subsection{Large-$x$ asymptotics of the initial values $u(x, 0)$ as $x\to-\infty$}

It follows from \eqref{potential-formula} that the initial value $u(x, 0)$ can be recovered from $Y(\lambda; x, 0)$
\begin{gather}
\begin{aligned}
&\frac{\partial \,u(x, 0)}{\partial x}=4 \lim_{\lambda\to\infty} \lambda Y_{1, 2}(\lambda; x, 0), \\[0.5em]
&\cos u(x, 0)=1-2Y_{1, 2}(0; x, 0)^2, \\[0.5em]
&\sin u(x, 0)=-2Y_{1, 1}(0; x, 0)Y_{1, 2}(0; x, 0).
\end{aligned}
\end{gather}
Recalling the transformations $Y(\lambda; x, 0)\mapsto T(\lambda; x, 0)$ in \eqref{Conjugation-0},
$T(\lambda; x, 0)\mapsto S(\lambda; x, 0)$ in \eqref{T0-To-S0}, and
$S(\lambda; x, 0)\mapsto E(\lambda; x, 0)$ in \eqref{E0},
together with the estimate \eqref{E0-Estimate},
we obtain that
\begin{gather}
\begin{aligned}
&\frac{\partial \,u(x, 0)}{\partial x}=4 \lim_{\lambda\to\infty} \lambda P^\infty_{1, 2}(\lambda; x, 0)+\mathcal{O}\left(\left|x\right|^{-1}\right), \\[0.5em]
&\cos u(x, 0)=1-2\left(P^\infty_{\pm1, 2}(0; x, 0)\mathrm{e}^{xg_{0\pm}\left(0\right)}f_{0\pm}^{-1}\left(0\right)\right)^2+\mathcal{O}\left(\left|x\right|^{-1}\right), \\[0.5em]
&\sin u(x, 0)=-2P^\infty_{\pm 1, 1}(0; x, 0)P^\infty_{\pm 1, 2}(0; x, 0)+\mathcal{O}\left(\left|x\right|^{-1}\right),
\end{aligned}
\end{gather}
where two ways, $\pm$, leads to the same result. The fact is clear by noting the jump condition for \( P^\infty(\lambda; x, 0) \):
\begin{gather}
P^\infty_+(\lambda; x, 0)=P^\infty_-(\lambda; x, 0)\mathrm{e}^{\Delta^0_1\sigma_3}, \quad \mathrm{for}\,\,\, \lambda\in\left(-\eta_1, \eta_1\right),
\end{gather}
and using \( \mathrm{e}^{x g_{0\pm}(0)} f_{0\pm}^{-1} = \mathrm{e}^{\pm \Delta^0_1 / 2} \).
A straightforward calculation leads to \eqref{large-x-left} as stated in Theorem \ref{large-x}.

\section{Long-time asymptotics for the sine-Gordon kink-soliton gas $u(x, t)$}
In this section, we explore the long-time asymptotic behavior of the kink-soliton gas as \( t \to +\infty \). For the case where \( \xi > -\eta_2^{-2} \), the asymptotics follow directly from a standard small norm argument, leading to the result in \eqref{large-right}.

When \( \xi < -\eta_2^{-2} \), we consider two distinct cases: first, the cases with reflection coefficients \( r = r_0 \) and \( r_c \), where \( \beta_0 \neq 0 \), and second, the case with \( r = r_0 \) and \( \beta_0 = 0 \). Given that the derivation for \( \beta_0 = 0 \) is straightforward and follows similarly from the more general case \( \beta_0 \neq 0 \), we will omit the details for the special case and leave it as an exercise for the reader.

The analysis proceeds by focusing on the long-time asymptotic behavior within the intervals \( \xi \in \left( -\infty, \xi_\mathrm{crit} \right) \cup \left( \xi_\mathrm{crit}, \xi_0 \right) \cup \left( \xi_0, -\eta_2^{-2} \right) \).
To prepare for the asymptotic analysis, we apply a sequence of transformations:
\[
Y(\lambda; x, t) \mapsto T(\lambda; x, t) \mapsto S(\lambda; x, t) \mapsto E(\lambda; x, t),
\]
ensuring that \( E(\lambda; x, t) \) normalizes to the identity matrix \( \mathbb{I}_2 \) as \( \lambda \to \infty \) and that its jump matrices decay exponentially and uniformly towards \( \mathbb{I}_2 \).

Following the Deift-Zhou steepest descent method, we decompose the matrices to facilitate the contour deformation. Specifically, we utilize the triangular decompositions:
\begin{gather}
\begin{gathered}
\mathcal{L}^{t\theta}\left[\mathrm{i}r\right]=\mathcal{U}^{t\theta}\left[-\mathrm{i}r^{-1}\right]\left(\mathcal{L}^{t\theta}\left[\mathrm{i}r\right]+\mathcal{U}^{t\theta}\left[\mathrm{i}r^{-1}\right]-2\,\mathbb{I}_2\right)\mathcal{U}^{t\theta}\left[-\mathrm{i}r^{-1}\right],\\[0.5em]
\mathcal{U}^{t\theta}\left[\mathrm{i}r\right]=\mathcal{L}^{t\theta}\left[-\mathrm{i}r^{-1}\right]\left(\mathcal{L}^{t\theta}\left[\mathrm{i}r^{-1}\right]+\mathcal{U}^{t\theta}\left[\mathrm{i}r\right]-2\,\mathbb{I}_2\right)\mathcal{L}^{t\theta}\left[-\mathrm{i}r^{-1}\right].
\end{gathered}
\end{gather}

Due to the sign of \( \Re(\theta) \), as shown in Figure \ref{Sign}, exponential decay is not achieved on certain lenses. To address this, we first perform a conjugation using a suitable \( g \)-function that modifies the problem such that decay is maintained across the lenses. Additionally, we introduce a scalar \( f \)-function to further simplify the jump matrices to constants.

The rigorous construction of the local parametrices for the transformed problem depends on the local behavior near specific points, as discussed in prior sections. While these details are crucial for the complete analysis, we assume that the necessary local properties of \( T \), \( S \), and \( E \) have already been established. Therefore, we focus on the main asymptotic results without re-examining the local asymptotics in depth.

\subsection{Riemann-Hilbert problem for $T(\lambda; x, t)$}

The \( f \)-function \( f = f(\lambda; \xi) \) is analytic for \( \lambda \in \mathbb{C} \setminus [-\eta_2, \eta_2] \) and asymptotically normalizes to \( 1 \) as \( \lambda \to \infty \). The continuous boundary values \( f_\pm \) on the real line satisfy specific jump conditions, which depend on the parameter \( \xi \):
\begin{gather}
\left\{\begin{aligned}
&f_+(\lambda; \xi)f_-(\lambda; \xi) =r, && \mathrm{for}\,\,\, \lambda\in\left(\alpha, \eta_2\right),   \\[0.5em]
&f_+(\lambda; \xi)f_- (\lambda; \xi)=r^{-1},&& \mathrm{for}\,\,\, \lambda\in\left(-\eta_2, -\alpha\right),  \\[0.5em]
&f_+^{-1}(\lambda; \xi)f_-(\lambda; \xi)=\mathrm{e}^{\Omega^\alpha\phi^\alpha}, &&\mathrm{for}\,\,\, \lambda\in\left(-\alpha, \alpha\right),
\end{aligned}\right.
\end{gather}
if $\xi\in\left(\xi_0, -\eta_2^{-2}\right)$; while $\xi<\xi_0$, the continuous boundary values are related by
\begin{gather}
\left\{\begin{aligned}
&f_+(\lambda; \xi)f_-(\lambda; \xi) =r, && \mathrm{for}\,\,\, \lambda\in\left(\alpha, \eta_0\right)\cup\left(\eta_0, \eta_2\right),   \\[0.5em]
&f_+(\lambda; \xi)f_-(\lambda; \xi) =r^{-1},&& \mathrm{for}\,\,\, \lambda\in\left(-\eta_2, -\eta_0\right)\cup\left(-\eta_0, -\alpha\right),  \\[0.5em]
&f_+^{-1}(\lambda; \xi)f_-(\lambda; \xi)=\mathrm{e}^{\Omega^\alpha\phi^\alpha}, &&\mathrm{for}\,\,\, \lambda\in\left(-\alpha, \alpha\right).
\end{aligned}\right.
\end{gather}
To construct \( f(\lambda; \xi) \), we divide by \( R(\lambda; \xi) \) and take the logarithmic form. Using Plemelj’s formula, we have:
\begin{gather}
f(\lambda; \xi)=\exp\left\{\frac{R(\lambda; \xi)}{\pi \mathrm{i}}\left(\int_{\alpha}^{\eta_2}\frac{\log r(s)}{R_{+}\left(s; \xi\right)}\frac{\lambda}{s^2-\lambda^2}\mathrm{d}s
-\int_{0}^{\alpha}\frac{\Omega^\alpha\phi^\alpha}{R(s; \xi)}\frac{\lambda}{s^2-\lambda^2}\mathrm{d}s\right)\right\}.
\end{gather}
With \( f(\lambda; \xi) \) defined, we can proceed to conjugate \( Y(\lambda; x, t) \) to \( T(\lambda; x, t) \). This conjugation is achieved by:
\begin{gather}\label{Conjugation}
T(\lambda; x, t)=Y(\lambda; x, t)\mathrm{e}^{tg(\lambda; \xi)\sigma_3}f(\lambda; \xi)^{-\sigma_3}.
\end{gather}

\begin{RH}
$T(\lambda; x, t)$ is a $2\times 2$ matrix-valued function  that satisfies a Riemann-Hilbert problem.

\begin{itemize}
\item{} $T(\lambda; x, t)$ is analytic in $\lambda$ for $\lambda\in\mathbb{C}\setminus\left(\left[-\eta_2, \eta_2\right]\right)$;

\item{}It normalizes to the identity matrix $\mathbb{I}_2$ as $\lambda\to\infty$;

\item{} For  $\lambda\in\left(-\eta_2, \eta_2\right)\setminus\left\{\pm\eta_1, \pm\eta_0, \pm\alpha\right\}$, $T(\lambda; x, t)$ admits continuous boundary values, and they are related by the following jump conditions
\begin{gather}
T_+=T_-
\begin{cases}
\mathcal{U}^{tp_-}_{f_-}\left[-\mathrm{i}r^{-1}\right]\left(\mathrm{i}\sigma_1\right)\mathcal{U}^{tp_+}_{f_+}\left[-\mathrm{i}r^{-1}\right], &\mathrm{for}\,\,\, \lambda\in\left(\alpha, \eta_2\right),\\[0.5em]
\mathcal{L}^{tp_-}_{f_-}\left[-\mathrm{i}r^{-1}\right]\left(\mathrm{i}\sigma_1\right)\mathcal{L}^{tp_+}_{f_+}\left[-\mathrm{i}r^{-1}\right], &\mathrm{for}\,\,\, \lambda\in\left(-\eta_2, -\alpha\right),\\[0.5em]
f_-^{\sigma_3}\mathrm{e}^{tp_-\sigma_3}\mathcal{L}\left[\mathrm{i}r\right]\mathrm{e}^{-tp_+\sigma_3}f_+^{-\sigma_3}, &\mathrm{for}\,\,\, \lambda\in\left(\eta_1, \eta_0\right)\cup\left(\eta_0, \alpha\right),\\[0.5em]
f_-^{\sigma_3}\mathrm{e}^{tp_-\sigma_3}\mathcal{U}\left[\mathrm{i}r\right]\mathrm{e}^{-tp_+\sigma_3}f_+^{-\sigma_3}, &\mathrm{for}\,\,\, \lambda\in\left(-\alpha, -\eta_0\right)\cup\left(-\eta_0, -\eta_1\right),\\[0.5em]
\mathrm{e}^{\Delta^\alpha\sigma_3}, &\mathrm{for}\,\,\, \lambda\in\left(-\eta_1, \eta_1\right).
\end{cases}
\end{gather}
if $\xi\in\left(\xi_0, -\eta_2^{-2}\right)$; if $\xi\in\left(\xi_\mathrm{crit}, \xi_0\right)$, the continuous boundary values are related by
\begin{gather}
T_+=T_-
\begin{cases}
\mathcal{U}^{tp_-}_{f_-}\left[-\mathrm{i}r^{-1}\right]\left(\mathrm{i}\sigma_1\right)\mathcal{U}^{tp_+}_{f_+}\left[-\mathrm{i}r^{-1}\right], &\mathrm{for}\,\,\, \lambda\in\left(\alpha, \eta_0\right)\cup\left(\eta_0, \eta_2\right),\\[0.5em]
\mathcal{L}^{tp_-}_{f_-}\left[-\mathrm{i}r^{-1}\right]\left(\mathrm{i}\sigma_1\right)\mathcal{L}^{tp_+}_{f_+}\left[-\mathrm{i}r^{-1}\right], &\mathrm{for}\,\,\, \lambda\in\left(-\eta_2, -\eta_0\right)\cup\left(-\eta_0, -\alpha\right),\\[0.5em]
f_-^{\sigma_3}\mathrm{e}^{tp_-\sigma_3}\mathcal{L}\left[\mathrm{i}r\right]\mathrm{e}^{-tp_+\sigma_3}f_+^{-\sigma_3}, &\mathrm{for}\,\,\, \lambda\in\left(\eta_1, \alpha\right),\\[0.5em]
f_-^{\sigma_3}\mathrm{e}^{tp_-\sigma_3}\mathcal{U}\left[\mathrm{i}r\right]\mathrm{e}^{-tp_+\sigma_3}f_+^{-\sigma_3}, &\mathrm{for}\,\,\, \lambda\in\left(-\alpha,  -\eta_1\right),\\[0.5em]
\mathrm{e}^{\Delta^\alpha\sigma_3}, &\mathrm{for}\,\,\, \lambda\in\left(-\eta_1, \eta_1\right),
\end{cases}
\end{gather}
and if $\xi<\xi_\mathrm{crit}$, the continuous boundary values are related by
\begin{gather}
T_+=T_-
\begin{cases}
\mathcal{U}^{tp_-}_{f_-}\left[-\mathrm{i}r^{-1}\right]\left(\mathrm{i}\sigma_1\right)\mathcal{U}^{tp_+}_{f_+}\left[-\mathrm{i}r^{-1}\right], &\mathrm{for}\,\,\, \lambda\in\left(\eta_1, \eta_0\right)\cup\left(\eta_0, \eta_2\right),\\[0.5em]
\mathcal{L}^{tp_-}_{f_-}\left[-\mathrm{i}r^{-1}\right]\left(\mathrm{i}\sigma_1\right)\mathcal{L}^{tp_+}_{f_+}\left[-\mathrm{i}r^{-1}\right], &\mathrm{for}\,\,\, \lambda\in\left(-\eta_2, -\eta_0\right)\cup\left(-\eta_0, -\eta_1\right),\\[0.5em]
\mathrm{e}^{\Delta_1\sigma_3}, &\mathrm{for}\,\,\, \lambda\in\left(-\eta_1, \eta_1\right).
\end{cases}
\end{gather}
\end{itemize}
\end{RH}

\subsection{Riemann-Hilbert problem for $S(\lambda; x, t)$}
The transformation of the Riemann-Hilbert problem \( Y(\lambda; x, t) \) into a new Riemann-Hilbert problem \( S(\lambda; x, t) \) is achieved by performing contour deformations that involve the introduction of lens-shaped regions around specific intervals of the real line. The specifics of this transformation vary based on the region of the parameter \( \xi \), with three cases to consider: \( \xi \in (\xi_0, -\eta_2^{-2}) \), \( \xi \in (\xi_\mathrm{crit}, \xi_0) \), and \( \xi < \xi_\mathrm{crit} \).

For this region \(\xi \in (\xi_0, -\eta_2^{-2})\), the contour deformation involves the intervals \( (\alpha, \eta_2) \) and \( (-\eta_2, -\alpha) \). Lenses are opened symmetrically above and below each interval:
Domains \( \mathrm{D}^+_1 \) and \( \mathrm{D}^-_1 \) are lenses above and below \( (\alpha, \eta_2) \), respectively.
Domains \( \mathrm{D}^+_2 \) and \( \mathrm{D}^-_2 \) are lenses above and below \( (-\eta_2, -\alpha) \), respectively.
We define \( \mathrm{D}_1 = \mathrm{D}^+_1 \cup \mathrm{D}^-_1 \) and \( \mathrm{D}_2 = \mathrm{D}^+_2 \cup \mathrm{D}^-_2 \) as the full lens-shaped regions encompassing each interval.

In this interval \(\xi \in (\xi_\mathrm{crit}, \xi_0)\), the contour deformation is more intricate due to the presence of the intermediate point \( \eta_0 \). The intervals are divided into segments surrounding \( \eta_0 \):

Domains \( \mathrm{D}^+_{1, l} \) and \( \mathrm{D}^-_{1, l} \): Lenses are positioned above and below \( (\alpha, \eta_0) \).
Domains \( \mathrm{D}^+_{1, r} \) and \( \mathrm{D}^-_{1, r} \): Lenses are positioned above and below \( (\eta_0, \eta_2) \).
Domains \( \mathrm{D}^+_{2, l} \) and \( \mathrm{D}^-_{2, l} \): Lenses are positioned above and below \( (-\eta_2, -\eta_0) \).
Domains \( \mathrm{D}^+_{2, r} \) and \( \mathrm{D}^-_{2, r} \): Lenses are positioned above and below \( (-\eta_0, -\alpha) \).

The definitions \( \mathrm{D}_{j, l} = \mathrm{D}_{j, l}^+ \cup \mathrm{D}_{j, l}^- \) and \( \mathrm{D}_{j, r} = \mathrm{D}_{j, r}^+ \cup \mathrm{D}_{j, r}^- \) are used to group the left and right segments of the lenses for each interval \( j = 1, 2 \). The complete lens regions are then \( \mathrm{D}_j = \mathrm{D}_{j, l} \cup \mathrm{D}_{j, r} \), and \( \mathrm{D}_2 = \mathrm{D}_{2, l} \cup \mathrm{D}_{2, r} \).

In this scenario \(\xi < \xi_\mathrm{crit}\), the lens contours are as described in Figure \ref{Deformation-0}, which likely involves a standard lens structure similar to those seen in the previous cases.

Across all cases, we define the region \( \mathrm{D}_\mathrm{o} \) as the area outside all lenses:
\[
\mathrm{D}_\mathrm{o} = \mathbb{C} \setminus \overline{\mathrm{D}_1 \cup \mathrm{D}_2 \cup (-\eta_1, \eta_1)},
\]
with further subdivision as:
\[
\mathrm{D}_\mathrm{o} = \mathrm{D}_\mathrm{o}^+ \cup \mathrm{D}_\mathrm{o}^- \cup (\eta_2, +\infty) \cup (-\infty, -\eta_2),
\]
where \( \mathrm{D}_\mathrm{o}^+ \) and \( \mathrm{D}_\mathrm{o}^- \) refer to portions in the upper and lower half-planes, respectively. This comprehensive contour structure allows for the introduction of an appropriate \( g \)-function, which, along with further transformations, will facilitate the simplification of the Riemann-Hilbert problem and provide the necessary exponential decay properties in the respective regions.

\begin{figure}[!t]
\centering
\includegraphics[scale=0.35]{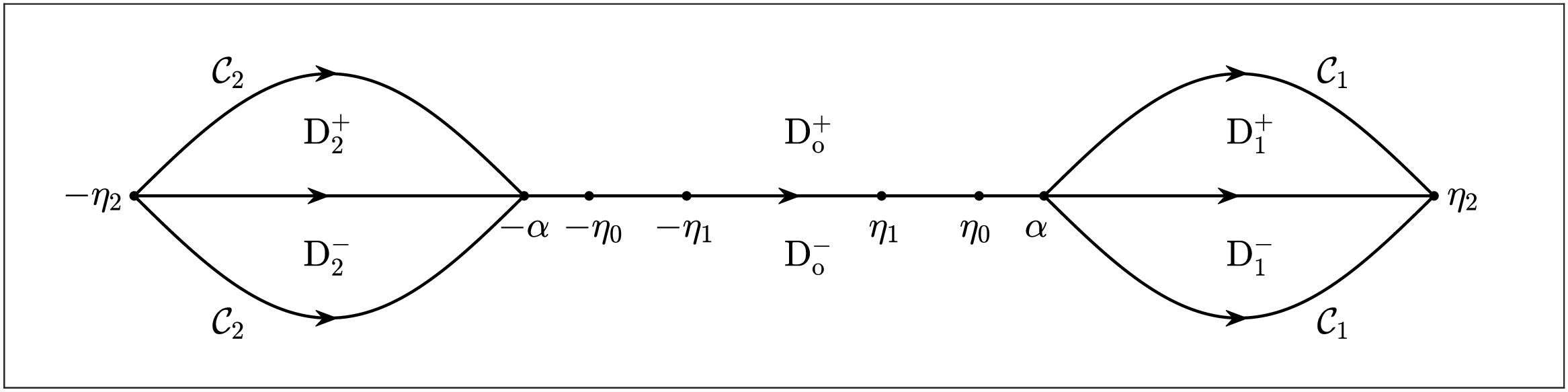}
\caption{Contour deformation in the region $\xi\in\left(\xi_0, -\eta_2^{-2}\right)$}
\label{Deformation-2}
\end{figure}

\begin{figure}[!t]
\vspace{0.1in}
\centering
\includegraphics[scale=0.35]{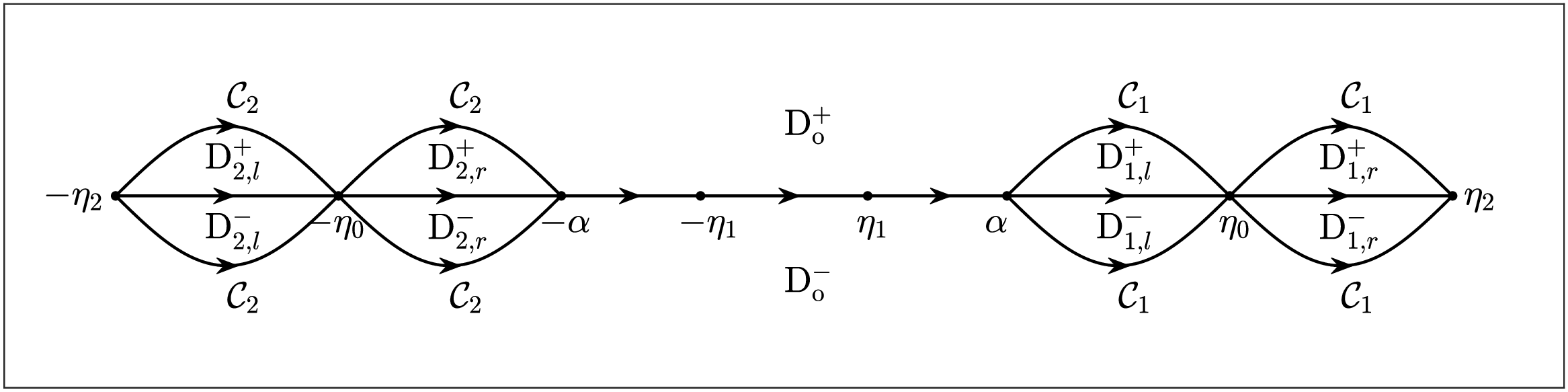}
\caption{Contour deformation in the region $\xi\in\left(\xi_\mathrm{crit}, \xi_0\right)$ }
\label{Deformation-3}
\end{figure}

The $2\times 2$ matrix-valued function $S(\lambda; x, t)$ is defined as follows:
\begin{gather}\label{T-To-S}
S(\lambda; x, t)=\begin{cases}
T(\lambda; x, t)\mathcal{U}^{tp}_{f}\left[-\mathrm{i}r^{-1}\right]^{\mp 1}, & \mathrm{for}\,\,\, \lambda\in\mathrm{D}_1^\pm,  \\[0.5em]
T(\lambda; x, t)\mathcal{L}^{tp}_{f}\left[-\mathrm{i}r^{-1}\right]^{\mp 1}, & \mathrm{for}\,\,\,  \lambda\in\mathrm{D}_2^\pm,  \\[0.5em]
T(\lambda; x, t), & \mathrm{for}\,\,\,  \lambda\in\mathrm{D}_\mathrm{o}.
\end{cases}
\end{gather}

\begin{RH} The function $S(\lambda; x, t)$ solves the following Riemann-Hilbert problem:

\begin{itemize}

\item{} $S(\lambda; x, t)$ is analytic in $\lambda$ for $\lambda\in\mathbb{C}\setminus\left(\left[-\eta_2, \eta_2\right]\cup\mathcal{C}_1\cup\mathcal{C}_2\right)$;

\item{} It normalizes to the identity matrix $\mathbb{I}_2$
as $\lambda\to\infty$;

\item{} For  $\lambda\in\left(-\eta_2, \eta_2\right)\cup\mathcal{C}_1\cup\mathcal{C}_2\setminus\left\{\pm\eta_1, \pm\eta_0\right\}$, $S(\lambda; x, 0)$ admits continuous boundary values  denoted by $S_+(\lambda; x, t)$ and $S_-(\lambda; x, t)$, respectively. These values are related by the following jump conditions: if $\xi\in\left(\xi_0, -\eta_2^{-2}\right)$, the jump condition is
\begin{gather}
S_+=S_-
\begin{cases}
\mathcal{U}^{tp}_{f}\left[-\mathrm{i}r^{-1}\right], &\mathrm{for}\,\,\, \lambda\in\mathcal{C}_1\setminus\left\{\alpha, \eta_2\right\},\\[0.5em]
\mathcal{L}^{tp}_{f}\left[-\mathrm{i}r^{-1}\right], &\mathrm{for}\,\,\, \lambda\in\mathcal{C}_2\setminus\left\{-\alpha, -\eta_2\right\},\\[0.5em]
\mathrm{i}\sigma_1, &\mathrm{for}\,\,\, \lambda\in\left(\alpha, \eta_2\right)\cup \left(-\eta_2, -\alpha\right),\\[0.5em]
f_-^{\sigma_3}\mathrm{e}^{tp_-\sigma_3}\mathcal{L}\left[-\mathrm{i}r\right]\mathrm{e}^{-tp_+\sigma_3}f_+^{-\sigma_3}, &\mathrm{for}\,\,\, \lambda\in\left(\eta_1, \eta_0\right)\cup\left(\eta_0, \alpha\right),\\[0.5em]
f_-^{\sigma_3}\mathrm{e}^{tp_-\sigma_3}\mathcal{U}\left[\mathrm{i}r\right]\mathrm{e}^{-tp_+\sigma_3}f_+^{-\sigma_3}, &\mathrm{for}\,\,\, \lambda\in\left(-\alpha, -\eta_0\right)\cup\left(-\eta_0, -\eta_1\right),\\[0.5em]
\mathrm{e}^{\Delta^\alpha\sigma_3}, &\mathrm{for}\,\,\, \lambda\in\left(-\eta_1, \eta_1\right);
\end{cases}
\end{gather}
if $\xi\in\left(\xi_\mathrm{crit}, \xi_0\right)$, the jump condition is
\begin{gather}
S_+=S_-
\begin{cases}
\mathcal{U}^{tp}_{f}\left[-\mathrm{i}r^{-1}\right], &\mathrm{for}\,\,\, \lambda\in\mathcal{C}_1\setminus\left\{\eta_0, \alpha, \eta_2\right\},\\[0.5em]
\mathcal{L}^{tp}_{f}\left[-\mathrm{i}r^{-1}\right], &\mathrm{for}\,\,\, \lambda\in\mathcal{C}_2\setminus\left\{-\eta_0, -\alpha, -\eta_2\right\},\\[0.5em]
\mathrm{i}\sigma_1, &\mathrm{for}\,\,\, \lambda\in\left(\alpha, \eta_2\right)\cup\left(-\eta_2, -\alpha\right)\setminus\{\pm\eta_0\},\\[0.5em]
f_-^{\sigma_3}\mathrm{e}^{tp_-\sigma_3}\mathcal{L}\left[-\mathrm{i}r\right]\mathrm{e}^{-tp_+\sigma_3}f_+^{-\sigma_3}, &\mathrm{for}\,\,\, \lambda\in\left(\eta_1, \alpha\right),\\[0.5em]
f_-^{\sigma_3}\mathrm{e}^{tp_-\sigma_3}\mathcal{U}\left[\mathrm{i}r\right]\mathrm{e}^{-tp_+\sigma_3}f_+^{-\sigma_3}, &\mathrm{for}\,\,\, \lambda\in\left(-\alpha, -\eta_1\right),\\[0.5em]
\mathrm{e}^{\Delta^\alpha\sigma_3}, &\mathrm{for}\,\,\, \lambda\in\left(-\eta_1, \eta_1\right);
\end{cases}
\end{gather}
and if $\xi<\xi_\mathrm{crit}$, the jump condition is
\begin{gather}
S_+=S_-
\begin{cases}
\mathcal{U}^{tp}_{f}\left[-\mathrm{i}r^{-1}\right], &\mathrm{for}\,\,\, \lambda\in\mathcal{C}_1\setminus\left\{\eta_0, \eta_1, \eta_2\right\},\\[0.5em]
\mathcal{L}^{tp}_{f}\left[-\mathrm{i}r^{-1}\right], &\mathrm{for}\,\,\, \lambda\in\mathcal{C}_2\setminus\left\{-\eta_0, -\eta_1, -\eta_2\right\},\\[0.5em]
\mathrm{i}\sigma_1, &\mathrm{for}\,\,\, \lambda\in\left(\eta_1, \eta_2\right)\cup\left(-\eta_2, -\eta_1\right)\setminus\{\pm\eta_0\},\\[0.5em]
\mathrm{e}^{\Delta_1\sigma_3}, &\mathrm{for}\,\,\, \lambda\in\left(-\eta_1, \eta_1\right).
\end{cases}
\end{gather}
\end{itemize}
\end{RH}

\subsection{Riemann-Hilbert problem for $E(\lambda; x, t)$}
We define the error matrix $E(\lambda; x, t)$ by
\begin{gather}\label{E}
E(\lambda; x, t)=S(\lambda; x, t)P(\lambda; x, t)^{-1},
\end{gather}
with the global parametrix $P(\lambda; x, t)$, which has different form in different regions.

\subsubsection{The region $\xi\in\left(\xi_0, -\eta_2^{-2}\right)$}
In the region $\xi\in\left(\xi_0, -\eta_2^{-2}\right)$, the global parametrix $P(\lambda; x, t)$ is constructed by
\begin{gather}
P(\lambda; x, t)=
\begin{cases}
P^\infty\left(\lambda; x, t\right), &\mathrm{for}\,\,\,\lambda\in\mathbb{C}\setminus\overline{B\left(\pm\eta_2, \pm\alpha\right)},  \\
P^{\eta_2}\left(\lambda; x, t\right), & \mathrm{for}\,\,\,\lambda\in B\left(\eta_2\right),  \\
P^{\alpha}\left(\lambda; x, t\right), & \mathrm{for}\,\,\,\lambda\in B\left(\alpha\right),  \\
\sigma_2P^{\eta_2}\left(-\lambda; x, t\right)\sigma_2, & \mathrm{for}\,\,\,\lambda\in B\left(-\eta_2\right),  \\
\sigma_2P^{\alpha}\left(-\lambda; x, t\right)\sigma_2, & \mathrm{for}\,\,\,\lambda\in B\left(-\alpha\right).
\end{cases}
\end{gather}
The outer parametrix $P^\infty(\lambda; x, t)=\left(P_{i, j}^\infty(\lambda; x, t)\right)_{2\times 2}$ is formulated as
\begin{gather}\label{out-alpha}
\begin{aligned}
P^\infty_{1, 1}(\lambda; x, t)&=\frac{\delta_\alpha+\delta_\alpha^{-1}}{2}
\frac{\vartheta_3\left(w_\alpha+\frac{1}{4}+\frac{\Delta^\alpha}{2\pi\mathrm{i}}; \tau_\alpha\right)}{\vartheta_3\left(w_\alpha+\frac{1}{4}; \tau_\alpha\right)}
\frac{\vartheta_3\left(0; \tau_\alpha\right)}{\vartheta_3\left(\frac{\Delta^\alpha}{2\pi\mathrm{i}}; \tau_\alpha\right)}, \\[0.5em]
P^\infty_{1, 2}(\lambda; x, t)&=\frac{\delta_\alpha-\delta_\alpha^{-1}}{2}
\frac{\vartheta_3\left(w_\alpha-\frac{1}{4}-\frac{\Delta^\alpha}{2\pi\mathrm{i}}; \tau_\alpha\right)}{\vartheta_3\left(w_\alpha-\frac{1}{4}; \tau_\alpha\right)}
\frac{\vartheta_3\left(0; \tau_\alpha\right)}{\vartheta_3\left(\frac{\Delta^\alpha}{2\pi\mathrm{i}}; \tau_\alpha\right)},\\[0.5em]
P^\infty_{2, 1}(\lambda; x, t)&=\frac{\delta_\alpha-\delta_\alpha^{-1}}{2}
\frac{\vartheta_3\left(w_\alpha-\frac{1}{4}+\frac{\Delta^\alpha}{2\pi\mathrm{i}}; \tau_\alpha\right)}{\vartheta_3\left(w_\alpha-\frac{1}{4}; \tau_\alpha\right)}
\frac{\vartheta_3\left(0; \tau_\alpha\right)}{\vartheta_3\left(\frac{\Delta^\alpha}{2\pi\mathrm{i}}; \tau_\alpha\right)}, \\[0.5em]
P^\infty_{2, 2}(\lambda; x, t)&=\frac{\delta_\alpha+\delta_\alpha^{-1}}{2}
\frac{\vartheta_3\left(w_\alpha+\frac{1}{4}-\frac{\Delta^\alpha}{2\pi\mathrm{i}}; \tau_\alpha\right)}{\vartheta_3\left(w_\alpha+\frac{1}{4}; \tau_\alpha\right)}
\frac{\vartheta_3\left(0; \tau_\alpha\right)}{\vartheta_3\left(\frac{\Delta^\alpha}{2\pi\mathrm{i}}; \tau_\alpha\right)},
\end{aligned}
\end{gather}
where $\delta_\alpha=\delta_\alpha(\lambda)$ is a branch of $\left((\lambda+\alpha)(\lambda-\eta_2)/(\lambda+\eta_2)(\lambda-\alpha)\right)^{1/4}$, with the branch cut $\left[\alpha, \eta_2\right]\cup\left[-\eta_2, -\alpha\right]$, normalizing to $\delta_\alpha=1+\mathcal{O}\left(\lambda^{-1}\right)$ as $\lambda\to \infty$, and $w_\alpha=w_\alpha(\lambda)$ is defined by
\begin{gather}
w_\alpha=-\frac{\eta_2}{4eK\left(m_\alpha\right)}\int_{\eta_2}^\lambda \frac{\mathrm{d}\zeta}{R(\zeta; \xi)}.
\end{gather}
The local parametrix $P^{\eta_2}(\lambda; x, t)$ for $\lambda\in\left(\mathrm{D}_1\cup\mathrm{D}_\mathrm{o}\right)\cap B(\eta_2)$,  is expressed as
\begin{gather}\label{parametrix-eta2}
P^{\eta_2}(\lambda; x, t)=P^\infty(\lambda; x, t)A^{\eta_2}C\zeta_{\eta_2}^{-\sigma_3/4}M^{\mathrm{mB}}\left(\zeta_{\eta_2}; \beta_2\right)\mathrm{e}^{-\sqrt{\zeta_{\eta_2}}\sigma_3}\left(A^{\eta_2}\right)^{-1},
\end{gather}
where $\zeta_{\eta_2}=t^2\mathcal{F}^{\eta_2}$, $A^{\eta_2}=\left(\mathrm{e}^{\pi\mathrm{i}/4}/fd^{\eta_2}\right)^{\sigma_3}\sigma_2$, and $\mathcal{F}^{\eta_2}=p(\lambda; \xi)^2$ is the conformal mapping.
The local parametrix $P^{\alpha}(\lambda; x, t)$ for $\lambda\in\left(\mathrm{D}_\mathrm{o}^\pm\cup\mathrm{D}_1^\pm\right)\cap B(\alpha)$, is expressed as follows
\begin{gather}\label{parametrix-alpha}
P^{\alpha}(\lambda; x, t)=P^\infty(\lambda; x, t)A_{\pm}^{\alpha}C\zeta_{\alpha}^{-\sigma_3/4}M^{\mathrm{Ai}}(\zeta_{\alpha})\mathrm{e}^{\zeta_\alpha^{3/2}\sigma_3}\left(A_{\pm}^{\alpha}\right)^{-1},
\end{gather}
where
$\zeta_{\alpha}=t^{2/3}\mathcal{F}_\pm^\alpha$, with the conformal mapping $\mathcal{F}^\alpha=\left(p\pm  \Omega_\xi^\alpha /2\right)^{2/3}$, $\Omega_\xi^\alpha=\Omega^\alpha\left(\xi+1/\alpha\eta_2\right)/4$,  and
$A_{\pm}^{\alpha}=\left(\mathrm{e}^{\pi\mathrm{i}/4\mp t\Omega_\xi^\alpha /2}/f\sqrt{r}\right)^{\sigma_3}\sigma_1$.

\begin{RH} As depicted in Figure \ref{ErrorE1}, the error matrix $E(\lambda; x, t)$ satisfies the following Riemann-Hilbert problem:

\begin{itemize}

\item{} $E(\lambda; x, t)$ is analytic in $\lambda$ for $\lambda\in\mathbb{C}\setminus\mathcal{C}_1^c\cup\mathcal{C}_2^c\cup\partial B(\pm\alpha, \pm\eta_2)\cup\left[\eta_1, \alpha\right]\cup\left[-\alpha, -\eta_1\right]$,

\item{}It normalizes to the identity matrix $\mathbb{I}_2$,

\item{} The jump condition is
\begin{gather}\label{E-Jump}
E_+(\lambda; x, t)=E_-(\lambda; x, t) V^E,
\end{gather}
where the jump matrix $V^E$ is
\begin{gather}
V^E=
\begin{cases}
P^\infty (\lambda)\mathcal{U}^{tp}_{f}\left[-ir^{-1}\right] \left(P^\infty (\lambda)\right)^{-1}, &\mathrm{for} \,\,\, \lambda\in\mathcal{C}^c_1,  \\[0.5em]
P^\infty (\lambda)\mathcal{L}^{tp}_{f}\left[-ir^{-1}\right] \left(P^\infty (\lambda)\right)^{-1}, &\mathrm{for} \,\,\, \lambda\in\mathcal{C}^c_2,  \\[0.5em]
P^{\eta_2}(\lambda)\left(P^\infty(\lambda)\right)^{-1}, &\mathrm{for} \,\,\, \lambda\in\partial B\left(\eta_2\right),  \\[0.5em]
P^{\alpha}(\lambda)\left(P^\infty(\lambda)\right)^{-1}, &\mathrm{for} \,\,\, \lambda\in\partial B\left(\alpha\right),   \\[0.5em]
\sigma_2P^{\eta_2}(-\lambda)\left(P^\infty\left(-\lambda\right)\right)^{-1}\sigma_2, &\mathrm{for} \,\,\, \lambda\in\partial B\left(-\eta_2\right), \\[0.5em]
\sigma_2P^{\alpha}\left(-\lambda\right)\left(P^\infty\left(-\lambda\right)\right)^{-1}\sigma_2, &\mathrm{for} \,\,\, \lambda\in\partial B\left(-\alpha\right),\\[0.5em]
\mathbb{I}_2+V_1^E, &\mathrm{for} \,\,\, \lambda\in\left(\eta_1, \eta_0\right)\cup\left(\eta_0, \alpha\right)\setminus \overline{B(\alpha)}, \\[0.5em]
\mathbb{I}_2+V_2^E, &\mathrm{for} \,\,\, \lambda\in\left(-\alpha, -\eta_0\right)\cup\left(-\eta_0, -\eta_1\right)\setminus \overline{B(-\alpha)},
\end{cases}
\end{gather}
with  $V_1^E$ and $V_2^E$ being
\begin{gather}
\begin{aligned}
V^E_1&=P^\infty_-f_-^{\sigma_3}\mathrm{e}^{tp_-\sigma_3}\left(\mathcal{L}\left[\mathrm{i}r\right]
-\mathbb{I}_2\right)\mathrm{e}^{-tp_+\sigma_3}f_+^{-\sigma_3}\left(P^\infty_+\right)^{-1}, \\[0.5em]
V^E_2&=P^\infty_-f_-^{\sigma_3}\mathrm{e}^{tp_-\sigma_3}\left(\mathcal{U}\left[\mathrm{i}r\right]-\mathbb{I}_2\right)\mathrm{e}^{-tp_+\sigma_3}f_+^{-\sigma_3}\left(P^\infty_+\right)^{-1}.
\end{aligned}
\end{gather}
\end{itemize}
\end{RH}

\begin{proposition}[Small norm estimate in the region $\xi\in\left(\xi_0, -\eta_2^{-2}\right)$]
For $\beta_2>-1, \beta_1\ge 0, \beta_0\ge 0$,
the jump matrices $V^E$ has the following small norm estimates:
\begin{gather}
\begin{aligned}
&\left\|V^E-\mathbb{I}_2\right\|_{L^1\cap L^2\cap L^\infty \left(\mathcal{C}^c\right)}=\mathcal{O}\left(\mathrm{e}^{-\mu t}\right), && \mathrm{as}\,\,\, t\to +\infty,  \\[0.5em]
&\left\|V^E-\mathbb{I}_2\right\|_{L^1\cap L^2\cap L^\infty \left(B\left(\pm\eta_2, \pm\alpha\right)\right)}=\mathcal{O}\left(t^{-1}\right), && \mathrm{as}\,\,\, t\to +\infty,
\end{aligned}
\end{gather}
which leads to
\begin{gather}\label{E-Estimate-right}
E(0; x, t)=\mathbb{I}_2+\mathcal{O}\left(t^{-1}\right), \quad E^{[1]}(x, t)=\mathcal{O}\left(t^{-1}\right), \quad \mathrm{as}\,\,\, t\to +\infty,
\end{gather}
where $E^{[1]}(x, t)=\lim_{\lambda\to\infty}\lambda\left(E(\lambda; x, 0)-\mathbb{I}_2\right)$, and
$\mathcal{C}^c=\mathcal{C}_1^c\cup\mathcal{C}_2^c\cup\left[\eta_1, \alpha \right]\cup\left[-\alpha, -\eta_1\right]\setminus B(\pm\alpha)$.
\end{proposition}
\begin{proof}
It follows from Proposition  \ref{p-sign} and the symmetry relation
\begin{gather}
p(\lambda; \xi)=p(\lambda^*; \xi)^*=-p(-\lambda; \xi), \quad \mathrm{for}\,\,\, \lambda\in\mathbb{C}\setminus\left[-\eta_2, \eta_2\right]
\end{gather}
that $\Re(p)<0$ for $\lambda\in \mathcal{C}_1\setminus \left\{\alpha, \eta_2\right\}$,
 $\Re(p)>0$ for $\lambda\in \mathcal{C}_2\setminus \left\{-\alpha, -\eta_2\right\}$,
 $p_++p_->0$, for $\lambda\in\left[\eta_1, \alpha\right)$, and
 $p_++p_-<0$, for $\lambda\in\left(-\alpha, -\eta_1\right]$,
which give rise to the estimate of the jump matrix on $\mathcal{C}^c$.
The estimate on $B\left(\pm\eta_2, \pm\alpha\right)$ is derived by
\begin{gather}
\begin{aligned}
&P^{\eta_2}(\lambda)\left(P^\infty(\lambda)\right)^{-1}=\mathbb{I}_2+\mathcal{O}\left(t^{-1}\right), && \mathrm{for}\,\,\, \lambda\in \partial B(\eta_2),\\[0.5em]
&P^{\alpha}(\lambda)\left(P^\infty(\lambda)\right)^{-1}=\mathbb{I}_2+\mathcal{O}\left(t^{-1}\right), && \mathrm{for}\,\,\, \lambda\in \partial B(\alpha).
\end{aligned}
\end{gather}
This completes the proof.
\end{proof}

\begin{figure}[!t]
\centering
\includegraphics[scale=0.38]{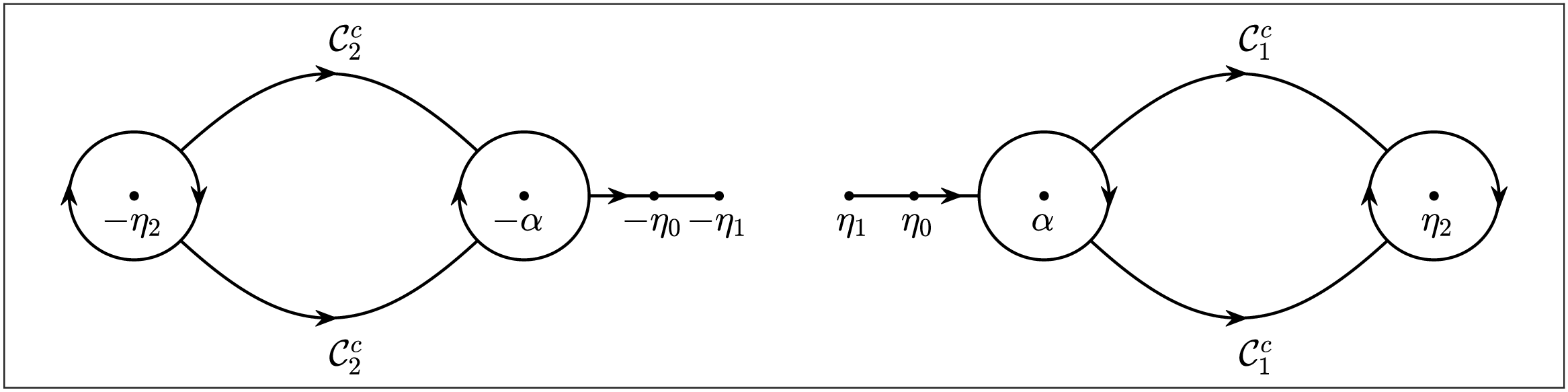}
\caption{Jump contour of the error matrix $E(\lambda; x, t)$ in the region $\xi_0<\xi<-\eta_2^{-2}$}
\label{ErrorE1}
\end{figure}

\subsubsection{The region $\xi\in\left(\xi_\mathrm{crit}, \xi_0\right)$}
In the region $\xi\in\left(\xi_\mathrm{crit}, \xi_0\right)$, the global parametrix $P(\lambda; x, t)$ is constructed by
\begin{gather}
P(\lambda; x, t)=
\begin{cases}
P^\infty\left(\lambda; x, t\right), &\mathrm{for}\,\,\,\lambda\in\mathbb{C}\setminus\overline{B\left(\pm\eta_2, \pm\alpha, \pm \eta_0\right)},  \\[0.5em]
P^{\eta_j}\left(\lambda; x, t\right), & \mathrm{for}\,\,\,\lambda\in B\left(\eta_j\right), \,\, j=0,2, \\[0.5em]
P^{\alpha}\left(\lambda; x, t\right), & \mathrm{for}\,\,\,\lambda\in B\left(\alpha\right),  \\[0.5em]
\sigma_2P^{\eta_j}\left(-\lambda; x, t\right)\sigma_2, & \mathrm{for}\,\,\,\lambda\in B\left(-\eta_j\right),\,\, j=0,2,  \\[0.5em]
\sigma_2P^{\alpha}\left(-\lambda; x, t\right)\sigma_2, & \mathrm{for}\,\,\,\lambda\in B\left(-\alpha\right).\\[0.5em]
\end{cases}
\end{gather}
The outer parametrix $P^\infty(\lambda; x, t)$ is still defined in \eqref{out-alpha}, the local parametrix $P^{\eta_2}(\lambda; x, t)$ in \eqref{parametrix-eta2}, and the local parametrix $P^{\alpha}(\lambda; x, t)$ in \eqref{parametrix-alpha}.
Below, we give the local parametrix $P^{\eta_0}(\lambda; x, t)$.
For the first type of  generalized reflection coefficient  $r=r_0$ with $\beta_0\ne 0$, the conformal mapping is defined as follows:
$\mathcal{F}^{\eta_0}=\pm(p-p_\pm(\eta_0))$ if $\lambda\in B(\eta_0)\cap\mathbb{C}^\pm$.
 Similar to the large-$x$ asymptotics of the initial value, let us define the domains
 $\mathrm{D}_{\mathrm{o}, 1, r}^+=\left(\mathcal{F}^{\eta_0}\right)^{-1} \left(\mathrm{D}^\zeta_2\cap B^\zeta(0)\right)$,
 $\mathrm{D}_{\mathrm{o}, 1, l}^+=\left(\mathcal{F}^{\eta_0}\right)^{-1} \left(\mathrm{D}^\zeta_3\cap B^\zeta(0)\right)$,
 $\mathrm{D}_{\mathrm{o}, 1, l}^-=\left(\mathcal{F}^{\eta_0}\right)^{-1} \left(\mathrm{D}^\zeta_6\cap B^\zeta(0)\right)$,
and
 $\mathrm{D}_{\mathrm{o}, 1, r}^-=\left(\mathcal{F}^{\eta_0}\right)^{-1} \left(\mathrm{D}^\zeta_7\cap B^\zeta(0)\right)$.
The local parametrix $P^{\eta_0}(\lambda; x, t)$ in the neighborhood of $\lambda=\eta_0$ is constructed as follows:
For $\lambda\in\left(\mathrm{D}_{1, r}^+\cup\mathrm{D}_{\mathrm{o}, 1, r}^+\right)\cap B(\eta_0)$,
\begin{gather}\label{eta0-1}
P^{\eta_0}(\lambda; x, t)=P^\infty(\lambda; x, t) A_{r+}^{\eta_0}\mathrm{e}^{\beta_0\pi\mathrm{i}\sigma_3/4}\left(-\mathrm{i}\sigma_2\right)M^{\mathrm{mb}}(\zeta_{\eta_0}; \beta_0)\mathrm{e}^{\zeta_{\eta_0}\sigma_3}\left(A_{r+}^{\eta_0}\right)^{-1},
\end{gather}
with $A_{r+}^{\eta_0}=\left(\mathrm{e}^{\pi\mathrm{i}/4+tp_{+}(\eta_0)}/fd_r^{\eta_0}\right)^{\sigma_3}\sigma_2$;
For $\lambda\in \left(\mathrm{D}_{1, l}^+\cup\mathrm{D}_{\mathrm{o}, 1, l}^+\right)\cap B(\eta_0)$, $P^{\eta_0}$ is formulated as
\begin{gather}
P^{\eta_0}(\lambda; x, t)=P^\infty (\lambda; x, t)A_{l+}^{\eta_0}\mathrm{e}^{-\beta_0\pi\mathrm{i}\sigma_3/4}\left(-\mathrm{i}\sigma_2\right)M^{\mathrm{mb}}(\zeta_{\eta_0}; \beta_0)\mathrm{e}^{\zeta_{\eta_0}\sigma_3}\left(A_{l+}^{\eta_0}\right)^{-1},
\end{gather}
with $A_{l+}^{\eta_0}=\left(\mathrm{e}^{\pi\mathrm{i}/4+tp_{+}(\eta_0)}/fd_l^{\eta_0}\right)^{\sigma_3}\sigma_2$;
For $\lambda\in\left(\mathrm{D}_{1, l}^-\cup\mathrm{D}_{\mathrm{o}, 1, l}^-\right)\cap B(\eta_0)$,  $P^{\eta_0}$ is written as
\begin{gather}
P^{\eta_0}(\lambda; x, t)=P^\infty(\lambda; x, t) A_{l-}^{\eta_0}\mathrm{e}^{\beta_0\pi\mathrm{i}\sigma_3/4}M^{\mathrm{mb}}(\zeta_{\eta_0}; \beta_0)\mathrm{e}^{-\zeta_{\eta_0}\sigma_3}\left(A_{l-}^{\eta_0}\right)^{-1},
\end{gather}
with $A_{l-}^{\eta_t}=\left(\mathrm{e}^{\pi\mathrm{i}/4+tp_{-}(\eta_0)}/fd_l^{\eta_0}\right)^{\sigma_3}\sigma_2$;
For $\lambda\in \left(\mathrm{D}_{1, r}^-\cup\mathrm{D}_{\mathrm{o}, 1, r}^-\right)\cap B(\eta_0)$, $P^{\eta_0}$ is expressed as
\begin{gather}
P^{\eta_0}(\lambda; x, t)=P^\infty(\lambda; x, t) A_{r-}^{\eta_0} \mathrm{e}^{-\beta_0\pi\mathrm{i}\sigma_3/4}M^{\mathrm{mb}}(\zeta_{\eta_0}; \beta_0)\mathrm{e}^{-\zeta_{\eta_0}\sigma_3}\left(A_{r-}^{\eta_0}\right)^{-1},
\end{gather}
with $A_{r-}^{\eta_0}=\left(\mathrm{e}^{\pi\mathrm{i}/4+tp_{-}(\eta_0)}/fd_r^{\eta_0}\right)^{\sigma_3}\sigma_2$,
and $\zeta_{\eta_0}=t\mathcal{F}^{\eta_0}$.

For the second generalized reflection coefficient  $r=r_c$,
the conformal mapping is defined as follows:
$\mathcal{F}_0^{\eta_0}=\pm 2(p-p_{\pm}(\eta_0))$ if $\lambda\in B(\eta_0)\cap\mathbb{C}^\pm$.
The local parametrix  in the neighborhood of $\lambda=\eta_0$ is constructed as follows:
For $\lambda\in \left(\mathrm{D}_1^+\cup \mathrm{D}_\mathrm{o}^+\right)\cap B(\eta_0)$, the local parametrix is formulated as
\begin{gather}
P^{\eta_0}(\lambda; x, t)=P^\infty (\lambda; x, t) A^{\eta_0}_{+}\left(\zeta_{\eta_0}^{\kappa_0\sigma_3}\mathrm{i}\sigma_2\mathrm{e}^{\kappa_0\pi\mathrm{i}\sigma_3}\right)^{-1}M^{\mathrm{CH}}(\zeta_{\eta_0}; \kappa_0)\mathrm{e}^{\zeta_{\eta_0}\sigma_3/2}\left(A^{\eta_0}_{+}\right)^{-1};
\end{gather}
For $\lambda\in \left(\mathrm{D}_1^-\cup \mathrm{D}_\mathrm{o}^-\right)\cap B(\eta_0)$, the local parametrix is expressed as
\begin{gather}\label{eta0-2}
P^{\eta_0}(\lambda; x, t)=P^\infty (\lambda; x, t)A^{\eta_0}_{-}\mathrm{e}^{-\kappa_0\pi\mathrm{i}\sigma_3}M^{\mathrm{CH}}(\zeta_{\eta_0};  \kappa_0)\mathrm{e}^{-\zeta_{\eta_0}\sigma_3/2} \left(A^{\eta_0}_{-}\right)^{-1},
\end{gather}
where  $A^{\eta_0}_{\pm}=\left(\mathrm{e}^{\pi\mathrm{i}/4+tp_{\pm}(\eta_0)}/fd^{\eta_0}\right)^{\sigma_3}\sigma_2$,
and $\zeta_{\eta_0}=t\mathcal{F}^{\eta_0}$.

\begin{RH}As depicted in Figure \ref{ErrorE2}, the error matrix $E(\lambda; x, t)$ satisfies the following Riemann-Hilbert problem:

\begin{itemize}

\item{}$E(\lambda; x, t)$ is analytic in $\lambda$ for $\lambda\in\mathbb{C}\setminus\mathcal{C}_1^c\cup\mathcal{C}_2^c\cup\partial B(\pm\alpha, \pm\eta_2)\cup\left[\eta_1, \alpha\right]\cup\left[-\alpha, -\eta_1\right]$,

\item{}It normalizes to the identity matrix $\mathbb{I}_2$,

\item{} The jump condition is described in \eqref{E-Jump} with the jump matrix
\begin{gather}
V^E=
\begin{cases}
P^\infty (\lambda)\mathcal{U}^{tp}_{f}\left[-ir^{-1}\right] \left(P^\infty (\lambda)\right)^{-1}, &\mathrm{for} \,\,\, \lambda\in\mathcal{C}^c_1,  \\[0.5em]
P^\infty (\lambda)\mathcal{L}^{tp}_{f}\left[-ir^{-1}\right] \left(P^\infty (\lambda)\right)^{-1}, &\mathrm{for} \,\,\, \lambda\in\mathcal{C}^c_2,  \\[0.5em]
P^{\eta_j}(\lambda)\left(P^\infty(\lambda)\right)^{-1}, &\mathrm{for} \,\,\, \lambda\in\partial B\left(\eta_2\right), \,\, j=0,2, \\[0.5em]
P^{\alpha}(\lambda)\left(P^\infty(\lambda)\right)^{-1}, &\mathrm{for} \,\,\, \lambda\in\partial B\left(\alpha\right),   \\[0.5em]
\sigma_2P^{\eta_j}(-\lambda)\left(P^\infty\left(-\lambda\right)\right)^{-1}\sigma_2, &\mathrm{for} \,\,\, \lambda\in\partial B\left(-\eta_j\right),\,\, j=0,2, \\[0.5em]
\sigma_2P^{\alpha}\left(-\lambda\right)\left(P^\infty\left(-\lambda\right)\right)^{-1}\sigma_2, &\mathrm{for} \,\,\, \lambda\in\partial B\left(-\alpha\right),\\[0.5em]
\mathbb{I}_2+V_1^E, &\mathrm{for} \,\,\, \lambda\in\left(\eta_1, \alpha\right)\setminus \overline{B(\alpha)},\\[0.5em]
\mathbb{I}_2+V_2^E, &\mathrm{for} \,\,\, \lambda\in\left(-\alpha,  -\eta_1\right)\setminus \overline{B(-\alpha)}.
\end{cases}
\end{gather}
\end{itemize}
\end{RH}

\begin{proposition}[Small norm estimate in the region $\xi\in\left(\xi_\mathrm{crit}, \xi_0\right)$]
For $\beta_2>-1, \beta_0>-1, \beta_1\ge 0$,
the jump matrices $V^E$ has the following small norm estimates:
\begin{gather}
\begin{aligned}
&\left\|V^E-\mathbb{I}_2\right\|_{L^1\cap L^2\cap L^\infty \left(\mathcal{C}^c\right)}=\mathcal{O}\left(\mathrm{e}^{-\mu t}\right), && \mathrm{as}\,\,\, t\to +\infty  \\[0.5em]
&\left\|V^E-\mathbb{I}_2\right\|_{L^1\cap L^2\cap L^\infty \left(B\left(\pm\eta_2, \pm\eta_0, \pm\alpha\right)\right)}=\mathcal{O}\left(t^{-1}\right), && \mathrm{as}\,\,\, t\to +\infty
\end{aligned}
\end{gather}
which leads to
\begin{gather}\label{E-Estimate-middle}
E(0; x, t)=\mathbb{I}_2+\mathcal{O}\left(t^{-1}\right), \quad E^{[1]}(x, t)=\mathcal{O}\left(t^{-1}\right), \quad \mathrm{as}\,\,\, t\to +\infty,
\end{gather}
where $E^{[1]}(x, t)=\lim_{\lambda\to\infty}\lambda\left(E(\lambda; x, 0)-\mathbb{I}_2\right)$, and
$\mathcal{C}^c=\mathcal{C}_1^c\cup\mathcal{C}_2^c\cup\left[\eta_1, \alpha \right]\cup\left[-\alpha, -\eta_1\right]\setminus B(\pm\alpha)$.
\end{proposition}
\begin{proof}
Again, by Proposition  \ref{p-sign} and the symmetry relation
\begin{gather}
p(\lambda; \xi)=p(\lambda^*; \xi)^*=-p(-\lambda; \xi), \quad \mathrm{for}\,\,\, \lambda\in\mathbb{C}\setminus\left[-\eta_2, \eta_2\right],
\end{gather}
we obtain that
$\Re(p)<0$ for $\lambda\in \mathcal{C}_1\setminus \left\{\alpha, \eta_0, \eta_2\right\}$,
 $\Re(p)>0$ for $\lambda\in \mathcal{C}_2\setminus \left\{-\alpha, -\eta_0, -\eta_2\right\}$,
 $p_++p_->0$, for $\lambda\in\left[\eta_1, \alpha\right)$, and
 $p_++p_-<0$, for $\lambda\in\left(-\alpha, -\eta_1\right]$,
which give rise to the estimate of the jump matrix on $\mathcal{C}^c$.
The estimate on $(B\left(\pm\eta_2, \pm\eta_0, \pm\alpha\right)$ is derived by
\begin{gather}
\begin{aligned}
&P^{\eta_j}(\lambda)\left(P^\infty(\lambda)\right)^{-1}=\mathbb{I}_2+\mathcal{O}\left(t^{-1}\right), && \mathrm{for}\,\,\, \lambda\in \partial B(\eta_j),\,\, j=0,2,\\[0.5em]
&P^{\alpha}(\lambda)\left(P^\infty(\lambda)\right)^{-1}=\mathbb{I}_2+\mathcal{O}\left(t^{-1}\right), && \mathrm{for}\,\,\, \lambda\in \partial B(\alpha). \\[0.5em]
\end{aligned}
\end{gather}
Then we complete the proof.
\end{proof}

\begin{figure}[!t]
\centering
\includegraphics[scale=0.38]{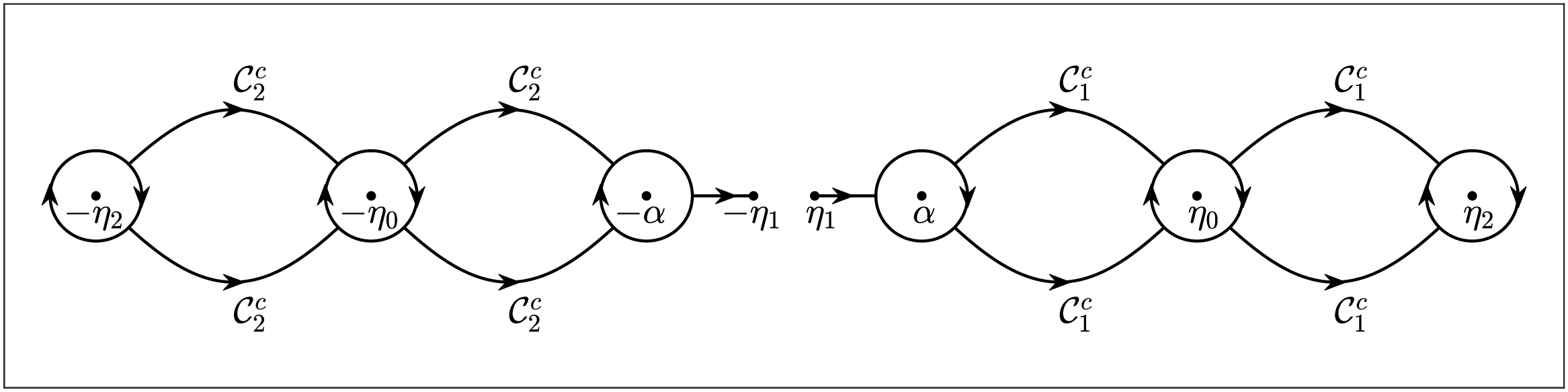}
\caption{Jump contour of the error matrix $E(\lambda; x, t)$ in the region $\xi_\mathrm{crit}<\xi<\xi_0$}
\label{ErrorE2}
\end{figure}

\subsubsection{The region $\xi<\xi_\mathrm{crit}$}

In the region $\xi<\xi_\mathrm{crit}$, the global parametrix $P(\lambda; x, t)$ is constructed by
\begin{gather}
P(\lambda; x, t)=
\begin{cases}
P^\infty\left(\lambda; x, t\right), &\mathrm{for}\,\,\,\lambda\in\mathbb{C}\setminus\overline{B\left(\pm\eta_2, \pm\eta_1, \pm \eta_0\right)},  \\[0.5em]
P^{\eta_j}\left(\lambda; x, t\right), & \mathrm{for}\,\,\,\lambda\in B\left(\eta_j\right),\,\, j=0,1,2,  \\[0.5em]
\sigma_2P^{\eta_j}\left(-\lambda; x, t\right)\sigma_2, & \mathrm{for}\,\,\,\lambda\in B\left(-\eta_j\right), \,\, j=0,1,2.
\end{cases}
\end{gather}
The outer parametrix
$P^\infty(\lambda; x, t)=\left(P_{i, j}^\infty(\lambda; x, t)\right)_{2\times 2}$ is formulated as
\begin{gather}
\begin{aligned}
P^\infty_{1, 1}(\lambda; x, t)&=\frac{\delta_1+\delta_1^{-1}}{2}
\frac{\vartheta_3\left(w_1+\frac{1}{4}+\frac{\Delta_1}{2\pi\mathrm{i}}; \tau_1\right)}{\vartheta_3\left(w_1+\frac{1}{4}; \tau_1\right)}
\frac{\vartheta_3\left(0; \tau_1\right)}{\vartheta_3\left(\frac{\Delta_1}{2\pi\mathrm{i}}; \tau_1\right)}, \\[0.5em]
P^\infty_{1, 2}(\lambda; x, t)&=\frac{\delta_1-\delta_1^{-1}}{2}
\frac{\vartheta_3\left(w_1-\frac{1}{4}-\frac{\Delta_1}{2\pi\mathrm{i}}; \tau_1\right)}{\vartheta_3\left(w_1-\frac{1}{4}; \tau_1\right)}
\frac{\vartheta_3\left(0; \tau_1\right)}{\vartheta_3\left(\frac{\Delta_1}{2\pi\mathrm{i}}; \tau_1\right)},\\[0.5em]
P^\infty_{2, 1}(\lambda; x, t)&=\frac{\delta_1-\delta_1^{-1}}{2}
\frac{\vartheta_3\left(w_1-\frac{1}{4}+\frac{\Delta_1}{2\pi\mathrm{i}}; \tau_1\right)}{\vartheta_3\left(w_1-\frac{1}{4}; \tau_1\right)}
\frac{\vartheta_3\left(0; \tau_1\right)}{\vartheta_3\left(\frac{\Delta_1}{2\pi\mathrm{i}}; \tau_1\right)}, \\[0.5em]
P^\infty_{2, 2}(\lambda; x, t)&=\frac{\delta_1+\delta_1^{-1}}{2}
\frac{\vartheta_3\left(w_1+\frac{1}{4}-\frac{\Delta_1}{2\pi\mathrm{i}}; \tau_1\right)}{\vartheta_3\left(w_1+\frac{1}{4}; \tau_1\right)}
\frac{\vartheta_3\left(0; \tau_1\right)}{\vartheta_3\left(\frac{\Delta_1}{2\pi\mathrm{i}}; \tau_1\right)}.
\end{aligned}
\end{gather}
For $\lambda\in\left(\mathrm{D}_{1}^\pm\cup \mathrm{D}_\mathrm{o}^\pm\right)\cap B(\eta_1)$, the local parametrix is $P^{\eta_1}(\lambda; x, t)$ is
\begin{gather}
P^{\eta_1}(\lambda; x, t)=P^\infty(\lambda; x, t)A_{\pm}^{\eta_1}C\zeta_{\eta_1}^{-\sigma_3/4}M^{\mathrm{mB}}(\zeta_{\eta_1}; \beta_1)\mathrm{e}^{-\sqrt{\zeta_{\eta_1}}\sigma_3}\left(A_{\pm}^{\eta_1}\right)^{-1};
\end{gather}
where  $\zeta_{\eta_1}=t^2\mathcal{F}^{\eta_1}$, $A_{\pm}^{\eta_1}=\left(\mathrm{e}^{\pi\mathrm{i}/4\mp t\hat{\Omega}_1 /2}/fd^{\eta_1}\right)^{\sigma_3}\sigma_1$,
$\mathcal{F}^{\eta_1}=\left(p\pm\hat{\Omega}_1/2\right)^2$ for $\lambda\in B(\eta_1)\cap\mathbb{C}^\pm$, and $\hat{\Omega}_1=\Omega_1\left(\xi+1/\eta_1\eta_2\right)/4$.
The local parametrix $P^{\eta_1}(\lambda; x, t)$ is still defined in \eqref{parametrix-eta2}, and the local parametrix
 $P^{\eta_0}(\lambda; x, t)$ in (\ref{eta0-1}--\ref{eta0-2}).
 The Riemann-Hilbert problem for the error matrix $E(\lambda; x, t)$ has the same jump contour as that described in Figure \ref{ErrorE0}.

 \begin{RH} The error matrix $E(\lambda; x, t)$ satisfies the following RH problem:
 \begin{itemize}

 \item{} $E(\lambda; x, t)$ is analytic in $\lambda$ for $\lambda\in\mathbb{C}\setminus\mathcal{C}_1^c\cup\mathcal{C}_2^c\cup\partial B(\pm\eta_1, \pm\eta_2, \pm\eta_0)$.

     \item{} It normalizes to the identity matrix $\mathbb{I}_2$.

     \item{} The jump condition is described in \eqref{E-Jump} with the jump matrix
\begin{gather}
V^E=
\begin{cases}
P^\infty (\lambda)\mathcal{U}^{tp}_{f}\left[-ir^{-1}\right] \left(P^\infty (\lambda)\right)^{-1}, &\mathrm{for} \,\,\, \lambda\in\mathcal{C}^c_1,  \\[0.5em]
P^\infty (\lambda)\mathcal{L}^{tp}_{f}\left[-ir^{-1}\right] \left(P^\infty (\lambda)\right)^{-1}, &\mathrm{for} \,\,\, \lambda\in\mathcal{C}^c_2,  \\[0.5em]
P^{\eta_j}(\lambda)\left(P^\infty(\lambda)\right)^{-1}, &\mathrm{for} \,\,\, \lambda\in\partial B\left(\eta_j\right),\,\, j=0,1,2,  \\[0.5em]
\sigma_2P^{\eta_j}(-\lambda)\left(P^\infty\left(-\lambda\right)\right)^{-1}\sigma_2, &\mathrm{for} \,\,\, \lambda\in\partial B\left(-\eta_j\right),\,\, j=0,1,2.
\end{cases}
\end{gather}
\end{itemize}
\end{RH}

\begin{proposition}[Small norm estimate in the region $\xi<\xi_\mathrm{crit}$]
For $\beta_2, \beta_1, \beta_0>-1$,
the jump matrices $V^E$ has the following small norm estimates:
\begin{gather}
\begin{aligned}
&\left\|V^E-\mathbb{I}_2\right\|_{L^1\cap L^2\cap L^\infty \left(\mathcal{C}_1^c\cup\mathcal{C}_2^c\right)}=\mathcal{O}\left(\mathrm{e}^{-\mu t}\right), && \mathrm{as}\,\,\, t\to +\infty,  \\[0.5em]
&\left\|V^E-\mathbb{I}_2\right\|_{L^1\cap L^2\cap L^\infty \left(B\left(\pm\eta_2, \pm\eta_0, \pm\eta_1 \right)\right)}=\mathcal{O}\left(t^{-1}\right), && \mathrm{as}\,\,\, t\to +\infty,
\end{aligned}
\end{gather}
which leads to
\begin{gather}\label{E-Estimate-left}
E(0; x, t)=\mathbb{I}_2+\mathcal{O}\left(t^{-1}\right), \quad E^{[1]}(x, t)=\mathcal{O}\left(t^{-1}\right), \quad \mathrm{as}\,\,\, t\to +\infty.
\end{gather}
\end{proposition}
\begin{proof}
It follows from Proposition  \ref{p-sign-1} and the symmetry relation
\begin{gather}
p(\lambda; \xi)=p(\lambda^*; \xi)^*=-p(-\lambda; \xi), \quad \mathrm{for}\,\,\, \lambda\in\mathbb{C}\setminus\left[-\eta_2, \eta_2\right]
\end{gather}
that $\Re(p)<0$ for $\lambda\in \mathcal{C}_1\setminus \left\{\eta_1, \eta_0, \eta_2\right\}$,  and
 $\Re(p)>0$ for $\lambda\in \mathcal{C}_2\setminus \left\{-\eta_1, -\alpha, -\eta_2\right\}$,
which give rise to the estimate of the jump matrix on $\mathcal{C}_1^c\cup \mathcal{C}_2^c$.
The estimate on $(B\left(\pm\eta_2, \pm\eta_0, \pm \eta_1\right)$ is derived by
\begin{gather}
\begin{aligned}
&P^{\eta_j}(\lambda)\left(P^\infty(\lambda)\right)^{-1}=\mathbb{I}_2+\mathcal{O}\left(t^{-1}\right), && \mathrm{for}\,\,\, \lambda\in \partial B(\eta_j),\,\, j=0,1,2.
\end{aligned}
\end{gather}
Then we complete the proof.
\end{proof}

Recalling the these applied transforms $Y(\lambda; x, t)\mapsto T(\lambda; x, t)\mapsto S(\lambda; x, t)\mapsto E(\lambda; x, t)$ and there propositions of small norm estimates,
we obtain that
\begin{gather}
\begin{aligned}
&\frac{\partial \,u(x, t)}{\partial x}=4 \lim_{\lambda\to\infty} \lambda P^\infty_{1, 2}(\lambda; x, t)+\mathcal{O}\left(t^{-1}\right), \\[0.5em]
&\cos u(x, t)=1-2\left(P^\infty_{\pm1, 2}(0; x, t)\mathrm{e}^{tg_{\pm}\left(0\right)}f_{\pm}^{-1}\left(0\right)\right)^2+\mathcal{O}\left(t^{-1}\right), \\[0.5em]
&\sin u(x, t)=-2P^\infty_{\pm 1, 1}(0; x, t)P^\infty_{\pm 1, 2}(0; x, t)+\mathcal{O}\left(t^{-1}\right).
\end{aligned}
\end{gather}
A straightforward calculation leads to \eqref{large-middle} and \eqref{large-left}. This completes the proof of Theorem 2.

\section{Conclusions and discussions}

In summary, this study examines the asymptotic behavior of sine-Gordon kink-soliton gases, characterized by two distinct types of generalized reflection coefficients within the framework of Riemann-Hilbert problems. Specifically, we consider two reflection coefficients: \( r_0 = (\lambda - \eta_1)^{\beta_1}(\eta_2 - \lambda)^{\beta_2}|\lambda - \eta_0|^{\beta_0}\gamma(\lambda) \) and \( r_c = (\lambda - \eta_1)^{\beta_1}(\eta_2 - \lambda)^{\beta_2}\chi_c(\lambda) \).

A primary challenge in this analysis is the construction of a suitable \( g \)-function. Since the sine-Gordon equation represents a negative flow in the AKNS hierarchy, it differs significantly in its phase function, \( \theta \), from the Korteweg-de Vries equation. Thus, constructing an appropriate \( g \)-function is key. In this work, we address this by defining a piecewise \( g \)-function specifically tailored to the sine-Gordon context.

Another significant challenge lies in constructing the local parametrices near the endpoints \( \eta_1 \) and \( \eta_2 \), as well as at the singularity \( \eta_0 \). For the endpoints \( \eta_j \) (with \( j=1, 2 \)), we use modified Bessel functions of the first and second kinds, indexed by \( \beta_j \), to construct the corresponding local parametrix \( P^{\eta_j} \). For the reflection coefficient \( r_0 \), the local parametrix \( P^{\eta_0} \) is constructed using modified Bessel functions with indices \( (\beta_0 \pm 1)/2 \). In contrast, for \( r_c \), we utilize confluent hypergeometric functions to form \( P^{\eta_0} \).

Despite these developments, several open questions remain. As noted in \cite{15}, deriving rigorous asymptotics for soliton gases in the presence of multiple nontrivial reflection coefficients is still a challenging problem. Additionally, handling the limit process as discrete spectra accumulate in separate components along the imaginary axis poses significant difficulties. The Camassa-Holm equations, which are fundamental integrable models supporting multi-soliton and multi-peakon solutions, also provide a rich area for further research. Investigating the asymptotics of soliton and peakon gases within this framework using the Deift-Zhou method offers an interesting avenue for future exploration. For more on these open problems, readers can refer to the review in \cite{22}, which highlights key theoretical and experimental challenges in this field.

\section*{Acknowledgements}

GZ was supported by the National Natural Science Foundation of China (Grant No. 12201615).
ZY was supported by the National Natural Science Foundation of China (Grant No. 11925108).













\end{document}